\documentclass[12pt]{article}

\usepackage[margin=1in]{geometry}
\usepackage{amsmath,amssymb,amsthm}
\usepackage{mathtools}
\usepackage{booktabs}
\usepackage{array}
\usepackage{tabularx}
\usepackage{setspace}
\usepackage{parskip}
\usepackage{bm}
\usepackage{multirow}
\usepackage{graphicx}
\usepackage{natbib}
\usepackage{hyperref}
\usepackage{cleveref}
\usepackage{enumitem}
\usepackage{caption}
\usepackage{placeins}

\newtheorem{assumption}{Assumption}
\newtheorem{theorem}{Theorem}
\newtheorem{remark}{Remark}
\newtheorem{proposition}{Proposition}
\newtheorem{lemma}{Lemma}
\newtheorem{corollary}{Corollary}

\newcommand{\Oijt}{O_{ijt}}
\newcommand{\Wijt}{W_{ijt}}
\newcommand{\Ejt}{E_{jt}}
\newcommand{\Ajt}{A_{jt}}
\newcommand{\Sijt}{S_{ijt}}
\newcommand{\Dijt}{\Delta_{ijt}}
\newcommand{\Yijt}{Y_{ijt}}

\newcommand{\gEhat}{\hat{g}_E}
\newcommand{\gDhat}{\hat{g}_\Delta}

\newcommand{\QY}{\bar{Q}_Y}
\newcommand{\Qint}{\bar{Q}_{\mathrm{int}}}
\newcommand{\QYhat}{\hat{\bar{Q}}_Y}
\newcommand{\Qinthat}{\hat{\bar{Q}}_{\mathrm{int}}}
\newcommand{\Qdag}{\bar{Q}^{\dagger}_{\mathrm{int}}}
\newcommand{\Te}{\mathbb{T}_e}
\newcommand{\Tov}{\mathbb{T}_{\mathrm{ov}}}
\newcommand{\inde}{\mathbf{1}\{t \in \Te\}}
\newcommand{\EICj}{\mathrm{EIF}_j}

\newcommand{\phij}{\varphi_j}
\newcommand{\PP}{\mathbb{P}}
\newcommand{\EE}{\mathbb{E}}
\newcommand{\oP}{o_P}

\newcommand{\Var}{\operatorname{Var}}
\newcommand{\Cov}{\operatorname{Cov}}
\newcommand{\logit}{\operatorname{logit}}
\newcommand{\calT}{\mathcal{T}}
\newcommand{\emax}{e_{\max}}
\newcommand{\PsiIT}{\Psi_{\mathrm{IT}}}
\newcommand{\bW}{\mathbf W}

\title{Surrogate-Assisted Estimation of Exposure-Time Effects\\
in Stepped-Wedge Trials with Administratively Censored Outcomes}
\author{
  Lin Li\thanks{Department of Biostatistics, University of California San Diego.} \and
  Tuo Lin\thanks{Department of Biostatistics, University of Florida.} \and
  Hao Wu \and
  Yiwen Chen \and
  Xin M.\ Tu\footnotemark[1]
}
\date{October 2026}

\begin{document}
\maketitle

\begin{abstract}
In stepped-wedge trials, effects often change with the time since a cluster
crossed over, and a delayed primary outcome is administratively censored
for clusters that cross over late. The design couples them: long
exposure times occur only in the last periods, where censoring is heaviest. We
study exposure-time-specific effects when a short-term surrogate is always
recorded and censoring is ignorable given the surrogate, exposure time,
covariates and calendar time. The usual immediate-treatment estimand
gives more weight to short exposures. Each exposure-time effect is identified
through a surrogate bridge, and because the rollout fixes the exposure law,
neither the surrogate distribution nor the bridge regression enters the
remainder of the efficient influence function. The efficiency bound is a
full-data part plus a censoring burden whose share grows with exposure time;
as a function of the allocation of clusters to waves it becomes a design
criterion, which can favor the heavily censored last wave. We propose a
cross-fitted surrogate-assisted augmented estimator, prove joint asymptotic
normality, show that with a consistent bridge its variance is the same for
independent and fixed-wave crossover times, and justify a fixed-partition cluster
jackknife and a heterogeneity test. Treating each crossover as a sequential
randomization gives a sharper bound that uses a cluster's pre-crossover
history; under a random-intercept working structure the optimal correction
is a one-dimensional index of the cluster's own surrogate and outcome
history, and the resulting estimator attains the bound with 30 clusters in
simulations, where its standard deviations are within 7\% of those of a
correctly specified exposure-time mixed model for the first three of five
exposure times. Simulations, one calibrated to the design of an expedited
partner therapy trial, illustrate the methods.
\end{abstract}

{\raggedright\noindent\textbf{Keywords:} Administrative censoring; Cluster
jackknife; Efficient influence function; Exposure time; Stepped-wedge design;
Surrogate endpoint.\par}

\newpage
\section{Introduction}
\label{sec:intro}

In a stepped-wedge cluster randomized trial (SW-CRT), all clusters begin in
the control condition and cross over to the intervention at randomly assigned
times. The design that motivated this work is the Washington State expedited
partner therapy (EPT) trial \citep{golden2015uptake}, which randomized $J=23$
local health jurisdictions (LHJs) to four crossover waves of a program
promoting patient-delivered partner therapy. Its primary outcomes were
measured at the community level, as chlamydia positivity among women tested
at sentinel clinics in the last three months of each period and as reported
gonorrhea incidence. The setting we study is a delayed-outcome version of
this design: each person diagnosed with chlamydia has a short-term surrogate,
EPT uptake within 30 days, and a delayed outcome, repeat positivity 12 months
later, and the analysis takes place at a fixed calendar date.
Section~\ref{sec:sim:ept} builds this setting from a synthetic replication of
the trial's design, which reproduces its rollout and the sample sizes assumed
in its planning but not its records or its outcomes.

Trials of this kind raise two difficulties. Programs such as EPT take time to
be adopted, so their effects may grow or fade with the time since crossover,
and analyses that assume an immediate, constant effect can be misleading
\citep{hughes2015current,kenny2022analysis}. This has motivated
exposure-time-specific estimands, tests for exposure-time heterogeneity and
model-robust estimators
\citep{kenny2022analysis,maleyeff2023assessing,wang2024modelrobust,chen2024modelassisted}.
In addition, a delayed outcome is administratively censored when follow-up
ends at a fixed calendar date, and the censoring falls on the clusters that
crossed over late; in the synthetic data underlying
Section~\ref{sec:sim:ept}, 85\% of the last wave's records lack the 12-month
outcome, while the surrogate is recorded for everyone.

The design links the two difficulties. A cluster can have been exposed for
$e$ periods only at calendar time $t\ge e+1$, and the contrast at exposure
time $e$ also requires clusters that are still unexposed, so long exposure
times are observed only in the final periods, where censoring is heaviest.
Inverse-probability-of-censoring weighting (IPCW) then tends to be unstable, a
complete-case analysis ignores the surrogate and may be biased if censoring
depends on it, and much of the information about the later exposure times has
to come from the surrogate.

Assuming that censoring is independent of the delayed outcome given the
surrogate, exposure time, baseline covariates and calendar time, we show that
the efficiency bound for each exposure-time effect is a full-data part plus a
\emph{censoring burden}, the residual variance of the outcome given the
surrogate weighted by the odds of censoring. The decomposition is standard in
missing-data theory. What is specific to the stepped-wedge design is that,
under calendar-driven censoring, the burden accounts for a growing share of
the bound as exposure time increases (Theorem~\ref{thm:burden}). Under a
simple cluster-level working model the bound is also an explicit function of
the allocation of clusters to waves, which gives a design criterion
(Proposition~\ref{prop:design}). For complete outcomes and a constant effect,
the first and last waves are known to warrant more clusters
\citep{lawrie2015optimal}. The criterion suggests that this remains true for
the last wave even when it is heavily censored, because its clusters provide
the only unexposed comparison for the longest exposure time; in a simulation
calibrated to the EPT trial, the selected allocation reduced the variance of
the longest-exposure effect by about half.

Each exposure-time effect is identified by a \emph{surrogate bridge}, the
complete-case outcome regression integrated over the exposure-specific
surrogate distribution, which we study using the geometric view of
semiparametric efficiency theory
\citep{bickel1993efficient,tsiatis2006semiparametric,nabi2026semiparametric}.
The record-level efficient influence function (EIF) has the form derived by
\citet{kallus2025surrogates} for outcomes missing at random given surrogates.
Because the randomized rollout fixes the law of exposure time given calendar
time, errors in the surrogate law and in the bridge regression do not
contribute to the remainder (Proposition~\ref{prop:dml_remainder}), so the
surrogate distribution need not be estimated. We also note that the
immediate-treatment estimand is a design-weighted average of exposure-time
effects with weights that decrease with exposure time
(Proposition~\ref{prop:it}).

A common alternative is a mixed model with exposure-time indicators
\citep{kenny2022analysis}. Its coefficients correspond to the exposure-time
effects when period and exposure effects are additive and do not change with
calendar time; with a logit link they are cluster-specific rather than
population-averaged; and robust standard errors address a misspecified
correlation structure but not a misspecified effect structure
\citep{wang2024modelrobust}. The estimator we propose, a cross-fitted
surrogate-assisted augmented inverse probability weighted estimator
(SA-AIPW), targets effects defined without an outcome model, remains
consistent if either the outcome regression or the observation model is
correct, relies on cluster-level variance estimators, and uses the surrogate
to recover information lost to censoring. The mixed model also uses a
cluster's own past, and a design-based estimator can do the same: treating
each crossover as a sequential randomization with known probabilities gives
the efficiency bound of the model that uses only this randomization and the
absence of anticipation (Proposition~\ref{prop:history}). Under a
random-intercept working structure the optimal correction is a
one-dimensional index built from the cluster's own surrogate and outcome
history, with one coefficient per target fitted at the cluster level, and
the estimator built on it attains the bound with 30 clusters in our
simulations. In the Gaussian design of Section~\ref{sec:sim} the bound can be
evaluated exactly, and so can the variance of the mixed model, so we can say
what the model adds beyond the design: with 30 clusters its standard
deviations are within 7\% of ours for the first three exposure times, 15\%
for the fourth and about 50\% for the longest, and the latter gap splits
about equally between its assumption that effects do not change with
calendar time and its use of the final period, in which every cluster is
exposed. When effects changed with calendar time and censoring depended
strongly on the surrogate, the mixed model's coverage fell to 0.86--0.88,
while that of SA-AIPW stayed between 0.94 and 0.96.

SA-AIPW (Section~\ref{sec:tmle}) integrates over the empirical surrogate
distribution rather than a density estimate. It is jointly asymptotically
normal and, with a consistent bridge, efficient in a model that restricts
only the rollout (Theorem~\ref{thm:asymp}); its variance is then the same
whether crossover times are independent or fixed in waves
(Proposition~\ref{prop:alloc}). With few clusters, a leave-one-cluster-out
jackknife with a fixed cross-fitting partition is consistent
(Theorem~\ref{thm:jack} and Proposition~\ref{prop:jack_valid}) and calibrates
a test of exposure-time heterogeneity (Corollary~\ref{cor:het}), and for rare
binary outcomes we use logit-scale intervals for the arm means. Proofs are in
the Supplementary Material.

\citet{kenny2022analysis} showed that immediate-treatment mixed models can be
biased when effects vary with exposure time, \citet{maleyeff2023assessing}
modelled exposure-time heterogeneity with random effects,
\citet{wang2024modelrobust} studied when model-based estimators are
consistent for calendar- and exposure-time-specific estimands, and
\citet{chen2024modelassisted} developed design-based model-assisted
estimators. \citet{roth2023efficient} treat the timing of treatment in a
staggered rollout as randomized and show that adjusting the difference in
means for pre-treatment outcomes, with a coefficient chosen to minimize
variance within that class, is efficient in the class; they recommend a
low-dimensional adjustment in practice. \citet{granston2013addressing}
studied lagged effects and interval-censored outcomes in stepped-wedge trials
with parametric models. To our knowledge, none of these works uses a
surrogate or addresses administrative censoring of a delayed outcome, and
the efficiency results of this paper are semiparametric rather than within
a class of estimators. We build on augmented inverse probability weighting
\citep{robins1994estimation,robins1995semiparametric,bang2005doubly},
sequential double robustness \citep{luedtke2017sequential} and
\citet{kallus2025surrogates}; the surrogate index of
\citet{athey2025surrogate} and work on evaluating surrogates
\citep{frangakis2002principal,gilbert2008evaluating} address different
questions. For inference we adapt tools from cluster-level targeted learning
\citep{balzer2015adaptive,balzer2023two}, covariate adjustment in randomized
trials \citep{tsiatis2008covariate} and finite-population central limit
theory \citep{li2017general}.

\section{Design, Estimands and Identification}
\label{sec:data_id}

\subsection{Design and Data}
\label{sec:design}

Consider a SW-CRT with $J$ clusters and $T$ calendar periods. Cluster $j$
contributes $n_j$ individuals in each period, so it has $M_j=n_jT$ records. At
baseline each cluster is randomized to a crossover period
$\tau_j\in\{2,\ldots,T\}$ from a prespecified allocation with probabilities
$p_\tau(s)=\PP(\tau_j=s)$. The exposure time of cluster $j$ at period $t$ is
\[
  \Ejt=(t-\tau_j+1)\,I(t\ge\tau_j)\in\{0,1,\ldots,T-1\},
\]
the number of periods since crossover, counting the crossover period as one,
and $\Ajt=I(\Ejt\ge1)$ is the usual treatment indicator. Because $\tau_j$ is
randomized, the law of exposure time given calendar time is fixed by the
design:
\begin{equation}
  \label{eq:gE}
  g_E(e\mid t)=\PP(\Ejt=e\mid t)=
  \begin{cases}
    p_\tau(t-e+1), & 1\le e\le t-1,\\
    \sum_{s>t}p_\tau(s), & e=0.
  \end{cases}
\end{equation}
Under uniform allocation, $g_E(e\mid t)=1/(T-1)$ for $1\le e\le t-1$. Crossover
times may be drawn independently or, as is more common, allocated in fixed
waves and randomly permuted across clusters; Section~\ref{sec:alloc} shows that
this choice does not affect the asymptotic variance of the proposed estimator
when its bridge regression is consistent.
The observed record for individual $i$ in cluster $j$ at period $t$ is
\begin{equation}
  \label{eq:data}
  \Oijt=\bigl(\Wijt,\;t,\;\Ejt,\;\Sijt,\;\Dijt,\;\Dijt\Yijt\bigr),
\end{equation}
with baseline covariates $\Wijt$, short-term surrogate $\Sijt$, an indicator
$\Dijt$ that the delayed outcome $\Yijt$ was observed by the administrative
end date, and the outcome itself when it was.

The data are generated by a structural model, given in Supplement~S2.1 and
drawn in Figure~S1, in the causal order $W\to E\to S\to\Delta\to Y$. The
crossover time $\tau_j$ and a cluster effect $b_j$ are drawn independently
across clusters; the surrogate depends on $(W,E,t)$, $b_j$ and individual
noise, the observation indicator on $(S,E,W,t,\tau_j)$ and its own noise, and
the outcome on $(W,E,S,t)$, $b_j$ and its own noise. The law~\eqref{eq:gE} is
known because the crossover time is randomized, as $\PP(A=1)=1/2$ is known in
a parallel-arm trial; exposure is not degenerate, since $0<g_E(e\mid t)<1$ in
the overlap periods defined below. The cluster effect is the source of the
within-cluster correlation. The outcome does not enter the observation
equation and the observation noise does not enter the outcome equation, which
is the structural form of Assumption~\ref{ass:mar} below. Administrative
censoring may depend on the crossover time, through the follow-up available
after crossover, but because $\tau_j$ is independent of the cluster effect,
the covariates and the surrogate and outcome noise given $(\Ejt,t)$, this does
not affect Assumption~\ref{ass:mar}. Finally, the surrogate and outcome depend
on the treatment history only through the current exposure time, which rules
out anticipation, carryover beyond exposure time and interference
\citep{kenny2022analysis,wang2024modelrobust}.

\subsection{Exposure-Time Estimands}
\label{sec:scm:po}

For $e\in\{0,1,\ldots,T-1\}$ let $S_{ijt}(e)$ and $Y_{ijt}(e)$ be the
potential surrogate and outcome under exposure time $e$
\citep{kenny2022analysis,wang2024modelrobust}. A contrast between exposure
times $e$ and $0$ at period $t$ is supported by the design when both occur
with positive probability, so we define the \emph{overlap periods}
\begin{equation}
  \label{eq:Te}
  \Te=\{t:\;g_E(e\mid t)>0\text{ and }g_E(0\mid t)>0\},
\end{equation}
which equal $\{e+1,\ldots,T-1\}$ in a complete stepped wedge and are nonempty
for $1\le e\le\emax=T-2$. The exposure-time effect is
\begin{equation}
  \label{eq:psi_def}
  \Psi(e) \;=\; \EE\bigl\{Y_{ijt}(e)-Y_{ijt}(0)\,\big|\,t\in\Te\bigr\},
  \qquad e=1,\ldots,\emax,
\end{equation}
an average over records in the overlap periods; we write
$\Tov=\bigcup_e\Te$ for the periods in which some contrast is supported.
Summaries take the form
$\Psi_\omega=\sum_e\omega_e\Psi(e)$ for prespecified weights; we use the
\emph{exposure-averaged effect} $\bar\Psi=\emax^{-1}\sum_e\Psi(e)$. Long
exposure times have the shortest and latest overlap windows. In a design with
$T=7$, $\Psi(1)$ averages over periods 2 to 6 but $\Psi(5)$ is supported only
by period~6. Effects that also vary with calendar time can be handled by
replacing~\eqref{eq:psi_def} with period-specific contrasts; we do not pursue
this here.

\subsection{Identification}
\label{sec:id}

\begin{assumption}[Consistency]
\label{ass:consistency}
If $\Ejt=e$, then $\Sijt=\Sijt(e)$ and $\Yijt=\Yijt(e)$.
\end{assumption}

\begin{assumption}[Randomized rollout]
\label{ass:seqrand}
$\{Y(e),S(e)\}_{e\ge0}\perp\Ejt\mid\Wijt,t$. This holds by design because
$\tau_j$ is randomized independently of potential outcomes and covariates.
\end{assumption}

\begin{assumption}[Surrogate-mediated missingness at random]
\label{ass:mar}
\[
  \PP(\Dijt=1\mid\Yijt,\Sijt,\Ejt,\Wijt,t)=\PP(\Dijt=1\mid\Sijt,\Ejt,\Wijt,t),
\]
that is, the delayed outcome is missing at random given the surrogate
\citep{rubin1976inference}.
\end{assumption}

Assumption~\ref{ass:mar} is most defensible when censoring is driven by the
calendar, through the follow-up available before the analysis date, rather
than by health status. A useful surrogate lies on the path from the
intervention to the outcome, is recorded before the censoring date, and has a
relationship with the delayed outcome that is stable across waves.

\begin{assumption}[Observation positivity]
\label{ass:pos_obs}
$g_\Delta(1\mid s,e,w,t)=\PP(\Dijt=1\mid\Sijt=s,\Ejt=e,\Wijt=w,t)\ge c>0$ for
$e\in\{1,\ldots,\emax\}$ and $t\in\Te$, and for $e=0$ and $t\in\Tov$.
\end{assumption}

\begin{assumption}[Exposure overlap]
\label{ass:pos_trt}
There is $c_E>0$ such that $g_E(e\mid t)\ge c_E$ and $g_E(0\mid t)\ge c_E$
for all $t\in\Te$, $e=1,\ldots,\emax$.
\end{assumption}

Write $\QY(s,e,w,t)=\EE(Y\mid S=s,E=e,W=w,t,\Delta=1)$ for the complete-case
outcome regression, $f_S(s\mid e,w,t)$ for the conditional density of the
surrogate with respect to a dominating measure $\mu$, and
\begin{equation}
  \label{eq:Qint_def}
  \Qint(e,w,t)\;=\;\int \QY(s,e,w,t)\,f_S(s\mid e,w,t)\,d\mu(s)
  \;=\;\EE\{\QY(S,e,W,t)\mid E=e,W=w,t\}
\end{equation}
for the \emph{bridge regression}.

\begin{theorem}[Identification through the surrogate bridge]
\label{thm:id}
Under Assumptions~\ref{ass:consistency}--\ref{ass:pos_trt}, for
$e=1,\ldots,\emax$,
\begin{equation}
  \label{eq:identification}
  \Psi(e) \;=\; \EE\bigl\{\Qint(e,W,t)-\Qint(0,W,t)\,\big|\,t\in\Te\bigr\}.
\end{equation}
\end{theorem}

The proof (Supplement~S2.2) follows the usual argument: randomization
equates the law of $\{S(k),Y(k)\}$ given $(W,t)$ with that of $(S,Y)$ given
$(E=k,W,t)$, and Assumptions~\ref{ass:mar}--\ref{ass:pos_obs} replace
$\EE(Y\mid S,E=k,W,t)$ with its complete-case version. The inner integral
in~\eqref{eq:Qint_def} transports the complete-case regression to all units
through the surrogate; the observation probability $g_\Delta$ appears nowhere
in~\eqref{eq:identification}.

\subsection{The Immediate-Treatment Estimand}
\label{sec:it}

Many analyses contrast treated and untreated records through the binary
indicator $\Ajt$. Define the immediate-treatment
estimand
\[
  \PsiIT=\EE\bigl\{\EE(Y\mid A=1,W,t)-\EE(Y\mid A=0,W,t)\,\big|\,t\in\Tov\bigr\},
\]
the overlap-period contrast that a binary-treatment analysis targets. Under
Assumptions~\ref{ass:consistency}--\ref{ass:pos_trt},
$\EE(Y\mid A=1,W,t)=\sum_{e\ge1}w_t(e)\Qint(e,W,t)$ with the weights below,
so $\PsiIT$ is identified through the same bridge.

\begin{proposition}[The immediate-treatment estimand]
\label{prop:it}
Under Assumptions~\ref{ass:consistency}--\ref{ass:pos_trt}, let
$\delta_t(e)=\EE\{Y(e)-Y(0)\mid t\}$ and
$w_t(e)=g_E(e\mid t)/\sum_{k\ge1}g_E(k\mid t)$. Then
\begin{equation}
  \label{eq:psi_it}
  \PsiIT=\EE\Bigl\{\sum_{e=1}^{t-1}w_t(e)\,\delta_t(e)\,\Big|\,t\in\Tov\Bigr\}.
\end{equation}
If $\delta_t(e)=\delta(e)$ does not depend on calendar time, then
$\Psi(e)=\delta(e)$ and $\PsiIT=\sum_e\omega^{\mathrm{IT}}_e\Psi(e)$ with
$\omega^{\mathrm{IT}}_e=\EE\{w_t(e)I(t\ge e+1)\mid t\in\Tov\}$. With equal
numbers of records per period, $\omega^{\mathrm{IT}}_e$ is nonincreasing in $e$
for every allocation, and under uniform allocation
$\omega^{\mathrm{IT}}_e=(T-2)^{-1}\sum_{t=e+1}^{T-1}(t-1)^{-1}$.
\end{proposition}

So the immediate-treatment estimand is a well-defined causal summary, but
its weights are set by the rollout schedule and favour short exposures.
With $T=7$ and uniform allocation, $\omega^{\mathrm{IT}}=(0.46,0.26,0.16,0.09,
0.04)$ for $e=1,\ldots,5$; in the EPT design the weights are
$(0.61,0.28,0.11)$ for $e=1,2,3$. When effects grow with exposure, $\PsiIT$
understates the exposure-averaged effect $\bar\Psi$, and trials with
different schedules may estimate different quantities. Supplement~S2.2 gives
the proof and the EIF of $\PsiIT$; our implementation uses a valid but
inefficient estimating function for $\PsiIT$, so the efficiency and
allocation results below apply to $\PsiIT$ only in its efficient form.

\section{Semiparametric Geometry of Exposure-Time Effects}
\label{sec:semiparam}

We now study $\Psi(e)$ as a functional of the observed-data distribution,
which shows which nuisance errors an estimator must control and how much
information censoring removes. We use the language of \citet{nabi2026semiparametric}, summarized
in Table~S1 of the Supplementary Material; the underlying theory is in
\citet{bickel1993efficient}, \citet{tsiatis2006semiparametric} and
\citet[Chapter~25]{vaart1998asymptotic}. In
Sections~\ref{sec:geom}--\ref{sec:burden} we work with the law $P$ of a generic
record $O=(W,t,E,S,\Delta,\Delta Y)$, the record-weighted (pooled) law across
clusters; Lemma~\ref{lem:summation} lifts the results to clusters, which are
the independent units.

\subsection{Paths, Scores and the Tangent Space}
\label{sec:geom}

The tangent space $\calT\subseteq L^2_0(P_0)$ is the closed linear span of the
scores of regular one-dimensional paths through $P_0$ in the model, and the
EIF is the unique gradient of $\Psi$ in $\calT$; its variance is the
efficiency bound.

The observed-data density factorizes along the time ordering
$W\to E\to S\to\Delta\to Y$:
\begin{equation}
  \label{eq:factorization}
  p(o)=p(w,t)\,g_E(e\mid t)\,f_S(s\mid e,w,t)\,
  g_\Delta(\delta\mid s,e,w,t)\,q_Y(y\mid s,e,w,t)^{\delta},
\end{equation}
where $q_Y$ is the conditional density of $Y$ among records with
$\Delta=1$. The model $\mathcal M$ leaves every factor unrestricted except
$g_E$, which is fixed by the design. Scores of paths that perturb different
factors are orthogonal, so
\begin{equation}
  \label{eq:tangent}
  \calT=\calT_{W,t}\oplus\calT_S\oplus\calT_\Delta\oplus\calT_Y,
  \qquad \calT_E=\{0\},
\end{equation}
with, for example, $\calT_S=\{s(S,E,W,t):\EE(s\mid E,W,t)=0\}$ and
$\calT_\Delta=\{(\Delta-g_\Delta(1\mid S,E,W,t))\,u(S,E,W,t)\}$.

\paragraph{The known exposure law.} The restriction $\calT_E=\{0\}$ concerns
the model, not the data: exposure time still varies across clusters, but no
admissible path changes its law. The removed directions
$\{u(E,W,t):\EE(u\mid W,t)=0\}$ would be nuisance directions in the
nonparametric model, because $\Psi(e)$ does not depend on $g_E$, so the EIF
in $\mathcal M$ coincides with the EIF in the nonparametric model
\citep[Section~8.2]{nabi2026semiparametric}, so knowledge of the exposure law
does not by itself lower the efficiency bound (the restriction of no
anticipation, which does, is taken up in Section~\ref{sec:history}). Its
consequences concern estimation: the exposure law
contributes no estimation error (Proposition~\ref{prop:dml_remainder}), and
influence functions with a component along the removed directions are valid
but inefficient (Remark~\ref{rem:unaug}) and sensitive to whether crossover
times are drawn independently or fixed in waves
(Proposition~\ref{prop:alloc}).

\subsection{The Efficient Influence Function}
\label{sec:eic:decomp}

For nuisance functions $\eta=(\QY,\Qint,g_E,g_\Delta)$ and an exposure level
$k$, define the uncentered estimating function
\begin{equation}
  \label{eq:phi}
  \begin{aligned}
  \phi_k(O;\eta)
  = \Qint(k,W,t)
  &+ \frac{I(E=k)}{g_E(k\mid t)}\bigl\{\QY(S,k,W,t)-\Qint(k,W,t)\bigr\}\\
  &+ \frac{I(E=k)\,\Delta}{g_E(k\mid t)\,g_\Delta(1\mid S,k,W,t)}
    \bigl\{Y-\QY(S,k,W,t)\bigr\}.
  \end{aligned}
\end{equation}
Its three terms are the $\calT_{W,t}$, $\calT_S$ and $\calT_Y$ components of
the gradient of the mean $\EE\{\Qint(k,W,t)\}$: the weight $1/g_E$ transports
the records with exposure $k$ to the population using the known design law,
and $\Delta/g_\Delta$ transports the observed subcohort to all records while
acting only on the mean-zero residual $Y-\QY$.

\begin{proposition}[Efficient influence function]
\label{prop:eif}
Under Assumptions~\ref{ass:consistency}--\ref{ass:pos_trt}, with
$\pi_e=\PP(t\in\Te)$ and $\eta_0$ the true nuisances, the EIF of $\Psi(e)$ in
$\mathcal M$ is
\begin{equation}
  \label{eq:Dstar}
  D^*_e(O)\;=\;\frac{\inde}{\pi_e}
  \bigl\{\phi_e(O;\eta_0)-\phi_0(O;\eta_0)-\Psi(e)\bigr\}.
\end{equation}
It is orthogonal to the removed exposure directions, so it is also the EIF in
the nonparametric model.
\end{proposition}

Supplement~S2.3 derives each component as the derivative of $\Psi(e)$ along
paths that perturb a single likelihood factor. Since the $\calT_S$ component
is nonzero, $\Psi(e)$ moves when the surrogate law moves, and every estimator
must account for $f_S$; Section~\ref{sec:tmle} does so through an empirical
average rather than a density estimate. At the record level~\eqref{eq:Dstar}
coincides in form with the EIF of \citet{kallus2025surrogates} for outcomes
missing at random given surrogates, applied to a multi-valued treatment with a
known assignment law.

The EIF has no component along the censoring directions: $\Psi(e)$ depends
on $P$ only through $p(w,t)$, $f_S$ and $\QY$, so $D^*_e\perp\calT_\Delta$, and
$g_\Delta$ enters $D^*_e$ only through the weight $\Delta/g_\Delta$ on the
residual $Y-\QY$ (Lemma~S1 of the Supplementary Material).
Assumption~\ref{ass:mar} is used only for identification; it places no
restriction on the observed-data law \citep{gill1997coarsening}, which is why
$\calT$ retains $\calT_\Delta$.

\begin{lemma}[Cluster-level influence function: sum, then scale]
\label{lem:summation}
Let $O_j=\{O_{ijt}\}$ be i.i.d.\ across clusters, with $\bar M=\EE(M_j)$.
$\Psi(e)$ depends on the cluster-level law only through the record-weighted
law $\bar P(\cdot)=\bar M^{-1}\EE\{\sum_{i,t}I(O_{ijt}\in\cdot)\}$. Hence, for
every cluster-level path with score $s(O_j)$,
$d\Psi(e)/d\varepsilon|_{\varepsilon=0}=\EE\{s(O_j)\,\EICj(e)\}$, where
\begin{equation}
  \label{eq:EICj_lemma}
  \EICj(e) \;=\; \frac{1}{\bar M}\sum_{i=1}^{n_j}\sum_{t=1}^{T} D^*_e(O_{ijt}).
\end{equation}
\end{lemma}

The contribution of a cluster is the sum of its record-level contributions,
divided by a constant common to all clusters. The sum absorbs within-period
and across-period correlation; dividing instead by the cluster's own size
would target a cluster-weighted functional. Supplement~S2.5 gives the proof
and shows that $\EICj(e)$ is orthogonal to every cluster-level assignment
direction, so it is the canonical gradient in the cluster-level model that
restricts only the law of the crossover times, with $W$ as the baseline
information. The number of clusters $J$ is the effective sample size. That
model does not use the further restrictions that the absence of anticipation
places on a cluster's data before it crosses over, nor does it assume additive
period and cluster effects; estimators that exploit either can have smaller
variance. Section~\ref{sec:history} examines the first.

\subsection{The Remainder and the Surrogate Law}
\label{sec:rsy}

Estimators that average $\phi_k(\cdot;\hat\eta)$ over overlap periods, whether
one-step, estimating-equation or targeted, have error governed by the bias
$B_k(\hat\eta)=\EE\{\phi_k(O;\hat\eta)\mid t\in\Te\}-\EE\{\Qint(k,W,t)\mid
t\in\Te\}$, the von Mises remainder of the bridge functional. Let
$\|\cdot\|_k$ denote the $L^2$ norm under the law of $(S,W,t)$ given $E=k$
and $t\in\Te$.

\begin{proposition}[Remainder of the bridge estimating function]
\label{prop:dml_remainder}
Let $\hat\eta=(\QYhat,\Qinthat,\gEhat,\gDhat)$ with $\gEhat$ and $\gDhat$
bounded below by $c>0$ and $\Qinthat$ any bounded function of $(k,w,t)$. Then
$B_k(\hat\eta)=R_{E,k}(\hat\eta)+R_{\Delta,k}(\hat\eta)$, where
\begin{align}
  R_{E,k}
  &= \EE\Bigl[\Bigl\{1-\frac{g^0_E(k\mid t)}{\gEhat(k\mid t)}\Bigr\}
     \bigl\{\Qinthat(k,W,t)-\Qint^0(k,W,t)\bigr\}\,\Big|\,t\in\Te\Bigr],
  \label{eq:RA}\\
  R_{\Delta,k}
  &= \EE\Bigl[\frac{g^0_E(k\mid t)}{\gEhat(k\mid t)}\int
     \bigl(\QYhat-\QY^0\bigr)(s,k,W,t)
     \Bigl(1-\frac{g^0_\Delta}{\gDhat}\Bigr)(s,k,W,t)\notag\\
  &\hspace{12em}\times f^0_S(s\mid k,W,t)\,d\mu(s)\,\Big|\,t\in\Te\Bigr].
  \label{eq:RD}
\end{align}
\begin{enumerate}[label=\normalfont(\alph*),leftmargin=2em]
  \item \emph{(Location of the surrogate-law error.)} For the density
  plug-in $\Qinthat=\int\QYhat\,\hat f_S\,d\mu$,
  \begin{align}
    \Qinthat-\Qint^0
    &= \int(\QYhat-\QY^0)f^0_S\,d\mu+\int\QY^0(\hat f_S-f^0_S)\,d\mu
      + R_{SY}, \label{eq:RSY}\\
    R_{SY}&=\int(\QYhat-\QY^0)(\hat f_S-f^0_S)\,d\mu , \notag
  \end{align}
  and all three terms, including the mixed curvature $R_{SY}$, enter $B_k$
  only through $R_{E,k}$, multiplied by the error in the exposure law.
  \item \emph{(Design-known exposure law.)} If $\gEhat=g^0_E$, as in a
  stepped-wedge trial, then $R_{E,k}\equiv 0$ for every choice of $\Qinthat$,
  and $|B_k(\hat\eta)|\le C\,\|\QYhat-\QY^0\|_k\,\|\gDhat-g^0_\Delta\|_k$ for a
  constant $C$ depending only on $c$ and $c_E$. Neither $f_S$ nor $\Qint$
  enters the bias.
  \item \emph{(Estimated assignment.)} If the assignment law is estimated,
  $|R_{E,k}|\le c^{-1}\|\gEhat-g^0_E\|\,\|\Qinthat-\Qint^0\|$, so for a density
  plug-in a rate for $\hat f_S$ is the natural sufficient condition, in
  addition to the product rate for $(\QYhat,\gDhat)$. The calculation is
  unchanged when the assignment law depends on $W$.
\end{enumerate}
\end{proposition}

The plug-in error of the bridge contains the mixed second-order term
$R_{SY}$ because $\Qint$ is bilinear in $(\QY,f_S)$. The $\calT_S$ component of
$\phi_k$, evaluated at $\QYhat$ rather than $\QY^0$, removes the whole plug-in
error of $\Qinthat$, mixed term included, and leaves only its product with the
error in the weight $1/g_E$. Because the rollout fixes that weight, the
surrogate law need not be estimated. This is most useful for long exposure
times, where few records have $E=k$ and the bridge regression rests on little
data; errors there affect efficiency but not bias. Supplement~S2.6
gives the proof and a numerical check.

\subsection{The Censoring Burden}
\label{sec:burden}

The variance of the EIF is the smallest asymptotic variance attainable by a
regular estimator. Under administrative censoring it decomposes as follows.

\begin{theorem}[Censoring burden]
\label{thm:burden}
Let $\sigma^2_k(S,W,t)=\Var(Y\mid S,E=k,W,t)$ be the residual variance of the
delayed outcome given the surrogate, and let $V^{\mathrm{full}}(e)$ be the
variance of the EIF of $\Psi(e)$ in the same model with $Y$ always observed.
Under Assumptions~\ref{ass:consistency}--\ref{ass:pos_trt},
\begin{enumerate}[label=\normalfont(\roman*),leftmargin=2em]
  \item $\Var\{D^*_e(O)\}=V^{\mathrm{full}}(e)+\mathcal B(e)$, where
  \begin{equation}
    \label{eq:burden}
    \mathcal B(e)=\frac{1}{\pi_e}\,\EE\Bigl[\sum_{k\in\{0,e\}}
      \frac{1}{g_E(k\mid t)}\,
      \EE\Bigl\{\sigma^2_k(S,W,t)\,
      \frac{1-g_\Delta(1\mid S,k,W,t)}{g_\Delta(1\mid S,k,W,t)}
      \,\Big|\,E=k,W,t\Bigr\}\,\Big|\,t\in\Te\Bigr]
  \end{equation}
  is the censoring burden: the outcome variance left unexplained by the
  surrogate, weighted by the odds of censoring.
  \item \emph{(Coupling.)} Suppose allocation is uniform with equal numbers
  of records per period, $\Var(Y\mid E=k,W,t)=v$, the conditional effect
  $\Qint(e,W,t)-\Qint(0,W,t)$ does not vary with $(W,t)$, and the
  surrogate-discounted censoring odds
  $\rho(t)=v^{-1}\EE\{\sigma^2_k(S,W,t)(1-g_\Delta)/g_\Delta\mid E=k,t\}$ do
  not depend on $k$ and are nondecreasing in $t$. Then the relative burden
  $\mathcal B(e)/V^{\mathrm{full}}(e)$ is a weighted average of $\rho(t)$ over
  $t\in\Te=\{e+1,\ldots,T-1\}$ with weights proportional to $1+(T-t)^{-1}$, and
  is therefore nondecreasing in $e$.
\end{enumerate}
\end{theorem}

Part (i) is exact and requires no homogeneity: administrative censoring
costs information only through the part of the outcome that the surrogate
does not predict, so if $\sigma^2_k$ is small where $g_\Delta$ is small, the
burden is small even under heavy censoring. In the setting of part
(ii), $\rho(t)$ is the censoring odds discounted by the average fraction of the
outcome variance that the surrogate leaves unexplained. Part (ii) formalizes
the coupling described in the introduction. When censoring increases with
calendar time, the overlap windows of long exposure times sit in the most
heavily censored periods, and the share of their variance due to censoring is
largest. Besides the design and the follow-up schedule, the burden is
reduced mainly by the predictive strength of the surrogate and the
covariates. At the cluster level the censoring part is
$\mathcal B(e)/\bar M$ when the observation indicators are conditionally
independent across records, and within-cluster covariances enter only the
full-data part (Supplement~S2.7).

\subsection{Allocating Clusters to Waves}
\label{sec:design_crit}

Theorem~\ref{thm:burden} describes the bound for a given rollout, but the
rollout is usually chosen by the investigators, and the burden depends on it
through $g_E$. Under a simple working model intended for planning, the
cluster-level bound has a closed form. Suppose the
residual $Y_{ijt}(k)-\Qint(k,W_{ijt},t)$ is $b_j+\epsilon_{ijt}$, with
$\Var(b_j\mid\bW_j)=\sigma_b^2$ and independent individual terms of variance
$\sigma_\epsilon^2$; that $\Var(Y\mid S,E,W,t)=\sigma_r^2$; and that the censoring odds for a record of period $t$ in a
cluster that crosses over at $s$ average $o_t(s)=\EE\{(1-g_\Delta)/g_\Delta\mid
t,\tau_j=s\}$, which allows administrative censoring to depend on the wave as
well as the period. Any of the targets above is a linear functional whose
efficient influence function puts weight $c_t(s)$ on the residuals of a period-$t$
record in a cluster with crossover time $s$; for $\Psi(e)$,
\[
  c_t(s)=\frac{\mathbf 1\{t\in\Te\}}{\sum_{u\in\Te}n_u}
  \Bigl\{\frac{I(s=t-e+1)}{p_\tau(t-e+1)}-\frac{I(s>t)}{\PP(\tau_j>t)}\Bigr\},
\]
and those of $\bar\Psi$ are their average over $e$. We leave $\PsiIT$ out,
because its weights over exposure times change with the allocation
(Proposition~\ref{prop:it}).

\begin{proposition}[Variance as a function of the allocation]
\label{prop:design}
Suppose the conditions of Lemma~\ref{lem:summation} and the working model
above hold, with $n_t$ records per cluster in period $t$. Let the observation
probability $g_\Delta(1\mid S,E,W,t,\tau_j)$ depend on the crossover wave if
need be, with the observation indicators of a cluster conditionally
independent given its other data, and let $\phij$ be the cluster-level
influence function of SA-AIPW with a consistent bridge regression and an
observation model that conditions on the wave (Section~\ref{sec:tmle}); $\phij$
is $\EICj$ when censoring depends on the wave only through exposure time. For
the target with coefficients $c_t(s)$,
\begin{equation}
  \label{eq:design_var}
  \Var(\phij)
  =V_W+\sum_{s=2}^{T}p_\tau(s)\Bigl[\sigma_b^2\Bigl\{\sum_t n_tc_t(s)\Bigr\}^2
  +\sum_t n_t\,c_t(s)^2\bigl\{\sigma_\epsilon^2+\sigma_r^2\,o_t(s)\bigr\}\Bigr],
\end{equation}
where $V_W$, the variance of the part that depends on $\bW_j$ alone, does not
depend on the allocation. For the longest exposure time
$\emax=T-2$ in a complete stepped wedge, only period $T-1$ contributes and
\begin{equation}
  \label{eq:neyman}
  \Var\{\phij(\emax)\}=V_W+\frac{\sigma_b^2+v(2)/n_{T-1}}{p_\tau(2)}
  +\frac{\sigma_b^2+v(T)/n_{T-1}}{p_\tau(T)},\qquad
  v(s)=\sigma_\epsilon^2+\sigma_r^2o_{T-1}(s),
\end{equation}
so that, for a fixed total number of clusters in the first and last waves,
the variance is smallest when $p_\tau(2)/p_\tau(T)=\{(\sigma_b^2+v(2)/n_{T-1})/
(\sigma_b^2+v(T)/n_{T-1})\}^{1/2}$.
\end{proposition}

The second part is a Neyman-type allocation. If the last wave is more
heavily censored, as when censoring follows the wave, it receives more
clusters rather than fewer, because each of its clusters carries less
information. Divided by $J$, \eqref{eq:design_var} approximates the variance
of SA-AIPW (Theorem~\ref{thm:asymp}) and can be evaluated before the trial.
Given pilot values of $\sigma_b^2$, $\sigma_\epsilon^2$, $\sigma_r^2$
and the censoring odds, one chooses the allocation that minimizes
$\sum_k w_k\Var\{\hat\Psi_k\}$ for weights reflecting the targets of interest,
subject to a minimum number of clusters per wave; for the trial sizes common
in practice the allocations can be enumerated. The working model
ignores heteroscedasticity and the spread of the observation probabilities
within a period, but it predicted the standard deviations of SA-AIPW to within 8\% in the
design of Section~\ref{sec:sim:design}, overstated them by 7\% to 19\% in the
simulation calibrated to the EPT trial, and ranked the allocations correctly
in both (Sections~\ref{sec:sim:results} and~\ref{sec:sim:ept}). The proof is in Supplement~S2.14, and
the function \texttt{design\_search} in our software performs the
enumeration. For mixed-model analyses of a constant effect, the balanced
allocation is known not to be optimal, and the first and last waves warrant
more clusters \citep{lawrie2015optimal}; see also
\citet{kasza2019information} on which cluster-periods carry the information,
and \citet{moerbeek2023optimal}, where attrition makes the optimal allocation
asymmetric. Here the target is an exposure-time effect, no outcome model is
assumed, and the information lost to censoring enters through $o_t(s)$.

\section{Surrogate-Assisted AIPW Estimation}
\label{sec:tmle}

Proposition~\ref{prop:dml_remainder}(b) shows that, under the design, every
choice of bridge regression gives an estimator whose bias is the single
product $R_{\Delta,k}$; the choice matters only for efficiency. SA-AIPW
estimates the bridge by integrating over the empirical surrogate distribution
within each exposure level and uses it in the augmentation term. Throughout,
$\hat m\equiv\QYhat$ and $\hat\pi\equiv\gDhat$.

\paragraph{Stage 1: cross-fitted outcome and observation models.} Partition
the clusters into $V$ folds $F_1,\ldots,F_V$, stratified by crossover wave so
that every fold's training set contains clusters from every wave. For each
fold $v$, fit $\QY(S,E,W,t)=\EE(Y\mid S,E,W,t,\Delta=1)$ and
$g_\Delta(1\mid S,E,W,t)$ on the clusters outside $F_v$, pooling across
exposure levels with exposure time as a covariate, and evaluate the fits on
the clusters in $F_v$. Any learner may be used; Supplement~S3 describes the
working models used in Section~\ref{sec:sim}. The exposure
law~\eqref{eq:gE} is known and is plugged in. When administrative censoring
depends on the crossover wave beyond exposure time, the observation model may
also condition on the wave. The pseudo-outcome below then remains doubly
robust, because $\Delta$ is independent of $Y$ given $(S,E,W,t,\tau_j)$ in the
structural model, and the estimator remains consistent and asymptotically
normal, although its influence function need not be the EIF of
Proposition~\ref{prop:eif}.

\paragraph{Stage 2: surrogate-assisted pseudo-outcome.} For every record,
\begin{equation}
  \label{eq:tildeY}
  \widetilde Y_{ijt}
  \;=\;
  \hat m(S_{ijt},E_{jt},W_{ijt},t)
  \;+\;
  \frac{\Delta_{ijt}}{\hat\pi(S_{ijt},E_{jt},W_{ijt},t)}
  \bigl\{Y_{ijt} - \hat m(S_{ijt},E_{jt},W_{ijt},t)\bigr\}.
\end{equation}
The first term reconstructs the delayed outcome from the surrogate and the
second corrects it with the inverse-observation-weighted residual, so that
$\EE(\widetilde Y\mid S,E=k,W,t)=\hat m+(g^0_\Delta/\hat\pi)(\QY^0-\hat m)$,
which equals $\QY^0$ if either $\hat m$ or $\hat\pi$ is correct.

\paragraph{Stage 3: bridge regression by empirical surrogate integration.}
For each fold $v$, compute $\widetilde Y$ on the training clusters with the
fold-$v$ fits of Stage~1 and regress it on $(E,W,t)$ among training records in
the overlap periods, that is, records with $E=k$ and $t\in\mathbb{T}_k$ for
$k\ge1$ and records with $E=0$ and $t\in\Tov$. Evaluating the fit at each
exposure level $k$ on the clusters in $F_v$ gives $\Qinthat(k,W_{ijt},t)$ for
every record, including records with $E_{jt}\ne k$. The regression averages
over the surrogate values observed at each exposure level, so it integrates
over the empirical surrogate distribution without a density estimate. Few
records have long exposure times, because few clusters reach them during the
trial, so we pool across exposure levels, for example with a regression that
is additive in exposure-time indicators, $W$ and period indicators; separate
regressions for each level are an alternative when every level is well
populated. By Proposition~\ref{prop:dml_remainder}(b), errors in $\Qinthat$,
including those from pooling, affect only efficiency.

\paragraph{Baseline information.} The covariates $W$ may include cluster-level
summaries of data from periods in which no cluster is exposed, such as the
cluster mean of the surrogate in period~1. Their distribution does not depend
on the crossover time, so they are baseline covariates, and because they
predict the cluster effect $b_j$ they can reduce the variance. The
period-1 mean of the outcome qualifies too when its observation in that period
does not depend on the crossover wave. The consistency and asymptotic
normality results below hold with such summaries in $W$. The efficiency and
orthogonality results also require that a record's surrogate and outcome be
conditionally mean independent of the other records' covariates given its own
$W$, which calls for the period-1 cluster means of the individual covariates
alongside those of the surrogate and outcome. With few clusters, however,
cluster-level columns in the record-level regressions are fitted to a
handful of clusters and can make the cross-fitted predictions unstable
(Supplement~S4). Section~\ref{sec:history} therefore lets the cluster's
history enter through a single cluster-level index instead, and the
estimator of Section~\ref{sec:sim} uses no cluster-level covariates in the
regressions.

\paragraph{The estimator.} With $N_e=\sum_{j,i,t}\inde$, the SA-AIPW estimator
of $\Psi(e)$ is
\begin{equation}
  \label{eq:psia_aipw}
  \hat\Psi(e)
  =\frac{1}{N_e}\sum_{j=1}^{J}\sum_{i=1}^{n_j}\sum_{t=1}^{T}\inde\,
  \bigl\{\phi_e(O_{ijt};\hat\eta)-\phi_0(O_{ijt};\hat\eta)\bigr\},
\end{equation}
where $\phi_k$ is~\eqref{eq:phi} with the known $g_E$, the cross-fitted
$\hat m$ and $\hat\pi$ (entering through $\widetilde Y$), and $\Qinthat$.
Summaries are $\hat\Psi_\omega=\sum_e\omega_e\hat\Psi(e)$, and $\PsiIT$ is
estimated in the same way from the second estimating function described after
Proposition~\ref{prop:it}.

\begin{remark}[The unaugmented estimator]
\label{rem:unaug}
Setting $\Qinthat\equiv 0$ gives
$\hat\Psi^{\mathrm{u}}(e)=N_e^{-1}\sum\inde\{I(E=e)/g_E(e\mid t)-I(E=0)/
g_E(0\mid t)\}\widetilde Y$, in which the surrogate enters only through the
pseudo-outcome. By Proposition~\ref{prop:dml_remainder}(b) it has the same
bias as $\hat\Psi(e)$. Its influence function, however, exceeds the EIF by a
function of $(E,W,t)$ with mean zero given $(W,t)$, which lies along the
exposure directions removed from the tangent space. It is valid but not
efficient, and at the true bridge the augmentation is the projection that
removes the excess. This is the stepped-wedge analogue of covariate adjustment in an
individually randomized trial \citep{tsiatis2008covariate}; see
\citet[Section~8.4]{nabi2026semiparametric} for the geometric account. Because
the bridge regression varies with calendar time through the secular trend,
the excess can be large (Section~\ref{sec:sim}).
\end{remark}

\section{Asymptotic Theory and Inference}
\label{sec:asymptotics}

\subsection{Influence Functions and Variance Estimators}
\label{sec:eic:cluster}

Let $\eta^\dagger=(\QY^0,\Qdag,g_E,g^0_\Delta)$ be the probability limit of the
nuisance estimates, with $\Qdag=\Qint^0$ when the bridge regression is
consistent. The record-level influence function of $\hat\Psi(e)$ is
\begin{equation}
  \label{eq:phi_rec}
  \varphi_e(O)=\frac{\inde}{\pi_e}\bigl\{\phi_e(O;\eta^\dagger)-\phi_0(O;\eta^\dagger)-\Psi(e)\bigr\}
  =D^*_e(O)+h_e(O),
\end{equation}
\begin{equation}
  \label{eq:h}
  h_e(O)=\frac{\inde}{\pi_e}\sum_{k\in\{e,0\}}\pm
   \Bigl\{1-\frac{I(E=k)}{g_E(k\mid t)}\Bigr\}
   \bigl\{\Qdag(k,W,t)-\Qint^0(k,W,t)\bigr\},
\end{equation}
with sign $+$ for $k=e$ and $-$ for $k=0$. At the cluster level,
$\phij(e)=\bar M^{-1}\sum_{i,t}\varphi_e(O_{ijt})=\EICj(e)+H_j(e)$, with
$H_j(e)=\bar M^{-1}\sum_{i,t}h_e(O_{ijt})$, and we
write $\phij$, $\EICj$ and $H_j$ for the vectors over $e=1,\ldots,\emax$. The
term $h_e$ is a function of $(E,W,t)$ with mean zero given $(W,t)$, so it lies
in $\calT^\perp$.

\paragraph{Variance estimators.} Let
$\hat\varphi_j(e)=(J/N_e)\sum_{i,t}\inde\{\phi_e(O_{ijt};\hat\eta)-
\phi_0(O_{ijt};\hat\eta)-\hat\Psi(e)\}$ be the estimator's own estimating
function, evaluated at the cross-fitted nuisances. The cluster-robust sandwich
covariance is
\begin{equation}
  \label{eq:var}
  \widehat\Sigma_{\mathrm{sand}}=\frac{1}{J(J-1)}\sum_{j=1}^J
  \hat\varphi_j\hat\varphi_j^{\top}.
\end{equation}
A sandwich $\widehat\Sigma_{\mathrm{EIF}}$ can be formed in the same way from an
estimate of the EIF; the two differ when the estimator is not augmented with
the same $\Qinthat$. The leave-one-cluster-out jackknife reruns Stages~1--3
and~\eqref{eq:psia_aipw} on the $J-1$ clusters other than $j$, \emph{keeping
the fold partition fixed} and removing cluster $j$ from its fold, giving
$\hat\Psi^{(-j)}$, and
\begin{equation}
  \label{eq:jack_var}
  \widehat\Sigma_{\mathrm{jack}}
  =\frac{J-1}{J}\sum_{j=1}^{J}
  \bigl(\hat\Psi^{(-j)}-\hat\Psi^{(\cdot)}\bigr)
  \bigl(\hat\Psi^{(-j)}-\hat\Psi^{(\cdot)}\bigr)^{\top},
  \qquad
  \hat\Psi^{(\cdot)}=J^{-1}\sum_{j}\hat\Psi^{(-j)}.
\end{equation}
For a contrast $a^\top\Psi$ the level $1-\alpha$ interval is
\begin{equation}
  \label{eq:ci_t}
  a^\top\hat\Psi\;\pm\;t_{J-1,\,\alpha/2}\,(a^\top\widehat\Sigma\,a)^{1/2},
\end{equation}
with $\widehat\Sigma$ either estimator and $t_{J-1,\alpha/2}$ the upper
$\alpha/2$ quantile of Student's $t$ with $J-1$ degrees of freedom.

\paragraph{Binary outcomes.} With a rare binary outcome, an arm mean that
rests on a few observed events has a standard error that moves with the
estimate, and the symmetric interval~\eqref{eq:ci_t} then misses mostly on
one side (Section~\ref{sec:sim:ept}). Since $\hat\Psi(e)=\hat\mu_e-\hat\mu_0$
is a difference of the arm means $\hat\mu_k$, the averages of $\phi_k$ over
$\Te$, we give each arm mean a $t_{J-1}$ interval $(l_k,u_k)$ on the logit
scale, with standard error from its own cluster-level estimating function,
and combine the two by the method of variance estimates recovery
\citep{newcombe1998interval,zou2008construction}:
\begin{equation}
  \label{eq:mover}
  \hat\Psi(e)-\bigl\{a_e^2+b_0^2-2\hat r\,a_eb_0\bigr\}^{1/2},\qquad
  \hat\Psi(e)+\bigl\{b_e^2+a_0^2-2\hat r\,b_ea_0\bigr\}^{1/2},
\end{equation}
with $a_k=\hat\mu_k-l_k$, $b_k=u_k-\hat\mu_k$ and $\hat r$ the estimated
correlation of $\hat\mu_e$ and $\hat\mu_0$ \citep{zou2008estimation}. Linear
summaries such as $\bar\Psi$ are handled in the same way
(Supplement~S3). When the arm limits are symmetric,~\eqref{eq:mover} is the
sandwich interval~\eqref{eq:ci_t}, so the asymptotic coverage is unchanged.

\subsection{Asymptotic Normality, Double Robustness and Efficiency}
\label{sec:eic:asymp}

\begin{theorem}[Asymptotic behavior of SA-AIPW]
\label{thm:asymp}
Suppose Assumptions~\ref{ass:consistency}--\ref{ass:pos_trt} hold, together
with the following conditions.
\begin{enumerate}[label=\normalfont(C\arabic*),leftmargin=2.2em]
  \item \emph{(Clusters.)} The cluster vectors $O_j$ are i.i.d.\ across
    $j=1,\ldots,J$, with $n_j\le\bar n<\infty$.
  \item \emph{(Boundedness.)} $|Y|\le K$; $g^0_\Delta\ge c$ and
    $\hat\pi\ge c$ on the overlap periods; $\hat m$ and $\Qinthat$ are
    uniformly bounded.
  \item \emph{(Product rate.)}
    \begin{equation}
      \label{eq:rate_condition}
      \max_{e\le\emax,\,k\in\{0,e\}}\bigl\|\hat m - \QY^0\bigr\|_k
      \,\bigl\|\hat\pi - g^0_\Delta\bigr\|_k
      \;=\; \oP(J^{-1/2}).
    \end{equation}
  \item \emph{(Cross-fitting and convergence.)} The nuisance estimators are
    cross-fitted at the cluster level with a fixed number of folds, and
    $\|\hat m-\QY^0\|_k$, $\|\hat\pi-g^0_\Delta\|_k$ and
    $\|\Qinthat(k,\cdot)-\Qdag(k,\cdot)\|_k$ are $o_P(1)$ for some bounded
    $\Qdag$.
\end{enumerate}
Then:
\begin{enumerate}[label=\normalfont(\roman*),leftmargin=2.2em]
  \item $\hat\Psi-\Psi=J^{-1}\sum_{j=1}^J\phij+o_P(J^{-1/2})$ jointly over
    $e=1,\ldots,\emax$, with $\phij=\EICj+H_j$;
  \item $\sqrt J(\hat\Psi-\Psi)\to_d\mathcal N(0,\Sigma_\varphi)$ with
    $\Sigma_\varphi=\Var(\phij)$, and
    $J\,\widehat\Sigma_{\mathrm{sand}}\to_P\Sigma_\varphi$;
  \item if $\Qdag=\Qint^0$, then $H_j=0$ and $\hat\Psi$ attains the
    efficiency bound $\Sigma_{\mathrm{eff}}=\Var(\EICj)$ of the cluster-level
    model of Lemma~\ref{lem:summation};
  \item if (C3) and the first two rates in (C4) are replaced by the
    requirement that one of $\hat m$, $\hat\pi$ converges to the truth and the
    other to a bounded limit (bounded below by $c$ for $\hat\pi$), then
    $\hat\Psi\to_P\Psi$, whatever the limit of $\Qinthat$.
\end{enumerate}
\end{theorem}

The proof (Supplement~S2.8) splits the error into an empirical-process term,
which cluster-level cross-fitting makes $o_P(J^{-1/2})$ without Donsker
conditions \citep{zheng2011cross,chernozhukov2018double}, and the remainder
$B_k$, which Proposition~\ref{prop:dml_remainder}(b) reduces to
$R_{\Delta,k}$, of order $o_P(J^{-1/2})$ under (C3); then
Lemma~\ref{lem:summation} with the delta method for the normalization by
$N_e$ turns the leading term into $J^{-1}\sum_j\phij$. Part (iii), like
Theorem~\ref{thm:jack}(a) and Proposition~\ref{prop:alloc} below, uses the
orthogonality of $\EICj$ to the assignment directions (Supplement~S2.5),
which rests on the structural model of Section~\ref{sec:design} and on the
condition on baseline information stated in Section~\ref{sec:tmle}. Condition (C3)
involves the outcome regression and the observation model only. It holds if
both converge at rate $o_P(J^{-1/4})$, or if one is a correctly specified
parametric model and the other is consistent. No condition is placed on
$f_S$, which is never estimated, or on the rate of $\Qinthat$, which affects
only the variance through $H_j$. Condition (C2) asks for a bounded outcome
and observation probabilities bounded away from zero; the Gaussian design of
Section~\ref{sec:sim:design} satisfies neither literally, and the results
there should be read as covering it only up to moment conditions that would
replace (C2) at the cost of longer proofs.

\subsection{Sandwich and Jackknife}
\label{sec:eic:be}

\begin{theorem}[Sandwich and jackknife]
\label{thm:jack}
Suppose the conditions of Theorem~\ref{thm:asymp} hold, and
\begin{enumerate}[label=\normalfont(C\arabic*),leftmargin=2.2em]
  \setcounter{enumi}{4}
  \item \emph{(Jackknife linearization.)}
    $J^{-1}\sum_{j=1}^J\bigl\|(J-1)(\hat\Psi^{(-j)}-\hat\Psi)
    +(\phij-\bar\varphi_J)\bigr\|^2=o_P(1)$, where
    $\bar\varphi_J=J^{-1}\sum_{l=1}^J\varphi_l$.
\end{enumerate}
Then:
\begin{enumerate}[label=\normalfont(\alph*),leftmargin=2.2em]
  \item \emph{(Pythagorean decomposition.)} $\EE(\EICj H_j^{\top})=0$, so
    $\Sigma_\varphi=\Sigma_{\mathrm{eff}}+\Sigma_H$ with
    $\Sigma_H=\Var(H_j)\succeq0$.
  \item \emph{(EIF-based sandwich.)} If the estimated EIF uses consistent
    nuisance estimates, $J\,\widehat\Sigma_{\mathrm{EIF}}\to_P
    \Sigma_{\mathrm{eff}}$, which understates the variance of $\hat\Psi$ by
    $\Sigma_H$; for a scalar target, with $\sigma^2_H$ and
    $\sigma^2_{\mathrm{eff}}$ the scalar versions of $\Sigma_H$ and
    $\Sigma_{\mathrm{eff}}$, the variance is understated by the factor
    $1+\sigma^2_H/\sigma^2_{\mathrm{eff}}$, which does not depend on $J$.
  \item \emph{(Jackknife.)} $J\,\widehat\Sigma_{\mathrm{jack}}\to_P
    \Sigma_\varphi$, and the intervals~\eqref{eq:ci_t} with
    $\widehat\Sigma=\widehat\Sigma_{\mathrm{jack}}$ have asymptotic coverage
    $1-\alpha$.
\end{enumerate}
\end{theorem}

Part (a) holds because $\EICj$ lies in the cluster-level tangent space while
$H_j$ lies along the removed assignment directions; it uses the absence of
interference and carryover beyond exposure time (Supplement~S2.9). The
jackknife is convenient here because it needs no formula for $\phij$: each
leave-one-cluster-out refit records how the implemented estimator responds to
removing a cluster, so it differentiates the procedure actually run rather
than a first-order formula that may not match the implementation. Condition
(C5) is high-level; the next result reduces it to a stability property of the
nuisance fits.

\begin{proposition}[Validity of the fixed-partition jackknife]
\label{prop:jack_valid}
Suppose the conditions of Theorem~\ref{thm:asymp} hold with the fold
partition held fixed across the leave-one-cluster-out refits, the exposure
law plugged in, and
\begin{enumerate}[label=\normalfont(S),leftmargin=2.2em]
  \item \emph{(Leave-one-cluster-out stability.)}
    $\max_{j,v}\|\hat\eta^{(v,-j)}-\hat\eta^{(v)}\|=O_P(J^{-1})$ in
    $L^2(P)$,
    where $\hat\eta^{(v)}$ is the fold-$v$ nuisance fit and
    $\hat\eta^{(v,-j)}$ its refit without cluster $j$.
\end{enumerate}
Then (C5) holds, and the conclusions of Theorem~\ref{thm:jack}(c) follow.
\end{proposition}

The proof (Supplement~S2.10) splits $\hat\Psi^{(-j)}-\hat\Psi$ into the
removal of cluster $j$'s own contribution, which gives the influence term, and
the change in the nuisance fits of the other folds. The second part is
$o_P(J^{-1})$ for two reasons particular to the design: changes in $\Qinthat$
enter only through terms whose expectation is exactly zero because $g_E$ is
known, and changes in $(\hat m,\hat\pi)$ enter $B_k$ only at second order,
through the product structure of $R_{\Delta,k}$. Re-randomizing the partition
in each refit would add variation of order larger than $J^{-1}$ and break
(C5). Stability holds for parametric working models fitted by least squares
or maximum likelihood; with data-adaptive learners we fix tuning parameters
and ensemble weights at their full-data values across the refits. As for
symmetric statistics \citep{efron1981jackknife}, the jackknife can be biased
upward with very few clusters. The sandwich~\eqref{eq:var}, built from the
estimator's own estimating function, is also consistent for $\Sigma_\varphi$
by Theorem~\ref{thm:asymp}(ii).

\subsection{Allocation of Crossover Times}
\label{sec:alloc}

Theorems~\ref{thm:asymp}--\ref{thm:jack} treat crossover times as i.i.d.\
draws. Most stepped-wedge trials instead fix the number of clusters per wave
and randomly permute them, as in the EPT trial (waves of 6, 6, 6 and 5
LHJs).

\begin{proposition}[Allocation invariance]
\label{prop:alloc}
Suppose the crossover times are a uniformly random permutation of a fixed
allocation in which crossover time $s$ occurs $Jp_\tau(s)$ times, the folds
are stratified by wave, and the conditions of Theorem~\ref{thm:asymp} hold
otherwise. Let $a(\tau)=\EE(H_j\mid\tau_j=\tau)$.
Then $\sqrt J(\hat\Psi-\Psi)\to_d\mathcal N\bigl(0,\Sigma_\varphi-\Var\{a(\tau_j)\}\bigr)$,
while $J\widehat\Sigma_{\mathrm{sand}}$ still converges to $\Sigma_\varphi$, as
does $J\widehat\Sigma_{\mathrm{jack}}$ if (C5) holds under the fixed
allocation. In particular, if $\Qdag=\Qint^0$, the
asymptotic variance is $\Sigma_{\mathrm{eff}}$ under both allocations and both
variance estimators are consistent; otherwise they are conservative by
$\Var\{a(\tau_j)\}$.
\end{proposition}

The proof (Supplement~S2.11) uses the decomposition $\phij=\EICj+H_j$. The
conditional mean of $\EICj$ given the crossover times and covariates depends
on the covariates only, so the variance of $\sum_j\EICj$ is the same under
both allocations; only the component of $H_j$ that is a function of the
crossover time alone loses its variance, because under a fixed allocation its
sum over clusters is constant. A finite-population central limit theorem
\citep{li2017general} gives normality. The augmented estimator is therefore
insensitive to the allocation scheme, whereas the variance of the unaugmented
estimator depends on it, and i.i.d.-based variance estimates for it can be
conservative under fixed-wave allocations.

\subsection{Pre-Crossover Histories}
\label{sec:history}

The cluster-level model of Lemma~\ref{lem:summation} uses only one consequence
of randomization, namely that the crossover time is independent of the
covariates. The structural model of Section~\ref{sec:design} implies more. A
cluster that has not yet crossed over is in the control condition, so its data
up to period $m$ have the same distribution whatever its crossover time, as
long as that time is later than $m$. Mixed models exploit something like this
when they compare a cluster with itself over time, and it is natural to ask
whether a design-based estimator can do the same.

It helps to look at the rollout one period at a time. Let
$h_{m+1}=p_\tau(m+1)/\PP(\tau_j>m)$ be the probability that a cluster still
unexposed after period $m$ crosses over at $m+1$, put
\[
  \zeta_{jm}=I(\tau_j>m)\bigl\{I(\tau_j=m+1)-h_{m+1}\bigr\},
  \qquad m=1,\ldots,T-2,
\]
and let $\mathcal H_{jm}$ denote cluster $j$'s covariates for all periods
together with its observed data from periods $1,\ldots,m$. We assume here that
censoring in periods before a cluster's crossover does not depend on its wave
(otherwise $\mathcal H_{jm}$ should include only the covariates and
surrogates). Randomization and the absence of anticipation then give
$\EE(\zeta_{jm}\mid\mathcal H_{jm})=0$ for every $m$: the crossover is
sequentially randomized with known probabilities. Write $\mathcal M_S$ for the
cluster-level model defined by these restrictions; it is contained in the
model of Lemma~\ref{lem:summation}.

\begin{proposition}[Efficiency bound with pre-crossover histories]
\label{prop:history}
Suppose the conditions of Lemma~\ref{lem:summation} hold.
\begin{enumerate}[label=\normalfont(\roman*),leftmargin=2.2em]
  \item In $\mathcal M_S$ the canonical gradient of $\Psi(e)$ is
  \begin{equation}
    \label{eq:eif_history}
    \EICj^S(e)=\EICj(e)-\sum_{m=1}^{T-2}\zeta_{jm}\,\alpha^*_m(\mathcal H_{jm}),
  \end{equation}
  where
  $\alpha^*_m(\mathcal H)=\EE\{\EICj(e)\mid\tau_j=m+1,\mathcal H_{jm}=\mathcal H\}
  -\EE\{\EICj(e)\mid\tau_j>m+1,\mathcal H_{jm}=\mathcal H\}$. The efficiency
  bound falls by
  $\Var\{\EICj(e)\}-\Var\{\EICj^S(e)\}=\sum_m\EE\{\zeta_{jm}^2\,\alpha^*_m(\mathcal H_{jm})^2\}$.
  \item For any functions $\alpha_m$, with $\hat\alpha_m$ cross-fitted and
  converging to $\alpha^\dagger_m$, the estimator
  \[
    \hat\Psi^S(e)=\hat\Psi(e)-\frac1J\sum_{j=1}^J\sum_{m=1}^{T-2}\zeta_{jm}\,\hat\alpha_m(\mathcal H_{jm})
  \]
  is consistent and asymptotically normal under the conditions of
  Theorem~\ref{thm:asymp}, with influence function
  $\phij(e)-\sum_m\zeta_{jm}\alpha^\dagger_m(\mathcal H_{jm})$. It attains the
  bound in (i) when $\Qdag=\Qint^0$ and $\alpha^\dagger_m=\alpha^*_m$.
  \item Suppose that for a record of a period $t>m$ in a cluster still unexposed
  after period $m$,
  \begin{equation}
    \label{eq:index_structure}
    Y_{ijt}(k)-\Qint(k,W_{ijt},t)=\lambda^\top u_j+\epsilon_{ijt},
  \end{equation}
  where $u_j$ is a vector of cluster effects with $\EE(u_j\mid\bW_j)=0$,
  $\epsilon_{ijt}$ has mean zero given $(u_j,\mathcal H_{jm})$, and $\lambda$
  depends on neither $t$ nor $k$; and suppose that $\EE(u_j\mid\mathcal
  H_{jm})=B_m\bar r_{jm}$ is linear in a vector $\bar r_{jm}$ of summaries of the
  cluster's data from periods $1,\ldots,m$. Then
  \begin{equation}
    \label{eq:index}
    \sum_{m=1}^{T-2}\zeta_{jm}\,\alpha^*_m(\mathcal H_{jm})=\lambda^\top Z_j,
    \qquad
    Z_j=\sum_{m=1}^{T-2}\zeta_{jm}\,\ell_{m,e}\,B_m\bar r_{jm},
  \end{equation}
  with the design constants
  $\ell_{m,e}=\EE\{L_e(\tau_j)\mid\tau_j=m+1\}-\EE\{L_e(\tau_j)\mid\tau_j>m+1\}$
  of Supplement~S2.13. The correction of (ii) with
  $\hat\alpha_m(\mathcal H_{jm})=\hat\kappa^\top\ell_{m,e}\hat B_m\bar r_{jm}$,
  where $\hat B_m$ is cross-fitted and $\hat\kappa$ is the least-squares
  coefficient of $\phij(e)$ on $Z_j$ over clusters, attains the bound in (i)
  when $\Qdag=\Qint^0$ and $\hat B_m\to B_m$, since then $\hat\kappa\to\lambda$.
  If~\eqref{eq:index_structure} or the linearity fails, the estimator remains
  consistent and asymptotically normal, with asymptotic variance
  $\Var\{\phij(e)\}-\Cov\{\phij(e),Z_j\}^\top\Var(Z_j)^{-1}\Cov\{\phij(e),Z_j\}
  \le\Var\{\phij(e)\}$.
\end{enumerate}
\end{proposition}

For a fixed loading the correction term in (ii) has mean zero by design,
whatever $\hat\alpha_m$ is; the least-squares loading adds a bias of order
$1/J$, as for regression adjustment in randomized experiments
\citep{freedman2008regression,lin2013agnostic}. Part (iii) says that under a
random-intercept structure for the pre-crossover data the whole correction
lives in the span of a single index built from the cluster's own history, so
that one coefficient per target is fitted at the cluster level. Our
implementation (Supplement~S3) takes $u_j=(u_{S,j},u_{Y,j})$ to be the cluster
effects of the surrogate and of the outcome given the surrogate, estimates the
five variance components of a bivariate random-intercept model from the
pre-crossover residuals of the other folds, and forms $B_m\bar r_{jm}$ as the
best linear predictor of $u_j$ from the cluster's residual means over periods
$1,\ldots,m$; the loading is the least-squares coefficient of the estimating
function on the index. The estimator is the intercept of a cluster-level
regression, and its standard error is the heteroskedasticity-robust standard
error of that intercept with the HC3 leverage correction
\citep{bell2002bias,lin2013agnostic} on $J-2$ degrees of freedom; the
jackknife of Section~\ref{sec:eic:be} remains available.

The condition that censoring before crossover does not depend on the wave is
needed for $\mathcal M_S$ but not for the estimator. If the weights $B_m$ are
held fixed given the other folds (we use the average record counts of the
training clusters in place of the cluster's own), $Z_j$ is a fixed linear
function of the residual means, and $\EE(\zeta_{jm}\bar r_{jm})=0$ requires
only the sequential randomization of $\tau_j$, the independence of the
observation indicators from the outcome residuals given the surrogate,
covariates, period and crossover time (Assumption~\ref{ass:mar} with the
wave among the conditioning variables, as in the structural model of
Section~\ref{sec:design}), and residual regressions with mean-zero errors;
the observation probabilities may then depend on the wave
(Supplement~S2.13).

If the residual $Y-\Qint$ contains a cluster effect, projecting on the
histories removes what they reveal about that effect beyond the covariates
(Supplement~S2.13). In the Gaussian design of Section~\ref{sec:sim:design}
the conditional means in $\alpha^*_m$ have closed forms, so the bound can be
evaluated exactly, and so can the variance of the exposure-time mixed model
(Supplement~S4). The histories lower the bound by 9\% to 23\% in variance
relative to the model that treats the period-1 data as baseline. The mixed
model fitted to the periods in which some cluster is still unexposed has
the precision of the bound for $\Psi(1)$ to $\Psi(3)$ to within 1\%, and lies 2\%, 12\% and 5\% below it for $\Psi(4)$, $\Psi(5)$ and $\bar\Psi$;
the final period, in which every cluster is exposed and no design-based
estimator has a comparison, lowers its standard errors by a further 1.5\% to
15\%, most at the longest exposure time. The design information therefore
accounts for the mixed model's precision except for the longest exposure
times, where its additive structure pools the period effect across exposure
levels and extrapolates to the final period. The estimator of part (iii)
attains the bound with 30 clusters in Section~\ref{sec:sim} and comes within
6\% to 16\% of it with 18 and 25\% to 36\% with 12.

\subsection{Testing Exposure-Time Heterogeneity}
\label{sec:het}

Whether effects change with exposure time is often the first question.
\citet{maleyeff2023assessing} answer it inside a mixed model, with a
permutation test for the variance of exposure-time random effects. The test
below is a Wald-type test on $\hat\Psi$ and requires no outcome model.
Let $C$ be the $(\emax-1)\times\emax$ contrast matrix with rows
$(-1,0,\ldots,1,\ldots,0)$, so that $H_0:C\Psi=0$ states that the
exposure-time effects $\Psi(1),\ldots,\Psi(\emax)$ are equal, and let
$q=\emax-1$. Because $\Psi(e)$ averages over the overlap periods
$\Te=\{e+1,\ldots,T-1\}$, which move later as $e$ grows, $H_0$ coincides
with the absence of exposure-time heterogeneity when the period-specific
effects $\delta_t(e)$ do not change with calendar time
(Proposition~\ref{prop:it}); if they do, the $\Psi(e)$ can differ even when
$\delta_t(e)$ is the same for every $e\ge1$ within each period, and the test
then rejects because the windows differ. A test free of this would compare
exposure times within a common set of periods, but with $T=7$ only period 6
supports all five exposure times, so such a test would rest on a single
period. We therefore state $H_0$ as the equality of the overlap-window
averages.

\begin{corollary}[Heterogeneity test]
\label{cor:het}
Under the conditions of Theorem~\ref{thm:jack} and $H_0$, and if
$C\Sigma_\varphi C^{\top}$ is nonsingular, the statistic
\begin{equation}
  \label{eq:T2}
  T^2=(C\hat\Psi)^{\top}\bigl(C\widehat\Sigma_{\mathrm{jack}}C^{\top}\bigr)^{-1}
  (C\hat\Psi)
\end{equation}
converges in distribution to $\chi^2_q$. For small $J$ we refer
$F=\{(J-q)/(q(J-1))\}T^2$ to the $F_{q,J-q}$ distribution, the Hotelling
calibration for a covariance estimated from $J$ clusters.
\end{corollary}

The $F$ reference is a heuristic small-sample calibration with the same
limit. Under fixed-wave allocation, Proposition~\ref{prop:alloc} implies that
$T^2$ is conservative unless the bridge regression is consistent; the same
holds for the sandwich version of the test.

Supplement~S1 bounds the coverage error of~\eqref{eq:ci_t} by a
Berry--Esseen term of order $J^{-1/2}$ plus the tail probabilities of the
variance-estimation error and of the remainder, which the conditions above
control only in probability; the bound is conservative, and we rely on
simulation for small samples.

\section{Simulation Study}
\label{sec:sim}

\subsection{Design}
\label{sec:sim:design}

We simulated cross-sectional stepped-wedge trials with $T=7$ periods,
$n=40$ individuals per cluster-period and $J\in\{12,18,30,48\}$ clusters.
Crossover times were allocated in fixed waves: an equal number of clusters
was assigned to each of the crossover periods $2,\ldots,7$, and the assignment
was randomly permuted across clusters. With the secular trend
$\lambda(t)=0.5\sin\{\pi(t-1)/6\}+0.3\{(t-1)/6\}^2$ and cluster effects
$b_j\sim\mathcal N(0,0.034)$, records were generated as
\begin{align*}
  W_1&\sim\mathcal N(0,1),\quad W_2\sim\mathrm{Bernoulli}(0.4),\quad
  W_3\sim\mathrm{Uniform}(0,1),\\
  S&=\theta(E)+0.4W_1-0.3W_2+0.5\lambda(t)+0.6b_j+\epsilon_S,\\
  \logit\PP(\Delta=1)&=3.0-\kappa\,I(t\ge6)-0.8\,I(t=7)+0.4S+0.3A-0.2W_3,\\
  Y&=\beta(E)+0.5S+0.4W_1-0.2W_2+0.6W_3+0.8\lambda(t)+b_j+\epsilon_Y,
\end{align*}
with $\epsilon_S\sim\mathcal N(0,0.25)$ and $\epsilon_Y\sim\mathcal N(0,\sigma^2_Y)$ independent,
and $\kappa=2.5$ and $\sigma^2_Y=0.64$ unless stated otherwise. Censoring is
driven by the calendar, as administrative censoring is, and also depends on
the surrogate, the treatment and a covariate, so complete-case analyses are
not consistent in general. With $\kappa=2.5$ between 4\% and 6\% of records are censored in
periods 1--5, 30\%
in period 6 and 47\% in period 7; with $\kappa=1$ and $\kappa=4$, 9\% and
65\% are censored in period 6. Period 6 is the only overlap period of
$\Psi(5)$ and one of five for $\Psi(1)$.
In the \emph{heterogeneous} scenario, $\theta(e)=0.6\{1-\exp(-e/1.5)\}$ and
$\beta(e)=-0.10+0.03(e-1)$ for $e\ge1$, so the surrogate effect builds up over
the first few exposure periods. The true exposure-time effects are
$\Psi(e)=\beta(e)+0.5\theta(e)=(0.046,0.151,0.219,0.269,0.309)$ for
$e=1,\ldots,5$, the exposure-averaged effect is $\bar\Psi=0.199$, and the
immediate-treatment estimand of Proposition~\ref{prop:it} is
$\PsiIT=0.131$. In the \emph{constant} scenario, $\theta(e)=0.8$ and
$\beta(e)=-0.28$, so $\Psi(e)=0.12$ for every $e$.

We compared SA-AIPW (Section~\ref{sec:tmle}) with its unaugmented version
(Remark~\ref{rem:unaug}), IPCW, and complete-case linear mixed models (LMM)
with a random cluster intercept and either a single treatment indicator
\citep[the immediate-treatment, or IT, model;][]{hussey2007design} or
exposure-time indicators \citep[the ETI model;][]{kenny2022analysis}.
SA-AIPW used parametric record-level working models on the individual
covariates and the period, a bridge regression pooled over exposure levels,
and the history index of Proposition~\ref{prop:history}(iii) with the
surrogate and outcome histories, the defaults of our software;
Supplement~S3 gives all specifications, and Supplement~S4 reports the
estimator without the history term, the estimator of an earlier version of
this paper that used the period-1 cluster means of the surrogate and outcome
as covariates instead, and other variants. For SA-AIPW we computed the
sandwich, the HC3 standard error of Section~\ref{sec:history} and the
fixed-partition jackknife; intervals used $t_{J-1}$ quantiles ($t_{J-2}$ for
HC3) and the heterogeneity test used $F_{4,J-4}$.

The study had six blocks. Block~I (heterogeneous) used
$J\in\{12,18,30,48\}$ and Block~II (constant) used $J\in\{12,18,30\}$, with
500 replications per cell, except 1000 for $J=18$ and $J=30$ in Block~I
and 250 for $J=48$; the Monte Carlo standard error of a coverage near 0.95 is
about 0.007 with 1000 replications and 0.010 with 500. Block~III fixed $J=30$
and crossed $\sigma^2_Y\in\{0.64,0.16\}$, which changes how much of the
outcome the surrogate leaves unexplained, with $\kappa\in\{1,2.5,4\}$, which
changes censoring in the final periods; it compared SA-AIPW with the same
estimator applied to the uncensored outcomes. Block~IV misspecified the
outcome regression (no exposure terms and a linear calendar trend), the
observation model (surrogate omitted), or both. Block~V compared fixed-wave
with independent crossover times at $J=30$, with 300 replications. Block~VI
let the surrogate effect vary with calendar time, $\theta(e)\{0.5+(t-2)/4\}$,
so that it is $1.5$ times larger in period 6 than in period 4; this changes the
true effects to $\Psi(e)=(0.046,0.179,0.284,0.374,0.454)$ and
$\bar\Psi=0.267$ and $\PsiIT=0.142$, which average over the overlap periods. Block~VI was run
with the default censoring at $J=30$, and with censoring that depends strongly
on the surrogate (coefficient $1.5$ instead of $0.4$, and $\kappa=4$) at
$J=30$ and $J=48$, with 500 replications each.

\subsection{Results}
\label{sec:sim:results}

\begin{table}[!t]
\centering
\caption{Block I, $J=30$: estimation of exposure-time effects (true values in
parentheses). SE/SD is the ratio of the mean standard error to the Monte Carlo
standard deviation; standard errors are jackknife for SA-AIPW, sandwich for
SA-AIPW without the history term, and model-based for the complete-case
mixed models (LMM), whose ETI form is used for $\Psi(e)$ and $\bar\Psi$ and
IT form for $\PsiIT$. Table~S3 of the Supplementary Material adds the
unaugmented estimator and IPCW.}
\label{tab:sim_main}
\small
\begin{tabular}{llrrrr}
\toprule
Target & Estimator & Bias & SD & SE/SD & Coverage\\
\midrule
$\Psi(1)$ (0.046) & SA-AIPW & $-$0.001 & 0.038 & 1.14 & 0.973 \\
 & SA-AIPW, no history term & 0.000 & 0.066 & 1.04 & 0.960 \\
 & Complete-case LMM, ETI & $-$0.003 & 0.036 & 1.05 & 0.970 \\
\addlinespace
$\Psi(3)$ (0.219) & SA-AIPW & $-$0.002 & 0.059 & 1.08 & 0.957 \\
 & SA-AIPW, no history term & $-$0.002 & 0.110 & 1.00 & 0.943 \\
 & Complete-case LMM, ETI & $-$0.005 & 0.055 & 1.00 & 0.958 \\
\addlinespace
$\Psi(5)$ (0.309) & SA-AIPW & $-$0.001 & 0.126 & 1.16 & 0.958 \\
 & SA-AIPW, no history term & $-$0.001 & 0.183 & 1.05 & 0.950 \\
 & Complete-case LMM, ETI & $-$0.007 & 0.083 & 1.05 & 0.971 \\
\addlinespace
$\bar\Psi$ (0.199) & SA-AIPW & $-$0.001 & 0.055 & 1.14 & 0.968 \\
 & SA-AIPW, no history term & $-$0.001 & 0.103 & 1.03 & 0.942 \\
 & Complete-case LMM, ETI & $-$0.004 & 0.047 & 1.04 & 0.966 \\
\addlinespace
$\Psi_{\mathrm{IT}}$ (0.131) & SA-AIPW & $-$0.002 & 0.041 & 1.13 & 0.965 \\
 & SA-AIPW, no history term & $-$0.002 & 0.082 & 1.02 & 0.951 \\
 & Complete-case LMM, IT & $-$0.081 & 0.034 & 1.02 & 0.354 \\
\bottomrule
\end{tabular}
\end{table}

\begin{table}[!t]
\centering
\caption{Blocks I and II: coverage of 95\% SA-AIPW intervals with sandwich (s),
HC3 (h) and jackknife (j) standard errors under heterogeneous effects, and
rejection rates of the 5\%-level heterogeneity test with sandwich and
jackknife covariances under constant effects (size, Block~II, not run with
$J=48$) and heterogeneous effects (power, Block~I).}
\label{tab:sim_inf}
\small
\setlength{\tabcolsep}{3.2pt}
\begin{tabular}{rccccccccccccc}
\toprule
 & \multicolumn{9}{c}{Coverage, Block I} & \multicolumn{4}{c}{Heterogeneity test}\\
\cmidrule(lr){2-10}\cmidrule(lr){11-14}
 & \multicolumn{3}{c}{$\Psi(1)$} & \multicolumn{3}{c}{$\Psi(5)$} & \multicolumn{3}{c}{$\bar\Psi$} & \multicolumn{2}{c}{Size} & \multicolumn{2}{c}{Power}\\
$J$ & s & h & j & s & h & j & s & h & j & s & j & s & j\\
\midrule
12 & 0.948 & 0.980 & 0.982 & 0.892 & 0.918 & 0.982 & 0.916 & 0.956 & 0.982 & 0.108 & 0.044 & 0.354 & 0.168 \\
18 & 0.951 & 0.970 & 0.974 & 0.922 & 0.937 & 0.962 & 0.938 & 0.966 & 0.976 & 0.064 & 0.042 & 0.553 & 0.420 \\
30 & 0.963 & 0.972 & 0.973 & 0.923 & 0.947 & 0.958 & 0.946 & 0.963 & 0.968 & 0.056 & 0.042 & 0.858 & 0.802 \\
48 & 0.964 & 0.964 & 0.964 & 0.952 & 0.956 & 0.960 & 0.948 & 0.948 & 0.952 & -- & -- & 0.988 & 0.984 \\
\bottomrule
\end{tabular}
\end{table}

\paragraph{Estimation.} Table~\ref{tab:sim_main} reports Block~I at $J=30$.
SA-AIPW was essentially unbiased for every exposure-time effect, for
$\bar\Psi$ and for $\PsiIT$, and its jackknife intervals covered at rates
between 0.96 and 0.97. Its standard deviations, 0.038, 0.059, 0.126 and
0.055 for $\Psi(1)$, $\Psi(3)$, $\Psi(5)$ and $\bar\Psi$, are those of the
bound of Proposition~\ref{prop:history}, which is 0.039, 0.057, 0.117 and
0.053 in this design (Supplement Table~S14): the history index recovers what
the structural model allows. Without the history term the same estimator
had standard deviations 1.5 to 2.0 times larger. The unaugmented estimator
and IPCW, both without the history term, were also unbiased; their standard
deviations were 4\% to 8\% above those of the augmented estimator without
the history term, and their sandwich standard errors overstated their
standard deviations by a factor of 2.5 to 3.0 (Supplement Table~S3). For
the unaugmented estimator both are expected: under the fixed-wave allocation
most of $\Sigma_H$ (Theorem~\ref{thm:jack}), the variance that the
augmentation removes, is already removed by the allocation, while
i.i.d.-based variance estimates still count it (Proposition~\ref{prop:alloc});
Block~V separates the two. IPCW is not
covered by these results, since its influence function also carries an
observation-process term, but it lacks the augmentation in the same way and
behaved alike. The history index is applied after the augmentation: applied
to the unaugmented estimator, whose estimating function keeps a component
that depends on the wave alone, the in-sample loading produced a bias of
0.03 for $\Psi(1)$ under the fixed allocation (Supplement~S3). The
immediate-treatment estimand was $\PsiIT=0.131$, a third smaller than
$\bar\Psi=0.199$, as Proposition~\ref{prop:it} predicts when effects grow with
exposure time. The immediate-treatment mixed model estimated neither
quantity. Its coefficient averaged $0.050$, close to
$\Psi(1)$, and its intervals covered $\PsiIT$ in 35\% of replications, in line
with \citet{kenny2022analysis}. The mixed model with exposure-time indicators
was nearly unbiased in this scenario, in which effects depend only on exposure
time and censoring is moderate and confined to the last two periods. Its
standard deviations were within 7\% of those of SA-AIPW for $\Psi(1)$ to
$\Psi(3)$, 13\% smaller for $\Psi(4)$, 15\% smaller for $\bar\Psi$ and 34\%
smaller for $\Psi(5)$; Section~\ref{sec:history} and Supplement~S4 trace
these differences to the model's use of the final period and of effects
constant in calendar time.

\paragraph{Inference with few clusters.} Table~\ref{tab:sim_inf} reports
coverage and the heterogeneity test across $J$. The estimator is the
intercept of a cluster-level regression on the history index, so besides the
sandwich of its adjusted estimating functions, which ignores the estimation
of the loading, we report the HC3 standard error of that intercept on
$J-2$ degrees of freedom (Section~\ref{sec:history}) and the jackknife. With
30 or more clusters all three gave coverage between 0.92 and 0.97, the
sandwich lowest for $\Psi(5)$, which rests on a single period. With 18
clusters the sandwich fell to 0.92 for $\Psi(5)$ while HC3 and the jackknife
stayed at 0.94 to 0.98. With 12 clusters the sandwich was anticonservative
(0.89 to 0.95), HC3 was close to nominal for $\Psi(1)$ and $\bar\Psi$ (0.98
and 0.96) and the jackknife was conservative (0.98 to 0.99). For $\Psi(5)$
with 12 clusters HC3 covered at 0.92 although its standard errors were 16\%
larger than the standard deviation: the longest exposure time is then
observed in two clusters, and the estimating function is not close to
normal. Under constant effects (Block~II) the jackknife-based test had size
0.042 to 0.044 with 12, 18 and 30 clusters, and the sandwich-based test
0.108, 0.064 and 0.056. Under heterogeneous effects the power of the
sandwich-based test was 0.35, 0.86 and 0.99 with 12, 30 and 48 clusters; the
jackknife version was less powerful with few clusters (0.42 against 0.55
with 18; Supplement~S4).

\begin{table}[!tbp]
\centering
\caption{Block III, $J=30$: censoring in period 6 and variance of SA-AIPW
relative to the same estimator applied to the uncensored outcomes, by outcome
noise $\sigma^2_Y$ and late-period censoring $\kappa$. The last two columns give
the same ratio for IPCW.}
\label{tab:burden}
\small
\setlength{\tabcolsep}{4.5pt}
\begin{tabular}{ccr cccccc cc}
\toprule
 & & & \multicolumn{6}{c}{SA-AIPW} & \multicolumn{2}{c}{IPCW}\\
\cmidrule(lr){4-9}\cmidrule(lr){10-11}
$\sigma^2_Y$ & $\kappa$ & Cens.\ (\%) & $e=1$ & $e=2$ & $e=3$ & $e=4$ & $e=5$ & $\bar\Psi$ & $e=1$ & $e=5$\\
\midrule
0.64 & 1.0 & 9.0 & 1.01 & 1.06 & 1.05 & 1.06 & 1.04 & 1.04 & 1.16 & 1.16 \\
0.64 & 2.5 & 30.4 & 1.12 & 1.10 & 1.17 & 1.18 & 1.26 & 1.19 & 1.21 & 1.36 \\
0.64 & 4.0 & 65.2 & 1.45 & 1.43 & 1.63 & 1.83 & 2.73 & 2.01 & 1.44 & 2.29 \\
\addlinespace
0.16 & 1.0 & 9.1 & 1.06 & 1.09 & 1.07 & 1.07 & 1.06 & 1.09 & 1.19 & 1.24 \\
0.16 & 2.5 & 30.4 & 1.11 & 1.15 & 1.10 & 1.27 & 1.25 & 1.26 & 1.23 & 1.39 \\
0.16 & 4.0 & 65.2 & 1.44 & 1.54 & 1.55 & 1.86 & 2.34 & 2.04 & 1.47 & 2.07 \\
\bottomrule
\end{tabular}
\end{table}

\paragraph{Censoring burden.} Table~\ref{tab:burden} gives the variance of
SA-AIPW relative to the same estimator with the delayed outcome always
observed. The ratios behave as Theorem~\ref{thm:burden} suggests. They grow
with the censoring in the final periods, from at most 1.09 when 9\% of
period-6 records are censored to between 1.43 and 2.73 when 65\% are. With
moderate or heavy censoring ($\kappa\ge2.5$) they are largest for the longest
exposure time, whose only overlap period is the most heavily censored one;
with $\kappa=4$ and $\sigma^2_Y=0.64$ the variance of $\hat\Psi(5)$ grows by a
factor of 2.7, while that of $\hat\Psi(1)$ grows by 45\%. Under heavy
censoring the loss is smaller when the surrogate explains more of the outcome
($\sigma^2_Y=0.16$): the ratio for $\Psi(5)$ falls from 2.73 to 2.34. With light
censoring the ratios are close to one under both noise levels. The variance
ratios for IPCW, which reweights the full outcome and is compared with the
full-data estimator without the history term, are between 1.2 and 2.3.
Coverage of SA-AIPW was between 0.92 and 0.97 in every Block~III cell
(Supplement~S4).

With misspecified working models (Block~IV, Supplement~S4), SA-AIPW remained
unbiased, with coverage between 0.93 and 0.96, when either the outcome
regression or the observation model was wrong, as
Theorem~\ref{thm:asymp}(iv) predicts; with both wrong its largest bias was
0.012. The history index does not affect this: its correction has mean zero
whatever the working models.

\paragraph{Implementation choices.} Supplement~S4 reruns SA-AIPW on the
same 1000 data sets as Table~\ref{tab:sim_main} under other specifications.
The history index lowers the standard deviations by 31\% to 47\% relative
to the same estimator without it; built from the surrogate histories alone,
which is the option when the observation of pre-crossover outcomes depends
on the wave and the weights cannot be fixed, it recovers most of this gain.
The estimator of an earlier version of this paper, which carried the period-1
cluster means of the surrogate and outcome as covariates in the regressions
instead, has standard deviations 10\% to 31\% larger than ours, and adding
the index to it recovers nothing, because the cluster-level columns have
absorbed, with cluster-level noise, what the index would; with 12 clusters
such columns made the estimator unstable (Supplement~S4). Fitting a separate
bridge for each exposure level, which with cluster-level covariates nearly
doubled the standard deviation of $\hat\Psi(5)$, is harmless without them.
None of these choices can bias the estimator
(Proposition~\ref{prop:dml_remainder}(b)), and the same specification was
used in Section~\ref{sec:sim:ept}.

\paragraph{Allocation of crossover times.} Block~V (Supplement~S4) behaved
as Proposition~\ref{prop:alloc} predicts: the augmented estimator had nearly
the same standard deviations under fixed waves and independent crossover
times (0.042 against 0.040 for $\Psi(1)$), whereas those of the unaugmented
estimator were 2.4 to 2.8 times larger under independent crossover times
than under fixed waves, and under fixed waves its sandwich overstated its
standard deviation by a factor of 2.5 to 3.0.

\paragraph{Effects that vary with calendar time.} In Block~VI (Supplement~S4,
Table~S10) the surrogate effect grows with calendar time, so $\Psi(e)$
averages period-specific effects over the overlap periods. SA-AIPW remained unbiased with coverage between 0.94 and 0.96 in all three
settings. The exposure-time mixed model, which assumes that effects do not
vary with calendar time, targets a model-weighted average instead. Its bias
for $\Psi(5)$ was $-0.036$ with the default censoring and about $-0.07$ when
censoring also depended strongly on the surrogate, and in the latter setting
its coverage for $\Psi(5)$ and $\bar\Psi$ fell to 0.86--0.88, lower with 48
clusters than with 30. The immediate-treatment mixed model covered
$\PsiIT$ in at most 14\% of replications.

\paragraph{Choosing the allocation.} To assess
Proposition~\ref{prop:design} in this design, we estimated its inputs
from one large data set ($J=900$) at each censoring level of Block~III and
searched over all allocations of 30 clusters with at least two per wave
(Supplement Table~S11). For $\Psi(5)$ the best allocation puts 11
clusters in each of the first and last waves and the minimum elsewhere,
reducing the variance by 55\%. Nearly all of this gain would remain without
censoring, because $\Psi(5)$ rests entirely on the contrast between these two
waves in period 6; censoring affects how the clusters are split between them.
With $\kappa=1$ the two waves are both lightly censored in period 6 (odds 0.09
and 0.15); with $\kappa=4$ the censoring odds there are 3.0 for the last wave
against 1.7 for the first, and \eqref{eq:neyman} tilts the split toward the
last wave, by less than one cluster with 30 clusters: moving one cluster from
the first to the last wave then changes the variance by 0.05\%, against 1.6\%
for the opposite move. The allocation that is best for $\bar\Psi$ is more
lopsided, with 12 or 13 clusters in the last wave, and its advantage over the
balanced design grows with censoring, from 25\% to 29\%. Both allocations
increase the variance of $\hat\Psi(1)$, by 39\% to 67\%. We checked
the predictions at $\kappa=4$ with 500 data sets per allocation (Supplement~S4).
The criterion models the estimator without the history term; its predicted
standard deviations were within 8\% of the simulated ones, and the sandwich
intervals for the targets in Table~S12 covered at rates between 0.94 and
0.98. The gains carry over to the estimator with the history term, whose
simulated variance of $\hat\Psi(5)$ under the allocation chosen for it was
0.38 times that under the balanced one.

\subsection{A Simulation Calibrated to the EPT Trial}
\label{sec:sim:ept}

The EPT design differs from the design above in several ways: 23 LHJs in
unequal waves, a binary outcome with prevalence near 5\%, a surrogate that
predicts it only weakly, and censoring tied to the wave rather than to the
calendar. We therefore built a data generator from the synthetic replication
data set distributed with our software. That data set takes from the trial
\citep{golden2015uptake} its rollout, 23 LHJs crossing over in waves of 6, 6,
6 and 5 at periods 2 to 5, the 162 records per LHJ-period and the 5\%
positivity assumed in its planning; the individual records, the 30-day
surrogate, the 12-month outcome and the censoring are synthetic, since the
trial's own outcomes were community-level positivity and incidence measured
at the end of each period.

\paragraph{Generator.} From the replication data we estimated, by logistic
regression, EPT uptake within 30 days given exposure time and period, 12-month
positivity given uptake, exposure time and period, and the censoring
probabilities by wave and period (85\% of the last wave's records in every
period and 38\% of the third wave's in periods 4 and 5). The replication data
contain almost no variation between LHJs, which is not plausible in a cluster
trial, so we added LHJ effects with assumed standard deviations. The effect of
exposure on positivity beyond uptake was set either to zero (scenario
\emph{none}, in which the program acts only through uptake) or to log-odds
ratios $-0.10$, $-0.25$ and $-0.40$ at exposure times 1 to 3 (scenario
\emph{growing}). Supplement~S5 gives the parameters, the true values and a
comparison of generated data with the replication data.

\paragraph{Analysis.} SA-AIPW used logistic outcome and observation models,
a pooled linear bridge and the history index with fixed weights; the
observation model includes period-by-wave indicators, so the influence
function need not be the EIF (Section~\ref{sec:tmle}), and the observation
probabilities depend on the wave, which is the case the fixed weights are
for (Section~\ref{sec:history}). Supplement~S5 gives the specifications. We used
$V=5$ folds stratified by wave and $t_{22}$ quantiles, and applied the same
estimator to the uncensored outcomes to measure what censoring costs. Each
setting had 1000 trials; the jackknife was computed for the first 500 trials
of the growing scenario under the trial's allocation.

\begin{table}[!t]
\centering
\caption{Simulation calibrated to the EPT trial ($J=23$, 1000 trials per
setting). (a) Trial allocation $(6,6,6,5)$: true effect, bias and Monte Carlo
standard deviation (SD) in percentage points; coverage of 95\% intervals with
the sandwich (s), HC3 (h), the jackknife (j; first 500 trials, growing
scenario) and the logit-scale interval~\eqref{eq:mover} with HC3 arm variances
(l); and the variance relative to the same estimator with the outcome always
observed. (b) Growing scenario: SD in percentage points predicted by
Proposition~\ref{prop:design} and obtained by simulation, under the trial's
allocation and under the allocation $(7,2,2,12)$ that minimizes the predicted
variance of $\hat\Psi(3)$; coverage and simulated variance ratio for the
latter.}
\label{tab:eptsim}
\small
\setlength{\tabcolsep}{4pt}
\begin{tabular}{llrrrccccc}
\multicolumn{10}{l}{(a) Trial allocation}\\
\toprule
Scenario & Target & Truth & Bias & SD & Cov.\ (s) & Cov.\ (h) & Cov.\ (j) & Cov.\ (l) & Var.\ ratio\\
\midrule
None & $\Psi(1)$ & $-$0.09 & $-$0.01 & 0.88 & 0.927 & 0.949 & -- & 0.952 & 2.40 \\
 & $\Psi(2)$ & $-$0.08 & 0.01 & 1.27 & 0.914 & 0.937 & -- & 0.946 & 2.78 \\
 & $\Psi(3)$ & $-$0.07 & 0.04 & 2.11 & 0.903 & 0.935 & -- & 0.958 & 3.32 \\
 & $\bar\Psi$ & $-$0.08 & 0.01 & 1.32 & 0.899 & 0.918 & -- & 0.957 & 3.66 \\
\addlinespace
Growing & $\Psi(1)$ & $-$0.51 & 0.01 & 0.89 & 0.932 & 0.949 & 0.964 & 0.961 & 2.59 \\
 & $\Psi(2)$ & $-$1.04 & 0.02 & 1.19 & 0.931 & 0.950 & 0.970 & 0.965 & 2.86 \\
 & $\Psi(3)$ & $-$1.50 & $-$0.02 & 2.09 & 0.890 & 0.915 & 0.938 & 0.953 & 3.85 \\
 & $\bar\Psi$ & $-$1.02 & 0.00 & 1.29 & 0.884 & 0.916 & 0.944 & 0.962 & 4.03 \\
\bottomrule
\end{tabular}

\bigskip
\begin{tabular}{lccccccc}
\multicolumn{8}{l}{(b) Allocation, growing scenario}\\
\toprule
 & \multicolumn{2}{c}{$(6,6,6,5)$} & \multicolumn{5}{c}{$(7,2,2,12)$}\\
\cmidrule(lr){2-3}\cmidrule(lr){4-8}
Target & Predicted & Simulated & Predicted & Simulated & Cov.\ (s) & Cov.\ (l) & Var.\ ratio\\
\midrule
$\Psi(1)$ & 0.99 & 0.89 & 1.14 & 0.99 & 0.951 & 0.969 & 1.23 \\
$\Psi(2)$ & 1.39 & 1.19 & 1.36 & 1.11 & 0.932 & 0.962 & 0.87 \\
$\Psi(3)$ & 2.20 & 2.09 & 1.57 & 1.40 & 0.930 & 0.964 & 0.45 \\
$\bar\Psi$ & 1.42 & 1.29 & 1.14 & 0.97 & 0.932 & 0.962 & 0.56 \\
\bottomrule
\end{tabular}
\end{table}

\paragraph{Results.} Table~\ref{tab:eptsim}(a) gives the results under the
trial's allocation. SA-AIPW was close to unbiased in both scenarios, with
biases of at most 0.04 percentage points against standard deviations of
0.88 to 2.11. The history index changed neither the estimates nor their
standard deviations by more than 2\%, as expected when the LHJ heterogeneity
is negligible (Supplement~S5). Censoring was more costly here than in the
design of Section~\ref{sec:sim:design}: the variance was 2.4 to 4.0 times
that of the same estimator applied to the complete data, and the ratio grew
with exposure time, as Theorem~\ref{thm:burden} predicts, because the
surrogate recovers little of what the last wave loses. Coverage was close to
nominal for $\Psi(1)$ and $\Psi(2)$ with the HC3 and jackknife standard
errors but fell to 0.88 to 0.90 for $\Psi(3)$ and $\bar\Psi$ with the
sandwich, to 0.92 to 0.94 with HC3 and to 0.94 with the jackknife on the 500
trials where it was computed. The sandwich understated the standard deviation
by 5\% to 13\%, which is the cost of fitting the loading of an uninformative index
to 23 clusters, and the standard errors were correlated with the estimates
($-0.46$ for $\hat\Psi(3)$ in the growing scenario): the contrast for
$\Psi(3)$ rests on the few observed positive outcomes among the five
last-wave LHJs in period 4, and a trial in which these happen to be rare
yields both a high estimate and a small standard error. The logit-scale
interval~\eqref{eq:mover} is intended for this situation. Computed from the
HC3 covariance of the arm means, it covered between 0.95 and 0.97 for every
target in every setting, at the price of intervals 37\% to 55\% wider than
the sandwich interval for $\Psi(3)$ and $\bar\Psi$ under the trial's
allocation and 18\% to 23\% wider elsewhere, most of which is the HC3
correction itself; Supplement~S5 gives the direction of the misses, the
version with the sandwich covariance of the arm means and an arcsine
version, prespecified as the alternative, which did less well.

Table~\ref{tab:eptsim}(b) turns to the design. With inputs estimated from one
pilot trial with 230 LHJs, Proposition~\ref{prop:design} gives $(7,2,2,12)$ as
the allocation with at least two LHJs per wave that minimizes the variance of
$\hat\Psi(3)$; the best allocation for $\bar\Psi$, $(6,3,2,12)$, is almost the
same. Under $(7,2,2,12)$ the simulated variance of $\hat\Psi(3)$ was
0.45 times that under the trial's allocation and that of $\hat{\bar\Psi}$
0.56 times, while the variance of $\hat\Psi(1)$ rose by 23\%.
Coverage also improved, to 0.93 to 0.95 with the sandwich and 0.95 to 0.97
with HC3 (Supplement~S5), as the period-4 contrast then has twelve unexposed LHJs instead
of five, even though 46\% of all records are censored, against 22\% under
the trial's allocation. In this setting, with a rare binary outcome, the
formula overstated the standard deviations by 5\% to 22\%, but it ranked
the allocations correctly and gave roughly the size of the gain (a variance
ratio of 0.51 for $\hat\Psi(3)$, against 0.45 simulated).


\FloatBarrier
\section{Discussion}
\label{sec:discussion}

Stepped-wedge trials with delayed outcomes raise two problems that are usually
handled separately, effects that change with exposure time and outcomes that
are administratively censored for late-crossing clusters. We have argued that
they belong together, because the information about long exposure times sits
in the periods that censoring removes. Theorem~\ref{thm:burden} quantifies
this, and Proposition~\ref{prop:design} makes it usable at the planning stage.
The known exposure law plays a central role: it removes the surrogate law and the bridge
regression from the bias (Proposition~\ref{prop:dml_remainder}), and it makes
the variance of the augmented estimator the same whether crossover times are
drawn independently or fixed in waves (Proposition~\ref{prop:alloc}).

\paragraph{Recommendations.} We suggest reporting exposure-time effects with
their overlap periods, together with any summary and its weights; for an
immediate-treatment effect, the weights of Proposition~\ref{prop:it} show
which exposure times it represents. We would use the augmented estimator
rather than the unaugmented one, with the sandwich or the fixed-partition
jackknife and $t_{J-1}$ quantiles; with data-adaptive learners we would use
the jackknife, with its partition held fixed across the refits. With a rare
binary outcome and a long-exposure contrast that rests on a few heavily
censored clusters, the logit-scale interval~\eqref{eq:mover} gave better
coverage than the symmetric one. When a trial is being planned and long
exposure times matter, it costs little to evaluate~\eqref{eq:design_var} for a
few candidate allocations; how far to move clusters toward the first and last
waves is a judgment about which effects matter most and how few clusters one
is willing to put in a wave.

\paragraph{Mixed models.} Under the moderate censoring of Block~I the
complete-case mixed model with exposure-time indicators was close to
unbiased, and with 30 clusters its standard deviations were within 7\% of
those of SA-AIPW for the first three exposure times and the design
information explains why: SA-AIPW attains the bound of
Proposition~\ref{prop:history}, and in the design of
Section~\ref{sec:sim:design} the mixed model fitted to the periods with
concurrent controls has that bound's precision to within 1\% for those
targets. What the model adds is structure: effects constant in calendar
time, which pool the period effect across exposure levels and matter for
the longest exposure times, and extrapolation to the final period, in which
every cluster is exposed; together they make it 50\% more precise for the
longest exposure time with 30 clusters, and its advantage grows as the
number of clusters falls, to about 30\% for the first three exposure times
with 12. When that structure holds, a mixed model is a reasonable choice and a
useful check. When effects varied with calendar time and censoring depended
strongly on the surrogate, its coverage fell to 0.86 with 48 clusters, while
SA-AIPW stayed near the nominal level.
SA-AIPW estimates an effect defined without an outcome model, on the
risk-difference scale for a binary outcome rather than as a cluster-specific
odds ratio; it needs randomization, Assumptions~\ref{ass:mar}
to~\ref{ass:pos_trt} and (C3), its intervals do not depend on a model for the correlation within clusters,
and it uses the surrogate without conditioning on it.

\paragraph{Sensitivity and extensions.} Assumption~\ref{ass:mar} cannot be
tested from the observed data \citep{gill1997coarsening}. Supplement~S2.15
gives the bias of $\hat\Psi(e)$ under a specified departure, which can be
estimated from the fitted nuisances and used for tipping-point analyses
\citep{scharfstein1999adjusting,vanderweele2017sensitivity}; because censoring
is heaviest in the overlap periods of long exposure times, the same departure
moves long-exposure effects the most. When the delayed outcome is never
observed in some overlap cells, Assumption~\ref{ass:pos_obs} fails and the
bridge has to be transported from earlier periods, which is a substantive
assumption. When the assignment law must be estimated, as in registries with
staggered adoption, $R_{E,k}$ of Proposition~\ref{prop:dml_remainder} returns,
and a rate condition on the bridge regression, or on the surrogate law for a
density plug-in, is needed in addition to (C3) \citep{roth2023efficient}.
Time-to-event outcomes with continuous administrative censoring,
small-sample refinements such as wild cluster bootstrap intervals for skewed
outcomes, and allocations that perform well over a range of plausible design
inputs are left for future work.

\section*{Supplementary Material}

The Supplementary Material contains a Berry--Esseen bound (Section~S1), the
derivation of the EIF and all proofs (Section~S2), implementation details
(Section~S3), additional simulation results (Section~S4) and the details of
the EPT-calibrated simulation (Section~S5). The methods are implemented in
the \texttt{R} package \texttt{swsurrdur}, which also contains the data and
scripts that reproduce every table and figure. 

\bibliographystyle{apalike}
\bibliography{references_2}

\end{document}


\maketitle

This supplement contains a Berry--Esseen bound on the coverage error of the
intervals (Section~S1); the semiparametric geometry and the proofs of all
results in the main text (Section~S2); implementation details
(Section~S3); additional simulation results (Section~S4); and the
generator and full results of the EPT-calibrated simulation, with an analysis
of the replication data set (Section~S5). Numbering in this document carries
the prefix ``S''; equations and results without the prefix refer to the main
text.

\paragraph{Notation.} A generic record is $O=(W,t,E,S,\Delta,\Delta Y)$ and
$P$ is the record-weighted (pooled) law across clusters. We write
$\QY(s,k,w,t)=\EE(Y\mid S=s,E=k,W=w,t,\Delta=1)$, $f_S(s\mid k,w,t)$ for the
conditional density of $S$, $\Qint(k,w,t)=\int\QY f_S\,d\mu$, $g_E(k\mid t)$ for
the design law~\eqref{M-eq:gE} and $g_\Delta(1\mid s,k,w,t)=\PP(\Delta=1\mid
s,k,w,t)$. For an exposure time $e\in\{1,\ldots,\emax\}$, $\pi_e=\PP(t\in\Te)$
and $\mu_k^{(e)}=\EE\{\Qint(k,W,t)\mid t\in\Te\}$, so that
$\Psi(e)=\mu_e^{(e)}-\mu_0^{(e)}$. The norm $\|\cdot\|_k$ is the $L^2$ norm
under the law of $(S,W,t)$ given $E=k$ and $t\in\Te$. A cluster is
$O_j=\{O_{ijt}\}$, with $M_j=n_jT$ records, $\bar M=\EE(M_j)$, crossover time
$\tau_j$ and covariates $\bW_j=\{W_{ijt}\}$. In algorithms,
$\hat m\equiv\QYhat$ and $\hat\pi\equiv\gDhat$. We write $\pm_k$ for the sign
$+$ when $k=e$ and $-$ when $k=0$.

\section{A Berry--Esseen Bound for the Coverage Error}
\label{supp:be}

Fix a contrast vector $a$ and write $\psi=a^{\top}\Psi$,
$\hat\psi=a^{\top}\hat\Psi$ and $\varphi^a_j=a^{\top}\phij$. By
Theorem~\ref{M-thm:asymp}(i), $\hat\psi-\psi=J^{-1}\sum_{j}\varphi^a_j+\rho_J$,
where $\rho_J=o_P(J^{-1/2})$ collects the empirical-process term, the
remainder and the error from the ratio normalization. Let
$\sigma^2=\Var(\varphi^a_j)$, $\rho_3=\EE|\varphi^a_j|^3$, and let $\widehat V$
be a variance estimator with $\hat\sigma_J^2=J\widehat V$. For a critical
value $c>0$, consider the interval $\hat\psi\pm c\,\widehat V^{1/2}$.

\begin{theorem}[Coverage error]
\label{thm:be}
Under conditions (C1)--(C2) of Theorem~\ref{M-thm:asymp}, for every
$\delta_1\in(0,1)$ and $\delta_2>0$,
\begin{align}
  \Bigl|\PP\bigl(\psi\in\hat\psi\pm c\,\widehat V^{1/2}\bigr)
  -\{2\Phi(c)-1\}\Bigr|
  \;\le\;&
  \frac{2C_{\mathrm{BE}}\,\rho_3}{\sigma^3\sqrt J}
  +\sqrt{2/\pi}\,(c\,\delta_1+\delta_2)
  \notag\\
  &+\PP\bigl(|\hat\sigma_J/\sigma-1|>\delta_1\bigr)
  +\PP\bigl(\sqrt J|\rho_J|>\sigma\delta_2\bigr),
  \label{eq:be_bound}
\end{align}
where $C_{\mathrm{BE}}<0.4748$ is the Berry--Esseen constant for
i.i.d.\ summands \citep{Shevtsova2011absolute}. For $c=t_{J-1,\alpha/2}$,
$2\Phi(c)-1=1-\alpha+\lambda_t/(J-1)+O(J^{-2})$ with
$\lambda_t=\Phi'(z_{\alpha/2})\,z_{\alpha/2}(z_{\alpha/2}^2+1)/2>0$.
\end{theorem}

\begin{proof}
Let $Z_J=J^{-1/2}\sum_j\varphi^a_j/\sigma$, $e_J=\sqrt J\rho_J/\sigma$ and
$D_J=\hat\sigma_J/\sigma$, so that $\psi$ lies in the interval if and only if
$|Z_J+e_J|\le c\,D_J$. On the event
$\mathcal E=\{|D_J-1|\le\delta_1,\ |e_J|\le\delta_2\}$,
\[
  \{|Z_J|\le c(1-\delta_1)-\delta_2\}
  \subseteq\{|Z_J+e_J|\le cD_J\}
  \subseteq\{|Z_J|\le c(1+\delta_1)+\delta_2\}.
\]
For any $c'>0$, $|\PP(|Z_J|\le c')-\{2\Phi(c')-1\}|\le
2\sup_x|\PP(Z_J\le x)-\Phi(x)|\le 2C_{\mathrm{BE}}\rho_3/(\sigma^3\sqrt J)$ by
the Berry--Esseen inequality applied to the i.i.d.\ $\varphi^a_j$, and
$|\{2\Phi(c')-1\}-\{2\Phi(c)-1\}|\le\sqrt{2/\pi}\,|c'-c|$. Combining these
with $\PP(\mathcal E^c)\le\PP(|D_J-1|>\delta_1)+\PP(|e_J|>\delta_2)$
gives~\eqref{eq:be_bound}. The expansion for $c=t_{J-1,\alpha/2}$ follows from
the Cornish--Fisher expansion
$t_{\nu,\alpha/2}=z_{\alpha/2}+z_{\alpha/2}(z_{\alpha/2}^2+1)/(4\nu)
+O(\nu^{-2})$ \citep[Chapter~2]{hall1992bootstrap}.
\end{proof}

\begin{remark}[Use of the bound]
\label{rem:be_use}
Under (C2) the cluster influence functions are bounded, $|\varphi^a_j|\le B$,
so $\rho_3/\sigma^3\le B/\sigma$ provided $\sigma>0$. The normal-approximation
term is therefore of order $J^{-1/2}$. The two probability terms
in~\eqref{eq:be_bound} tend to zero under the conditions of
Theorems~\ref{M-thm:asymp} and~\ref{M-thm:jack}, but those conditions
control the variance-estimation error and the remainder only in probability,
so~\eqref{eq:be_bound} gives no rate for the total coverage error without
further assumptions. For a variance estimator with known nuisances,
Bernstein's inequality gives $\PP(|\hat\sigma_J/\sigma-1|>\delta_1)\to0$ with
$\delta_1$ of order $(\log J/J)^{1/2}$; for the remainder the relevant
quantity is $\sqrt J\,\|\hat m-\QY^0\|_k\|\hat\pi-g^0_\Delta\|_k/\sigma$,
whose tail behaviour depends on the learners. The bound is conservative and
is largest for long exposure times, whose influence functions have the
largest $B/\sigma$, and because the margin $\lambda_t/(J-1)$ from the
$t_{J-1}$ critical value is of smaller order than the normal-approximation
term, it does not settle the small-sample calibration, which we assess by
simulation.
\end{remark}

\section{Semiparametric Geometry and Proofs}
\label{app:eic}

We follow the geometric exposition of \citet{nabi2026semiparametric}: the
model determines the admissible directions of movement (the tangent space),
the parameter determines which of them are informative, and the efficient
influence function (EIF) is the gradient that lies in the tangent space.
Standard references are \citet{bickel1993efficient},
\citet{tsiatis2006semiparametric} and \citet[Chapter~25]{vaart1998asymptotic}.

Table~\ref{tab:dictionary} lists the correspondence between ordinary calculus,
calculus on a space of distributions and the objects of the main text, and
Figure~\ref{fig:dag} draws the structural model of
Section~\ref{M-sec:design}.

\begin{table}[!htbp]
\centering
\caption{Correspondence between ordinary calculus, calculus on a space of
  distributions \citep{nabi2026semiparametric}, and the objects of this
  paper.}
\label{tab:dictionary}
\small
\renewcommand{\arraystretch}{1.12}
\begin{tabularx}{\textwidth}{>{\raggedright\arraybackslash}p{.21\textwidth}
                              >{\raggedright\arraybackslash}p{.27\textwidth}X}
\toprule
Calculus & Space of distributions & This paper\\
\midrule
point $x$ & distribution $P$ & $P_0$, the law of the trial data\\
function $f(x)$ & functional $\Psi(P)$ & exposure-time effect~\eqref{M-eq:identification}\\
curve and velocity & path $P_\varepsilon$ and score $\dot\ell$ & perturbation of one likelihood factor\\
admissible directions & tangent space $\calT$ & $\calT_{W,t}\oplus\calT_S\oplus\calT_\Delta\oplus\calT_Y$, with $\calT_E=\{0\}$\\
directions that leave $f$ unchanged & nuisance tangent space & includes $\calT_\Delta$ (Lemma~\ref{lem:tdelta})\\
gradient $\nabla f$ & influence function & $D^*_e$, the EIF (projected gradient)\\
squared length of the gradient & efficiency bound & full-data bound plus censoring burden (Theorem~\ref{M-thm:burden})\\
Taylor remainder & von Mises remainder & $R_{E}+R_{\Delta}$ (Proposition~\ref{M-prop:dml_remainder})\\
\bottomrule
\end{tabularx}
\end{table}

\begin{figure}[!htbp]
\centering
\begin{tikzpicture}[
  >=Stealth,
  endog/.style={circle, draw=black, thick, minimum size=1.05cm, inner sep=1pt,
                font=\small},
  exog/.style={circle, draw=black, dashed, minimum size=0.9cm, inner sep=1pt,
               font=\small},
  arr/.style={->, semithick},
  earr/.style={->, thin, gray},
  icc/.style={->, semithick, dashed},
]
\node[exog] (UW)  at (0.0, 4.4)  {$U_W$};
\node[exog] (UA)  at (3.0, 4.4)  {$U_\tau$};
\node[exog] (US)  at (6.0, 4.4)  {$U_S$};
\node[exog] (UD)  at (9.0, 4.4)  {$U_\Delta$};
\node[exog] (UY)  at (11.6,0.5)  {$U_Y$};
\node[exog] (t)   at (0.0, 0.0)  {$t$};
\node[exog] (bj)  at (0.0,-1.8)  {$b_j$};
\node[endog] (W)  at (0.0, 2.5)  {$W_{ijt}$};
\node[endog] (A)  at (3.0, 2.5)  {$E_{jt}$};
\node[endog] (S)  at (6.0, 2.5)  {$S_{ijt}$};
\node[endog] (D)  at (9.0, 2.5)  {$\Delta_{ijt}$};
\node[endog] (Y)  at (9.0, 0.5)  {$Y_{ijt}$};
\draw[earr] (UW)--(W);
\draw[earr] (UA)--(A);
\draw[earr] (US)--(S);
\draw[earr] (UD)--(D);
\draw[earr] (UY)--(Y);
\draw[arr] (W)--(S);
\draw[arr] (W) to[out=22,in=158] (D);
\draw[arr] (W) to[out=-25,in=185] (Y);
\draw[arr] (t) to[out=15,in=225] (A);
\draw[arr] (t) to[out=18,in=218] (S);
\draw[arr] (t) to[out=20,in=210] (D);
\draw[arr] (t)--(Y);
\draw[arr] (A)--(S);
\draw[arr] (A) to[out=18,in=162] (D);
\draw[arr] (A) to[out=-28,in=148] (Y);
\draw[arr] (S)--(D);
\draw[arr] (S) to[out=-28,in=122] (Y);
\draw[icc] (bj) to[out=35,in=242] (S);
\draw[icc] (bj) to[out=55,in=212] (Y);
\node[draw=black, rectangle, minimum size=0pt, inner sep=3pt,
      font=\footnotesize\itshape] (nde) at (10.9,-0.9)
  {\shortstack{no edge\\$Y\to\Delta$}};
\draw[->,dotted,thin] (nde.north)--(9.45,1.45);
\node[font=\footnotesize\itshape] at (5.5,-3.0)
  {Time ordering:\quad $W\;\to\; E\;\to\; S\;\to\;\Delta\;\to\; Y$};
\end{tikzpicture}
\caption{Directed acyclic graph for the stepped-wedge structural model.
  Solid circles are endogenous variables; dashed circles are exogenous noise
  and design variables; dashed arrows from $b_j$ mark the source of
  within-cluster correlation. Exposure time is determined by calendar time and
  the randomized crossover time, so there is no edge $W\to E$. The
  absence of an edge $\Yijt\to\Dijt$ encodes Assumption~\ref{M-ass:mar}. The
  dependence of $\Dijt$ on $\tau_j$ is omitted for readability.}
\label{fig:dag}
\end{figure}

\subsection{The Structural Model, Paths and Tangent Space}
\label{app:eic:tangent}

Let $\mathbf U=(U_W,U_\tau,b_j,U_S,U_\Delta,U_Y)$ collect exogenous variables:
the crossover variable $U_{\tau,j}$ and the cluster effect $b_j$ are
cluster-level and i.i.d.\ across clusters, and the remaining components are
individual-level noise, mutually independent and independent of the
cluster-level terms. The endogenous variables of the structural model of
Section~\ref{M-sec:design} are generated in their causal order:
\begin{align}
  \Wijt &= f_W(U_{W,ijt}),
    &&\text{(baseline covariates)}\label{eq:scm_W}\\[2pt]
  \Ejt  &= (t-\tau_j+1)\,I(t\ge\tau_j),\quad\tau_j=f_\tau(U_{\tau,j}),
    &&\text{(randomized exposure time)}\label{eq:scm_A}\\[2pt]
  \Sijt &= f_S(\Wijt,\,\Ejt,\,b_j,\,t,\,U_{S,ijt}),
    &&\text{(short-term surrogate)}\label{eq:scm_S}\\[2pt]
  \Dijt &= f_\Delta(\Sijt,\,\Ejt,\,\Wijt,\,t,\,\tau_j,\,U_{\Delta,ijt}),
    &&\text{(observation indicator)}\label{eq:scm_D}\\[2pt]
  \Yijt &= f_Y(\Wijt,\,\Ejt,\,\Sijt,\,b_j,\,t,\,U_{Y,ijt}).
    &&\text{(delayed primary outcome)}\label{eq:scm_Y}
\end{align}
The cluster effect $b_j$ enters $f_S$ and $f_Y$, the outcome does not enter
$f_\Delta$, and $U_\Delta$ does not enter $f_Y$.

The record-level density factorizes as in~\eqref{M-eq:factorization}. The
model $\mathcal M$ leaves $p(w,t)$, $f_S$, $g_\Delta$ and $q_Y$ unrestricted
and fixes $g_E(\cdot\mid t)$ at the design law. Regular paths through $P_0$
can be built factor by factor \citep[Section~8]{nabi2026semiparametric}:
\begin{itemize}[leftmargin=1.5em]
  \item $p_\varepsilon(w,t)\propto p_0(w,t)\exp\{\varepsilon s_W(w,t)\}$, with
    score $s_W(W,t)$, $\EE s_W=0$;
  \item $f_{S,\varepsilon}(s\mid k,w,t)\propto f_S(s\mid k,w,t)
    \exp\{\varepsilon s_S(s,k,w,t)\}$, with score $s_S(S,E,W,t)$,
    $\EE(s_S\mid E,W,t)=0$;
  \item $\logit g_{\Delta,\varepsilon}(1\mid s,k,w,t)=\logit g_\Delta(1\mid
    s,k,w,t)+\varepsilon u(s,k,w,t)$, with score
    $\{\Delta-g_\Delta(1\mid S,E,W,t)\}u(S,E,W,t)$;
  \item $q_{Y,\varepsilon}(y\mid s,k,w,t)\propto q_Y(y\mid s,k,w,t)
    \exp\{\varepsilon s_Y(y,s,k,w,t)\}$, with score $\Delta\,s_Y(Y,S,E,W,t)$,
    $\EE(s_Y\mid S,E,W,t,\Delta=1)=0$.
\end{itemize}
No path perturbs $g_E$. Each score has conditional mean zero given the
variables that precede its factor, so scores from different factors are
orthogonal and~\eqref{M-eq:tangent} holds. In the nonparametric model, which
leaves the law of $E$ given $(W,t)$ free, the tangent space also contains
$\calT_E=\{u(E,W,t):\EE(u\mid W,t)=0\}$ and equals $L^2_0(P_0)$. The
restriction removes $\calT_E$; it does not make $E$ degenerate, since
$0<g_E(e\mid t)<1$ for $t\in\Te$.

\subsection{Proofs of Theorem~\ref{M-thm:id} and Proposition~\ref{M-prop:it}}
\label{app:id}

\paragraph{Theorem~\ref{M-thm:id}.} Fix $k\in\{0,e\}$ and $t\in\Te$, so that
$g_E(k\mid t)\ge c_E>0$ by Assumption~\ref{M-ass:pos_trt}. By
Assumption~\ref{M-ass:seqrand} and then Assumption~\ref{M-ass:consistency},
\[
  \begin{aligned}
  \EE\{Y(k)\mid W,t\}&=\EE\{Y(k)\mid E=k,W,t\}=\EE(Y\mid E=k,W,t)\\
  &=\EE\{\EE(Y\mid S,E=k,W,t)\mid E=k,W,t\}.
  \end{aligned}
\]
By Assumption~\ref{M-ass:mar}, $\Delta$ is independent of $Y$ given
$(S,E,W,t)$, and by Assumption~\ref{M-ass:pos_obs} the event $\{\Delta=1\}$
has positive probability given $(S,E=k,W,t)$, so
$\EE(Y\mid S,E=k,W,t)=\EE(Y\mid S,E=k,W,t,\Delta=1)=\QY(S,k,W,t)$. Hence
$\EE\{Y(k)\mid W,t\}=\Qint(k,W,t)$, and averaging the difference between
$k=e$ and $k=0$ over the law of $(W,t)$ given $t\in\Te$
gives~\eqref{M-eq:identification}. \qed

\paragraph{Proposition~\ref{M-prop:it}.} Because $\tau_j$ is randomized
independently of the potential outcomes and covariates,
$\PP(E=e\mid A=1,W,t)=g_E(e\mid t)/\sum_{k\ge1}g_E(k\mid t)=w_t(e)$, and by
Assumptions~\ref{M-ass:consistency}--\ref{M-ass:seqrand},
\[
  \EE(Y\mid A=1,W,t)=\sum_{e=1}^{t-1}w_t(e)\,\EE\{Y(e)\mid W,t\},\qquad
  \EE(Y\mid A=0,W,t)=\EE\{Y(0)\mid W,t\}.
\]
Since $\sum_ew_t(e)=1$, the difference is
$\sum_ew_t(e)\EE\{Y(e)-Y(0)\mid W,t\}$; averaging over $W$ given $t$ and then
over $t\in\Tov$ gives~\eqref{M-eq:psi_it}, and the bridge representation
follows from the proof of Theorem~\ref{M-thm:id}. If $\delta_t(e)=\delta(e)$
for all $t$, then $\Psi(e)=\delta(e)$ because $\Psi(e)$ averages
$\delta_t(e)$ over $t\in\Te$, and exchanging the sums gives
$\PsiIT=\sum_e\delta(e)\,\EE\{w_t(e)I(t\ge e+1)\mid t\in\Tov\}$. With equal
numbers of records per period, $t$ is uniform on $\Tov=\{2,\ldots,T-1\}$
given $t\in\Tov$. Write $G(t)=\sum_{s\le t}p_\tau(s)=g_A(1\mid t)$, so that
$w_t(e)=p_\tau(t-e+1)/G(t)$ and
$(T-2)\,\omega^{\mathrm{IT}}_e=\sum_{t=e+1}^{T-1}p_\tau(t-e+1)/G(t)$.
Substituting $s=t-e+1$ in this sum and $s=t-e$ in the corresponding sum for
$e+1$,
\[
  (T-2)\bigl(\omega^{\mathrm{IT}}_e-\omega^{\mathrm{IT}}_{e+1}\bigr)
  =\sum_{s=2}^{T-1-e}p_\tau(s)\Bigl\{\frac{1}{G(s+e-1)}-\frac{1}{G(s+e)}\Bigr\}
   +\frac{p_\tau(T-e)}{G(T-1)}\;\ge\;0,
\]
because $G$ is nondecreasing. Under uniform allocation,
$g_E(e\mid t)=1/(T-1)$ for $1\le e\le t-1$, so $w_t(e)=1/(t-1)$ and
$\omega^{\mathrm{IT}}_e=(T-2)^{-1}\sum_{t=e+1}^{T-1}(t-1)^{-1}$. \qed

\paragraph{The EIF of $\PsiIT$ and our implementation.} Since $w_t(e)/g_E(e\mid t)=1/g_A(1\mid t)$ with
$g_A(1\mid t)=\sum_{e\ge1}g_E(e\mid t)$, the EIF of $\PsiIT$ is built from
$\sum_ew_t(e)\phi_e-\phi_0$, whose exposed-arm part is
$\sum_ew_t(e)\Qint(e,W,t)+\{I(A=1)/g_A(1\mid t)\}\{\QY-\Qint(E,W,t)\}$ plus the
outcome residual term. Our implementation uses the simpler
$\sum_ew_t(e)\Qint(e,W,t)$ in place of $\Qint(E,W,t)$ in the middle term. This
leaves the estimator consistent and asymptotically normal, but adds a
mean-zero function of $(E,W,t)$ to its influence function, as for the
unaugmented estimator of Remark~\ref{M-rem:unaug}; the efficiency and
allocation results of Section~\ref{M-sec:asymptotics} therefore apply to
$\PsiIT$ only in the first form.

\subsection{The Efficient Influence Function}
\label{app:eic:derivation}

\begin{proof}[Proof of Proposition~\ref{M-prop:eif}]
Fix $k\in\{0,e\}$ and write $\mu_k=\mu_k^{(e)}=\EE\{\inde\Qint(k,W,t)\}/\pi_e$.
We compute the derivative of $\mu_k$ along each type of path and represent it
as an inner product. Two identities are used repeatedly:
$\EE(\Delta X\mid S,E,W,t)=g_\Delta\,\EE(X\mid S,E,W,t,\Delta=1)$, and, since
$g_E$ does not depend on $W$,
$\EE\{I(E=k)X/g_E(k\mid t)\mid W,t\}=\EE(X\mid E=k,W,t)$.

\emph{(a) Outcome factor.} Along a path with score $\Delta s_Y$,
$\partial_\varepsilon\QY=\EE\{(Y-\QY)s_Y\mid S,E,W,t,\Delta=1\}$, so
\[
  \partial_\varepsilon\mu_k=\frac{1}{\pi_e}\EE\Bigl[\inde\int
  \EE\{(Y-\QY)s_Y\mid s,k,W,t,\Delta=1\}f_S(s\mid k,W,t)\,d\mu(s)\Bigr],
\]
and the two identities show that this equals
$\EE[\inde\,I(E=k)\Delta(Y-\QY)/\{g_Eg_\Delta\}\cdot\Delta s_Y]/\pi_e$.

\emph{(b) Surrogate factor.} Along a path with score $s_S$,
$\partial_\varepsilon\Qint(k,W,t)=\EE\{\QY(S,k,W,t)s_S\mid E=k,W,t\}
=\EE\{(\QY-\Qint)s_S\mid E=k,W,t\}$, using $\EE(s_S\mid E,W,t)=0$. The second
identity gives
$\partial_\varepsilon\mu_k=\EE[\inde\,I(E=k)(\QY-\Qint)/g_E\cdot s_S]/\pi_e$.

\emph{(c) Covariate and period factor.} Along a path with score $s_W$,
$\partial_\varepsilon\mu_k=\EE\{\inde(\Qint(k,W,t)-\mu_k)s_W\}/\pi_e$, where
the centering comes from differentiating the normalization $\pi_e$.

\emph{(d) Observation factor.} The derivative is zero
(Lemma~\ref{lem:tdelta}).

The three representers are the three terms of $\phi_k$ in~\eqref{M-eq:phi},
multiplied by $\inde/\pi_e$ and centered, and they are mutually orthogonal: the
outcome term has mean zero given $(S,E,W,t,\Delta)$; the surrogate term is a
function of $(S,E,W,t)$ with mean zero given $(E,W,t)$; the covariate term is
a function of $(W,t)$. Taking the difference between $k=e$ and $k=0$ shows
that $D^*_e$ represents the derivative of $\Psi(e)$ along every score in
$\calT$. It also lies in $\calT$, each term in the corresponding subspace, so
it is the canonical gradient. Finally, for $u\in\calT_E$, the outcome and
surrogate terms have mean zero given $(E,W,t)$, and the covariate term is a
function of $(W,t)$ while $\EE(u\mid W,t)=0$; hence $\EE(D^*_e u)=0$. In the
nonparametric model the derivative along $u$ is zero because $\Psi(e)$ does
not involve $g_E$, so $D^*_e$ is a gradient there too, and since it lies in
that model's tangent space it is the EIF there
\citep[Section~8.2]{nabi2026semiparametric}.
\end{proof}

At the record level, $D^*_e$ has the form of the EIF of
\citet[Section~2]{kallus2025surrogates} for treatment effects with outcomes
missing at random given surrogates, specialized to a multi-valued treatment
with a known assignment law and restricted to overlap periods.

\subsection{Censoring Directions and the Role of Assumption~\ref{M-ass:mar}}
\label{app:eic:lemma1}

\begin{lemma}[Censoring directions are nuisance directions]
\label{lem:tdelta}
Under Assumptions~\ref{M-ass:consistency}--\ref{M-ass:pos_trt}, $\Psi(e)$
depends on $P$ only through $p(w,t)$, $f_S$ and $\QY$. Hence (i)
$d\Psi(P_\varepsilon)/d\varepsilon|_{\varepsilon=0}=0$ for every path with
score in $\calT_\Delta$, and (ii) $D^*_e\perp\calT_\Delta$: the EIF has no
$(\Delta-g_\Delta)$-score component, and $g_\Delta$ enters $D^*_e$ only
through the weight $\Delta/g_\Delta$ on the residual $Y-\QY$.
\end{lemma}

\begin{proof}
By the proof of Theorem~\ref{M-thm:id}, $\Psi(e)$ is a functional of
$p(w,t)$, $f_S$ and $\QY$ only. A path with score in $\calT_\Delta$ perturbs
$g_\Delta$ and leaves these three factors unchanged, so the derivative is
zero, which is (i). For (ii), $D^*_e\perp\calT_\Delta$ follows from (i) and
the fact that $D^*_e\in\calT$ is the canonical gradient; directly, the
outcome term of $D^*_e$ has mean zero given $(S,E,W,t,\Delta)$, and the other
two terms are functions of $(S,E,W,t)$, while every element of
$\calT_\Delta$ has mean zero given $(S,E,W,t)$.
\end{proof}

Assumption~\ref{M-ass:mar} is used only through identification. Without it,
$\EE(Y\mid S,E=k,W,t)=\QY+(1-g_\Delta)\gamma_k$ with
$\gamma_k=\EE(Y\mid S,E=k,W,t,\Delta=0)-\QY$, and the causal parameter depends
on $\gamma_k$, which is not a functional of the observed-data law. Because
missingness at random places no restriction on the observed-data law
\citep{gill1997coarsening}, the tangent space keeps $\calT_\Delta$:
Lemma~\ref{lem:tdelta} is a statement about the parameter (the
direction is uninformative), not about the model (the direction is
admissible). The weight $\Delta/g_\Delta$ in $D^*_e$ is therefore not a
projection onto $\calT_\Delta$; it transports the residual of the observed
subcohort to all records.

\subsection{Proof of Lemma~\ref{M-lem:summation} and the Cluster-Level Geometry}
\label{app:eic:cluster}

\paragraph{Gradient.} Let $P_J$ denote the law of $O_j$ and consider the
cluster-level model $\mathcal M_J$ that leaves $P_J$ unrestricted except that
the crossover time follows the known allocation law, independently of the
cluster size $n_j$ and the covariates $\bW_j$. The true law $P_0$ also satisfies
the structural model of Section~\ref{M-sec:design}, which we use below to
compute conditional moments at $P_0$. For a regular path $P_{J,\varepsilon}$ in $\mathcal M_J$
with score $s(O_j)$, the induced record-weighted law
$\bar P_\varepsilon(B)=\EE_\varepsilon\{\sum_{i,t}I(O_{ijt}\in B)\}/
\EE_\varepsilon(M_j)$ has derivative
$\dot{\bar P}(B)=\bar M^{-1}\EE\{s(O_j)\sum_{i,t}I(O_{ijt}\in B)\}
-\bar P(B)\,\EE\{s(O_j)M_j\}/\bar M$. The induced path keeps the pooled
exposure law $\bar P(E=k\mid W,t)=g_E(k\mid t)$ fixed, so it is a path in the
record-level model $\mathcal M$, and since $D^*_e$ is a gradient of $\Psi(e)$
at $\bar P$ with $\int D^*_e\,d\bar P=0$,
\[
  \frac{d}{d\varepsilon}\Psi(e;\bar P_\varepsilon)\Big|_{\varepsilon=0}
  =\int D^*_e\,d\dot{\bar P}
  =\frac{1}{\bar M}\EE\Bigl\{s(O_j)\sum_{i=1}^{n_j}\sum_{t=1}^T D^*_e(O_{ijt})\Bigr\}
  =\EE\{s(O_j)\,\EICj(e)\}.
\]
No assumption on the within-cluster dependence is used: the sum absorbs the
correlation within a period and across periods.

\paragraph{Sum versus average.} The quantity
$M_j^{-1}\sum_{i,t}D^*_e(O_{ijt})$ is the gradient of the cluster-weighted
functional obtained by replacing $\bar P$ with
$\tilde P(B)=\EE\{M_j^{-1}\sum_{i,t}I(O_{ijt}\in B)\}$. When $M_j$ is constant
the two coincide; otherwise they target different parameters, and a sandwich
built from the cluster average does not estimate the variance of the
record-weighted estimator.

\paragraph{Orthogonality to assignment directions.} The assignment directions
of the cluster-level nonparametric model are
$\calT_{E,J}=\{u(\tau_j,\bW_j):\EE(u\mid\bW_j)=0\}$, and the tangent space of
$\mathcal M_J$ is $\calT_J=L^2_0(P_J)\ominus\calT_{E,J}$. Under the structural
model~\eqref{eq:scm_W}--\eqref{eq:scm_Y}, the surrogate and outcome of a
record depend on the crossover time only through its current exposure time
(no anticipation or carryover beyond exposure time) and do not depend on
other records' covariates (no interference), and $\tau_j$ is independent of
$(b_j,\bW_j)$ and of the exogenous noise. Hence, on $\{\Ejt=k\}$,
\begin{equation}
  \label{eq:cond_tau}
  \EE\{\QY(S_{ijt},k,W_{ijt},t)\mid\tau_j,\bW_j\}=\Qint(k,W_{ijt},t),
  \qquad
  \EE\{\Dijt(\Yijt-\QY)\mid S_{ijt},\tau_j,\bW_j\}=0,
\end{equation}
the second because $\Dijt$ is independent of $\Yijt$ given
$(\Sijt,\tau_j,\bW_j)$ and $\Yijt$ is independent of $\tau_j$ given
$(\Sijt,\Ejt,W_{ijt},t)$. Therefore
\begin{equation}
  \label{eq:eif_cond}
  \EE\{D^*_e(O_{ijt})\mid\tau_j,\bW_j\}
  =\frac{\inde}{\pi_e}\bigl\{\Qint(e,W_{ijt},t)-\Qint(0,W_{ijt},t)-\Psi(e)\bigr\},
\end{equation}
a function of $\bW_j$ alone, so $\EE\{\EICj(e)\,u\}=0$ for every
$u\in\calT_{E,J}$. Thus $\EICj(e)\in\calT_J$ is the canonical gradient of
$\Psi(e)$ in $\mathcal M_J$ at $P_0$.

\paragraph{What $\mathcal M_J$ leaves out.} The structural model implies
restrictions on $P_J$ that $\mathcal M_J$ does not impose: without
anticipation, a cluster's data before it crosses over have the same law
whatever its crossover time. These restrictions are testable, so the tangent
space of the structural model is a proper subspace of $\calT_J$ and its
efficiency bound can be smaller. Period-1 cluster summaries in $W$ use part of
this information. The model that imposes all of these restrictions is
$\mathcal M_S$ of Proposition~\ref{M-prop:history}, whose bound is derived in
Section~\ref{app:history}; the same restrictions underlie the estimators of
\citet{roth2023efficient}.

\subsection{Proof of Proposition~\ref{M-prop:dml_remainder}}
\label{app:eic:dml}

Fix $k$ and condition on $t\in\Te$ throughout. Write $e_Q=\QYhat-\QY^0$,
$r_E=g^0_E(k\mid t)/\gEhat(k\mid t)$ and
$r_\Delta=g^0_\Delta(1\mid S,k,W,t)/\gDhat(1\mid S,k,W,t)$. By
Assumption~\ref{M-ass:mar},
$\EE\{\Delta(Y-\QYhat)\mid S,E=k,W,t\}=g^0_\Delta(\QY^0-\QYhat)$. Taking
expectations of the three terms of $\phi_k(O;\hat\eta)$, first given
$(S,E,W,t)$ and then given $(W,t)$,
\begin{align*}
  \EE\{\phi_k(O;\hat\eta)\mid W,t\}
  &=\Qinthat
   +r_E\Bigl\{\int\QYhat f^0_S\,d\mu-\Qinthat\Bigr\}
   -r_E\int r_\Delta\,e_Q\,f^0_S\,d\mu\\
  &=(1-r_E)\Qinthat+r_E\int\QY^0 f^0_S\,d\mu
   +r_E\int(1-r_\Delta)\,e_Q\,f^0_S\,d\mu .
\end{align*}
Since $\int\QY^0f^0_S\,d\mu=\Qint^0$, subtracting $\Qint^0$ and averaging over
$(W,t)$ gives
$B_k(\hat\eta)=\EE\{(1-r_E)(\Qinthat-\Qint^0)\}
+\EE\{r_E\int(1-r_\Delta)e_Qf^0_S\,d\mu\}=R_{E,k}+R_{\Delta,k}$.

Part~(a): the decomposition~\eqref{M-eq:RSY} is algebra, and $\Qinthat$ enters
$B_k$ only through $R_{E,k}$. The cancellation happens in the second term
above: the $\calT_S$ component of $\phi_k$ is evaluated at $\QYhat$, so its
expectation $r_E\{\int\QYhat f^0_S\,d\mu-\Qinthat\}$ removes the entire
plug-in error of $\Qinthat$, including $R_{SY}$, up to the factor $r_E$. A
correction evaluated at $\QY^0$ would remove only the first-order parts and
leave $R_{SY}$. Part~(b): if $\gEhat=g^0_E$ then $r_E=1$ and $R_{E,k}=0$;
moreover $|1-r_\Delta|\le|\gDhat-g^0_\Delta|/c$, and the Cauchy--Schwarz
inequality, together with the bounded ratio between the laws of $(W,t)$ given
$\{t\in\Te\}$ and given $\{E=k,t\in\Te\}$ implied by
Assumption~\ref{M-ass:pos_trt}, gives the bound. Part~(c): with $\gEhat$
estimated, $|R_{E,k}|\le c^{-1}\|\gEhat-g^0_E\|\,\|\Qinthat-\Qint^0\|$ by the
Cauchy--Schwarz inequality. For the density plug-in,
$\Qinthat-\Qint^0=\int\QYhat(\hat f_S-f^0_S)\,d\mu
+\int(\QYhat-\QY^0)f^0_S\,d\mu$, so a rate for $\hat f_S$, together with one
for $\QYhat$, is sufficient; it is not necessary in general, because errors
in $\hat f_S$ orthogonal to $\QYhat$ do not move $\Qinthat$. The algebra is
unchanged when the assignment law depends on $W$. \qed

\paragraph{Numerical check.} Table~\ref{tab:rsy_check} verifies the
decomposition in an i.i.d.\ design with a single exposure level, which
suffices because the algebra above is level by level. With $N=4\times10^6$
records, $W\sim\mathcal N(0,1)$, a known assignment probability $0.5$, a
surrogate $S\mid E,W\sim\mathcal N(0.8E+0.5W,1)$, observation probability
$\logit g_\Delta=0.5+0.8S-0.3W$, and $Y=E+0.7S+0.3S^2+0.4W+\epsilon$ with
$\epsilon\sim\mathcal N(0,1)$. The outcome regression is misspecified as
$\QYhat=\QY+0.5S$, the surrogate density as
$\mathcal N(0.8E+0.5W+0.5,\,1.5^2)$, so that $R_{SY}=0.25$ for $k=1$, and, in
the second row, the observation model as $\logit\gDhat=0.5+0.3S-0.3W$. In
every case the bias of the density-plug-in one-step estimator equals
$R_{E,1}+R_{\Delta,1}$, and $R_{SY}$ does not appear when the assignment law
is known. The script is included in the software
(\texttt{inst/checks/rsy\_numerical\_check.py}).

\begin{table}[ht]
\centering
\caption{Bias of the density-plug-in one-step estimator of $\mu_1$ in the
  design of Section~\ref{app:eic:dml}, compared with
  Proposition~\ref{M-prop:dml_remainder}. In all rows $R_{SY}=0.25$.}
\label{tab:rsy_check}
\small
\begin{tabularx}{\textwidth}{X r r}
\toprule
Nuisance configuration & Observed bias & $R_{E,1}+R_{\Delta,1}$ \\
\midrule
Assignment law known, $\gDhat$ correct & $-0.0004$ & $0$ \\
Assignment law known, $\gDhat$ misspecified & $-0.109$ & $-0.108$ \\
$\gEhat=0.6$, $\gDhat$ correct & $0.281$ & $0.282$ \\
\midrule
Unaugmented form ($\Qinthat\equiv0$), assignment known, $\gDhat$ correct & $-0.0004$ & $0$\\
\bottomrule
\end{tabularx}
\end{table}

In the same design, the variance of $\phi_1(O;\eta)-\mu_1$ at the true
nuisances is 7.2 with the augmentation at the true $\Qint$ and 12.7 without
augmentation, an instance of the term $\Sigma_H$ of
Theorem~\ref{M-thm:jack}.

\paragraph{Other representations.} The IPCW representation $\EE[I(E=k)\Delta Y/\{g_E g_\Delta\}]$ puts
$g_\Delta$ on the whole outcome, so observation probabilities near zero
inflate the variance. A density plug-in of the bridge removes $g_\Delta$ but
introduces $R_{SY}$. The estimating function~\eqref{M-eq:phi} keeps the bridge,
applies $\Delta/g_\Delta$ only to a residual, and pairs every surrogate-law
error with the exposure-law error, which the design sets to zero. The choice
among these estimators can thus be seen as a choice of which errors are
paired in the remainder.

\paragraph{Cross-fitting.} Let $\hat\eta$ converge to a limit $\eta^\dagger$ and write $\PP_{e}$ for the
empirical mean over records with $t\in\Te$. Then
\begin{equation}
  \label{eq:three_terms}
  \begin{aligned}
  \PP_e\phi_k(\cdot;\hat\eta)-\EE\{\Qint(k,W,t)\mid t\in\Te\}
  ={}& (\PP_e-P_0)\,\phi_k(\cdot;\eta^\dagger)\\
  &+ (\PP_e-P_0)\bigl\{\phi_k(\cdot;\hat\eta)-\phi_k(\cdot;\eta^\dagger)\bigr\}
  + B_k(\hat\eta).
  \end{aligned}
\end{equation}
Cross-fitting at the cluster level controls the middle (empirical-process)
term without Donsker conditions \citep{zheng2011cross,chernozhukov2018double}.
It does not act on $B_k$, which is a population quantity; here $B_k$ is
controlled by the design ($R_{E,k}=0$) and by a product-rate condition on
$(\QYhat,\gDhat)$.

\subsection{Proof of Theorem~\ref{M-thm:burden} and the Cluster-Level Burden}
\label{app:burden}

\paragraph{Part (i).} Write $D^*_e=(\inde/\pi_e)(U_W+U_S+U_Y)$ with
\begin{align*}
  U_W&=\Qint(e,W,t)-\Qint(0,W,t)-\Psi(e),\\
  U_S&=\sum_{k\in\{e,0\}}\pm_k\frac{I(E=k)}{g_E(k\mid t)}\{\QY(S,k,W,t)-\Qint(k,W,t)\},\\
  U_Y&=\sum_{k\in\{e,0\}}\pm_k\frac{I(E=k)\Delta}{g_E(k\mid t)g_\Delta(1\mid S,k,W,t)}\{Y-\QY(S,k,W,t)\}.
\end{align*}
The three terms are uncorrelated: $U_Y$ has mean zero given
$(S,E,W,t,\Delta)$, and $U_S$ has mean zero given $(E,W,t)$. Because
$I(E=e)I(E=0)=0$ and, by Assumption~\ref{M-ass:mar},
$\EE\{\Delta(Y-\QY)^2\mid S,E=k,W,t\}=g_\Delta\sigma^2_k$,
\[
  \EE(U_Y^2\mid W,t)=\sum_{k\in\{e,0\}}\frac{1}{g_E(k\mid t)}
  \EE\Bigl\{\frac{\sigma^2_k(S,W,t)}{g_\Delta(1\mid S,k,W,t)}\,\Big|\,E=k,W,t\Bigr\}.
\]
In the model in which $Y$ is always observed, the EIF of $\Psi(e)$ is the same
expression with $\Delta\equiv1$ and $g_\Delta\equiv1$; it has the same $U_W$
and $U_S$, because under Assumption~\ref{M-ass:mar} the complete-case
regression $\QY$ equals $\EE(Y\mid S,E,W,t)$. Subtracting,
$\Var(D^*_e)-V^{\mathrm{full}}(e)=\pi_e^{-2}\EE\{\inde(U_Y^2-U_{Y,\mathrm{full}}^2)\}$,
and $1/g_\Delta-1=(1-g_\Delta)/g_\Delta$ gives~\eqref{M-eq:burden}.

\paragraph{Part (ii).} Under the stated conditions $U_W\equiv0$, and the
full-data variance of $U_S+U_Y$ given $(W,t)$ is
$\sum_k g_E(k\mid t)^{-1}\Var(Y\mid E=k,W,t)=v\sum_kg_E(k\mid t)^{-1}$. Since
$g_E$ does not depend on $W$ and $W$ is independent of $E$ given $t$, the
burden integrand averages over $W$ to $v\,\rho(t)\sum_kg_E(k\mid t)^{-1}$.
With equal numbers of records per period, $t$ is uniform on $\Te$ given
$t\in\Te$, so
\[
  \frac{\mathcal B(e)}{V^{\mathrm{full}}(e)}
  =\frac{\sum_{t\in\Te}c(t)\rho(t)}{\sum_{t\in\Te}c(t)},\qquad
  c(t)=\sum_{k\in\{e,0\}}\frac{1}{g_E(k\mid t)}=(T-1)\Bigl(1+\frac{1}{T-t}\Bigr),
\]
using $g_E(e\mid t)=1/(T-1)$ and $g_E(0\mid t)=(T-t)/(T-1)$ under uniform
allocation. Write $A_e=\sum_{t\in\Te}c(t)\rho(t)$ and $C_e=\sum_{t\in\Te}c(t)$.
Since $\mathbb{T}_{e+1}=\Te\setminus\{e+1\}$,
$A_{e+1}/C_{e+1}\ge A_e/C_e$ if and only if $A_e/C_e\ge\rho(e+1)$, which holds
because $\rho$ is nondecreasing and $e+1$ is the smallest element of $\Te$. \qed

\paragraph{Cluster level.} By Lemma~\ref{M-lem:summation},
$\Var\{\EICj(e)\}=\bar M^{-2}\Var\{\sum_{i,t}D^*_e(O_{ijt})\}$. Suppose in
addition that, given the surrogates, exposures, covariates and outcomes of a
cluster, the observation indicators are independent across records with
probabilities $g_\Delta(1\mid\Sijt,\Ejt,\Wijt,t)$; this holds in the
structural model when censoring depends on the crossover time only through
$(\Ejt,t)$. Then for two distinct records the cross moment
$\EE\{U_Y(O_{ijt})U_Y(O_{i'jt'})\}$ is the same as in the full data, because
each factor $\Delta/g_\Delta$ averages to one given the other variables; by
the same argument the cross moments of $U_Y$ for one record with $U_S$ and
$U_W$ for another equal their full-data values, and the remaining cross
moments do not involve $\Delta$. Only the diagonal terms change, and
\[
  \Var\{\EICj(e)\}=V^{\mathrm{full}}_J(e)+\frac{\mathcal B(e)}{\bar M},
\]
where $V^{\mathrm{full}}_J(e)$ is the cluster-level variance with $Y$ always
observed. When censoring depends on the crossover time beyond $(\Ejt,t)$, a
cross-record term appears, but the split into a full-data part and a part that
vanishes when $g_\Delta\equiv1$ remains.

\subsection{Proof of Theorem~\ref{M-thm:asymp}}
\label{app:eic:r2}

Fix $e$; the argument applies to each coordinate and hence, by the
Cram\'er--Wold device, to the vector. Write
$\phi_{e0}(O;\eta)=\phi_e(O;\eta)-\phi_0(O;\eta)$ and let $N_e$ be the number of
records with $t\in\Te$.

\emph{Step 1 (ratio normalization).} By (C1), $N_e/J\to_P\pi_e\bar M$ with
error $O_P(J^{-1/2})$, and
\[
  \hat\Psi(e)-\Psi(e)=\frac{J}{N_e}\cdot\frac1J\sum_{j=1}^J
  \sum_{i,t}\inde\{\phi_{e0}(O_{ijt};\hat\eta)-\Psi(e)\}.
\]

\emph{Step 2 (empirical process).} Fix a fold $v$. Conditional on the
training clusters, the clusters in $F_v$ are i.i.d.\ and independent of
$\hat\eta^{(v)}$. Let $X_j(\eta)=\sum_{i,t}\inde\{\phi_{e0}(O_{ijt};\eta)-
\phi_{e0}(O_{ijt};\eta^\dagger)\}$. Under (C2), $\phi_k$ is Lipschitz in
$(\QY,\Qint,g_\Delta)$ with $g_E$ fixed, so
$\EE\{X_j(\hat\eta^{(v)})^2\mid\text{training}\}\le C(\bar nT)^2\sum_k
(\|\hat m-\QY^0\|_k^2+\|\hat\pi-g^0_\Delta\|_k^2+\|\Qinthat-\Qdag\|_k^2)
=o_P(1)$ by (C4). By Chebyshev's inequality conditional on the training
clusters, $J^{-1}\sum_{j\in F_v}[X_j(\hat\eta^{(v)})-\EE\{X_j(\hat\eta^{(v)})
\mid\text{training}\}]=o_P(J^{-1/2})$. Summing over the fixed number of folds
gives the same order.

\emph{Step 3 (remainder).}
$\EE[\sum_{i,t}\inde\{\phi_{e0}(O_{ijt};\eta)-\Psi(e)\}]
=\pi_e\bar M\{B_e(\eta)-B_0(\eta)\}$, and by
Proposition~\ref{M-prop:dml_remainder}(b),
$B_k(\hat\eta^{(v)})=R_{\Delta,k}(\hat\eta^{(v)})=o_P(J^{-1/2})$ under (C3). At
the limit, $B_k(\eta^\dagger)=0$ because $\hat m$ and $\hat\pi$ are consistent
and $R_{E,k}\equiv0$.

\emph{Step 4 (leading term and normality).} Combining Steps 1--3, and noting
that the product of the $O_P(J^{-1/2})$ error in $J/N_e$ with the
$O_P(J^{-1/2})$ leading average is $o_P(J^{-1/2})$,
$\hat\Psi(e)-\Psi(e)=J^{-1}\sum_j\phij(e)+o_P(J^{-1/2})$, with
$\phij(e)=\bar M^{-1}\sum_{i,t}\varphi_e(O_{ijt})$ and $\varphi_e$ as
in~\eqref{M-eq:phi_rec}. The decomposition $\varphi_e=D^*_e+h_e$ follows from
\[
  \phi_k(O;\eta^\dagger)=\phi_k(O;\eta_0)
  +\Bigl\{1-\frac{I(E=k)}{g_E(k\mid t)}\Bigr\}\{\Qdag(k,W,t)-\Qint^0(k,W,t)\},
\]
which holds because $\QY^\dagger=\QY^0$ and $g_\Delta^\dagger=g^0_\Delta$. The
$\phij$ are i.i.d.\ with mean zero and bounded under (C1)--(C2), so the
multivariate Lindeberg--Feller theorem
\citep[Proposition~2.27]{vaart1998asymptotic} gives part~(ii).

\emph{Step 5 (sandwich).} By (C4), Step~1 and the consistency of
$\hat\Psi$, $J^{-1}\sum_j\|\hat\varphi_j-\phij\|^2=o_P(1)$, and
$J^{-1}\sum_j\phij\phij^{\top}\to_P\Sigma_\varphi$ by the law of large numbers;
hence $J\widehat\Sigma_{\mathrm{sand}}\to_P\Sigma_\varphi$.

\emph{Part (iii).} If $\Qdag=\Qint^0$ then $h_e\equiv0$, so $\phij=\EICj$, the
canonical gradient of Section~\ref{app:eic:cluster}, and $\Sigma_\varphi$
equals the efficiency bound.

\emph{Part (iv).} If $\hat m\to\QY^0$ and $\hat\pi\to\bar\pi$, or
$\hat m\to\bar m$ and $\hat\pi\to g^0_\Delta$, in $L^2$, then
$R_{\Delta,k}(\hat\eta)\to\EE\{\int(\bar m-\QY^0)(1-g^0_\Delta/\bar\pi)f^0_S\,d\mu
\mid t\in\Te\}=0$, since one factor is zero, and the empirical averages
converge by the law of large numbers; $R_{E,k}\equiv0$ whatever the limit of
$\Qinthat$. \qed

\subsection{Proof of Theorem~\ref{M-thm:jack}}
\label{app:eic:jack}

\emph{Part (a).} By~\eqref{M-eq:h}, $h_e(O_{ijt})$ is a function of
$(\Ejt,W_{ijt},t)$, and $\Ejt$ is a function of $(\tau_j,t)$, so $H_j$ is a
function of $(\tau_j,\bW_j)$. Since $\tau_j$ is independent of $\bW_j$,
$\EE\{I(\Ejt=k)\mid\bW_j\}=g_E(k\mid t)$ and $\EE(H_j\mid\bW_j)=0$, so
$H_j\in\calT_{E,J}$. By~\eqref{eq:eif_cond},
$\EE(\EICj\mid\tau_j,\bW_j)=c(\bW_j)$ for a function $c$ of the covariates
alone, so $\EE(\EICj H_j^{\top})=\EE\{c(\bW_j)\EE(H_j^{\top}\mid\bW_j)\}=0$ and
$\Sigma_\varphi=\Sigma_{\mathrm{eff}}+\Sigma_H$. The conditional
moments~\eqref{eq:cond_tau} use the absence of interference and of carryover
beyond exposure time.

\emph{Part (b).} If the estimated EIF uses consistent nuisance estimates, the
argument of Step~5 above gives $J\widehat\Sigma_{\mathrm{EIF}}\to_P
\Var(\EICj)=\Sigma_{\mathrm{eff}}$, while
$J\Var(\hat\Psi)\to\Sigma_\varphi=\Sigma_{\mathrm{eff}}+\Sigma_H$ by part~(a).
For a scalar target the ratio of the two limits is
$1+\sigma^2_H/\sigma^2_{\mathrm{eff}}$, free of $J$.

\emph{Part (c).} Let $d_j=(J-1)(\hat\Psi^{(-j)}-\hat\Psi)$ and
$\epsilon_j=d_j+(\phij-\bar\varphi_J)$, so that
$J^{-1}\sum_j\|\epsilon_j\|^2=o_P(1)$ by (C5). Since
$\hat\Psi^{(-j)}-\hat\Psi^{(\cdot)}=(d_j-\bar d)/(J-1)$,
\[
  J\,\widehat\Sigma_{\mathrm{jack}}
  =\frac{J-1}{J}\cdot J\sum_{j=1}^J\frac{(d_j-\bar d)(d_j-\bar d)^{\top}}{(J-1)^2}
  =\frac{1}{J-1}\sum_{j=1}^J(d_j-\bar d)(d_j-\bar d)^{\top},
\]
the sample covariance of the $d_j$. Writing $d_j=-(\phij-\bar\varphi_J)+\epsilon_j$,
the cross terms are bounded by the Cauchy--Schwarz inequality and the
$\epsilon_j$ term by $J^{-1}\sum_j\|\epsilon_j\|^2$, so
$J\,\widehat\Sigma_{\mathrm{jack}}=J^{-1}\sum_j(\phij-\bar\varphi_J)(\phij-\bar\varphi_J)^{\top}
+o_P(1)\to_P\Sigma_\varphi$. Coverage of the intervals~\eqref{M-eq:ci_t}
follows from Theorem~\ref{M-thm:asymp}(ii), Slutsky's theorem and
$t_{J-1,\alpha/2}\to z_{\alpha/2}$. \qed

\subsection{Proof of Proposition~\ref{M-prop:jack_valid}}
\label{app:jack_valid}

Fix $e$ and suppress it from the notation. Let $v(j)$ be the fold of cluster
$j$. Removing cluster $j$ leaves the nuisance fit used for fold $v(j)$
unchanged, because cluster $j$ is not in its training set, and changes the
fits used for every other fold $v$ from $\hat\eta^{(v)}$ to
$\hat\eta^{(v,-j)}$. Write $\hat\Psi=N^{-1}\sum_{j'}Z_{j'}(\hat\eta^{(v(j'))})$ with
$Z_{j'}(\eta)=\sum_{i,t}\inde\phi_{e0}(O_{ij't};\eta)$ and $N=N_e$. Then
\begin{equation}
  \label{eq:jack_split}
  \hat\Psi^{(-j)}-\hat\Psi
  =\underbrace{\frac{N\hat\Psi-Z_j(\hat\eta^{(v(j))})}{N-N_{(j)}}-\hat\Psi}_{\text{(I)}}
  +\underbrace{\frac{1}{N-N_{(j)}}\sum_{v\ne v(j)}\sum_{j'\in F_v}
  \bigl\{Z_{j'}(\hat\eta^{(v,-j)})-Z_{j'}(\hat\eta^{(v)})\bigr\}}_{\text{(II)}},
\end{equation}
where $N_{(j)}$ is the number of records of cluster $j$ with $t\in\Te$.

\emph{Term (I)} is the removal of cluster $j$'s own contribution:
$(J-1)\cdot\text{(I)}=-(J/N)\{Z_j(\hat\eta^{(v(j))})-N_{(j)}\hat\Psi\}+O_P(J^{-1})
=-(\hat\varphi_j-\bar{\hat\varphi})+O_P(J^{-1})$, and
$J^{-1}\sum_j\|\hat\varphi_j-\phij\|^2=o_P(1)$ as in Step~5 of
Section~\ref{app:eic:r2}.

\emph{Term (II)} is the response of the other folds to the change in their
nuisance fits. For each $v\ne v(j)$ split
\[
  \sum_{j'\in F_v}\bigl\{Z_{j'}(\hat\eta^{(v,-j)})-Z_{j'}(\hat\eta^{(v)})\bigr\}
  =\sum_{j'\in F_v}\bigl[\Delta Z_{j'}-\EE(\Delta Z_{j'}\mid\text{training})\bigr]
  +|F_v|\,\pi_e\bar M\,\bigl\{\Delta B_e-\Delta B_0\bigr\},
\]
with $\Delta Z_{j'}=Z_{j'}(\hat\eta^{(v,-j)})-Z_{j'}(\hat\eta^{(v)})$ and
$\Delta B_k=B_k(\hat\eta^{(v,-j)})-B_k(\hat\eta^{(v)})$; the training set of
both fits excludes $F_v$, so the clusters in $F_v$ remain independent of
them. The first sum has conditional mean zero and conditional variance at most
$|F_v|C\|\hat\eta^{(v,-j)}-\hat\eta^{(v)}\|^2=O_P(J\cdot J^{-2})$ by
(S), so it is $O_P(J^{-1/2})$, and after division by $N-N_{(j)}\asymp J$ and
multiplication by $J-1$ it is $O_P(J^{-1/2})$. For the second sum, the design
makes the response second order. Because $g_E$ is plugged in,
$R_{E,k}\equiv0$ and $B_k=R_{\Delta,k}$ does not depend on $\Qinthat$, so the
change in the bridge regressions contributes nothing. The term $R_{\Delta,k}$
is bilinear in $(\hat m-\QY^0,\,1-g^0_\Delta/\hat\pi)$, so
\[
  |\Delta B_k|\le C\,\|\hat\eta^{(v,-j)}-\hat\eta^{(v)}\|\,
  \bigl(\|\hat m^{(v)}-\QY^0\|_k+\|\hat\pi^{(v)}-g^0_\Delta\|_k
  +\|\hat\eta^{(v,-j)}-\hat\eta^{(v)}\|\bigr)
  =O_P(J^{-1})\,o_P(1)
\]
by (S) and (C4). After division by $N-N_{(j)}$ and multiplication by $J-1$,
the second sum is $|F_v|\,o_P(J^{-1})=o_P(1)$, uniformly in $j$ because (S) is
uniform. Summing over the fixed number of folds,
$(J-1)\cdot\text{(II)}=\epsilon'_j$ with $J^{-1}\sum_j\|\epsilon'_j\|^2=o_P(1)$.
Together with term (I) this is (C5). \qed

If the partition were redrawn for each refit, $\hat\eta^{(v,-j)}$ would be
trained on a different set of clusters, $\|\hat\eta^{(v,-j)}-\hat\eta^{(v)}\|$
would be of order $J^{-1/2}$ rather than $J^{-1}$, and the first sum above
would be of order one after scaling, so (C5) would fail.

\subsection{Proof of Proposition~\ref{M-prop:alloc}}
\label{app:alloc}

Let $\bm\tau=(\tau_1,\ldots,\tau_J)$ be a uniformly random permutation of a
fixed vector in which crossover time $s$ appears $Jp_\tau(s)$ times, and let
the other cluster-level variables be i.i.d.\ and independent of $\bm\tau$.
Conditional on $\bm\tau$, clusters are independent, and the records of a
cluster with $\tau_j=s$ have the same law as under i.i.d.\ allocation given
$\tau_j=s$.

\emph{Expansion.} Steps 2 and~3 of Section~\ref{app:eic:r2} hold conditionally
on $\bm\tau$ when the folds are stratified by wave: the empirical-process
argument uses only conditional independence of the clusters in a fold from
the training clusters, and the conditional mean of $\sum_j Z_j(\eta)$ is
$J\sum_sp_\tau(s)\EE\{Z_j(\eta)\mid\tau_j=s\}$, which equals its i.i.d.\ value
because the allocation proportions are those of the design. Hence
$\hat\Psi-\Psi=J^{-1}\sum_j\phij+o_P(J^{-1/2})$ as before.

\emph{Variance.} Write $\phij=\EICj+\{H_j-a(\tau_j)\}+a(\tau_j)$. Since
$\EE\{a(\tau)\}=\EE(H_j)=0$ under the allocation law $p_\tau$ and the
allocation is fixed, $\sum_ja(\tau_j)=J\sum_sp_\tau(s)a(s)=0$ exactly. By
\eqref{eq:eif_cond}, $\EE(\EICj\mid\tau_j)=\EE\{c(\bW_j)\}=0$, and
$\EE\{H_j-a(\tau_j)\mid\tau_j\}=0$. Thus, conditional on $\bm\tau$,
$\sum_j\phij=\sum_j\{\EICj+H_j-a(\tau_j)\}$ is a sum of independent, bounded,
mean-zero terms, with average conditional covariance
\[
  \begin{aligned}
  \sum_sp_\tau(s)\Var(\EICj+H_j\mid\tau_j=s)
  &=\EE\{\Var(\phij\mid\tau_j)\}\\
  &=\Sigma_\varphi-\Var\{\EE(\phij\mid\tau_j)\}
  =\Sigma_\varphi-\Var\{a(\tau_j)\},
  \end{aligned}
\]
where the expectations are under i.i.d.\ allocation. The Lindeberg--Feller
theorem for triangular arrays, applied conditionally on $\bm\tau$, gives the
limit, which does not depend on $\bm\tau$; see \citet{li2017general} for
finite-population central limit theorems of this type.

\emph{Variance estimators.} Both estimators are built from the cluster-level
terms $\hat\varphi_j$, whose empirical second moment converges to
$\sum_sp_\tau(s)\EE(\phij\phij^{\top}\mid\tau_j=s)=\Sigma_\varphi$; the jackknife
inherits this limit through (C5). They are therefore conservative by
$\Var\{a(\tau_j)\}$, which is zero when $H\equiv0$, that is, when
$\Qdag=\Qint^0$. \qed

\subsection{Proof of Corollary~\ref{M-cor:het}}
\label{app:het}

Under $H_0$, $C\Psi=0$, so $\sqrt J\,C\hat\Psi\to_d\mathcal N(0,C\Sigma_\varphi C^{\top})$
by Theorem~\ref{M-thm:asymp}(ii), and $J\,C\widehat\Sigma_{\mathrm{jack}}C^{\top}\to_P
C\Sigma_\varphi C^{\top}$ by Theorem~\ref{M-thm:jack}(c). Since the limit is
nonsingular, the continuous mapping theorem gives $T^2\to_d\chi^2_q$. The $F$
reference is a heuristic: if $C\hat\Psi$ and $C\widehat\Sigma_{\mathrm{jack}}C^{\top}$
behaved like the mean and the scaled sample covariance of $J$ i.i.d.\ Gaussian
vectors, $\{(J-q)/(q(J-1))\}T^2$ would follow $F_{q,J-q}$ exactly. It has the
same limit, since $qF_{q,J-q}\to_d\chi^2_q$. \qed

\subsection{Proof of Proposition~\ref{M-prop:history}}
\label{app:history}

\emph{Tangent space.} Write $\zeta_m$ for $\zeta_{jm}$ and $\mathcal H_m$ for
$\mathcal H_{jm}$. The model $\mathcal M_S$ is the set of laws of $O_j$ with
$\EE(\zeta_m\mid\mathcal H_m)=0$ for $m=1,\ldots,T-2$, and is otherwise
unrestricted. Along a regular path with score $s$, differentiating
$\EE_\varepsilon\{\zeta_m\alpha(\mathcal H_m)\}=0$ gives
$\EE\{s\,\zeta_m\alpha(\mathcal H_m)\}=0$ for every bounded $\alpha$; conversely,
for bounded $s$ orthogonal to all such products, $p(1+\varepsilon s)$ stays in
the model. So the tangent space is the orthogonal complement of the closed
span $\Lambda$ of $\{\zeta_m\alpha_m(\mathcal H_m)\}$. The terms for different
$m$ are orthogonal: for $m<m'$ the product
$\zeta_m\alpha_m(\mathcal H_m)\zeta_{m'}\alpha_{m'}(\mathcal H_{m'})$ vanishes
unless $\tau_j>m'$, in which case $\zeta_m=-h_{m+1}$, and since
$\mathcal H_m\subseteq\mathcal H_{m'}$ its expectation is
$-h_{m+1}\EE\{\alpha_m\alpha_{m'}\EE(\zeta_{m'}\mid\mathcal H_{m'})\}=0$. The
restriction $\tau_j\perp\bW_j$ of the model of Lemma~\ref{M-lem:summation} is
implied, because $\bW_j\subseteq\mathcal H_m$ for every $m$; so
$\mathcal M_S$ is a submodel of that model.

\emph{Part (i).} $\EICj(e)$ is a gradient in the larger model, hence in
$\mathcal M_S$, and the canonical gradient in $\mathcal M_S$ is
$\EICj(e)-\Pi\{\EICj(e)\mid\Lambda\}=\EICj(e)-\sum_m\Pi\{\EICj(e)\mid\Lambda_m\}$,
with $\Lambda_m$ the span for a single $m$. Minimizing
$\EE\{X-\zeta_m\alpha(\mathcal H_m)\}^2$ over $\alpha$ gives
$\alpha^*_m=\EE(X\zeta_m\mid\mathcal H_m)/\EE(\zeta_m^2\mid\mathcal H_m)$. On
$\{\tau_j>m\}$ we have $\zeta_m=I(\tau_j=m+1)-h_{m+1}$ and
$\PP(\tau_j=m+1\mid\tau_j>m,\mathcal H_m)=h_{m+1}$, so both the numerator and
the denominator carry the factor $\PP(\tau_j>m\mid\mathcal H_m)h_{m+1}(1-h_{m+1})$,
and the ratio is the difference of conditional means in the statement. The
variance identity follows from orthogonality of the $\Lambda_m$.

\emph{Closed form of $\alpha^*_m$.} If the observation probability depends on
the crossover time only through exposure time, then
\begin{equation}
    \label{eq:alpha_star}
    \begin{aligned}
    \alpha^*_m(\mathcal H)=\frac{1}{\bar M\pi_e}\sum_{i}\sum_{t\in\Te,\,t>m}
    \Bigl[&\Bigl\{\frac{I(t=m+e)}{p_\tau(m+1)}-\frac{I(t>m+e)}{\PP(\tau_j>m+1)}\Bigr\}
    \varrho_{e,ijt}(\mathcal H)\\
    &+\frac{\varrho_{0,ijt}(\mathcal H)}{\PP(\tau_j>m+1)}\Bigr],
    \end{aligned}
  \end{equation}
  where $\varrho_{k,ijt}(\mathcal H)=\EE\{Y_{ijt}(k)-\Qint(k,W_{ijt},t)\mid
  \mathcal H_{jm}=\mathcal H\}$ is the part of a later residual that the
  history predicts.
To see this, fix $s>m$. Records in periods $t\le m$ are unexposed under
every $s>m$, and by the absence of anticipation their joint law with
$\mathcal H_m$ does not depend on $s$; they drop out of the difference. For a
record in a period $t>m$, write $D^*_e=(\inde/\pi_e)(U_W+U_S+U_Y)$ as in
Section~\ref{app:burden}. The term $\EE(U_W\mid\tau_j=s,\mathcal H_m)$ does not
depend on $s$. With exposure $k=E_t(s)$ fixed by $s$, the sum $U_S+U_Y$ equals
$\pm_k g_E(k\mid t)^{-1}\{Y-\Qint+(\Delta/g_\Delta-1)(Y-\QY)\}$. In the
structural model the observation noise $U_\Delta$ is independent of
everything else, so $\Delta$ is independent of $(Y_{ijt},\mathcal H_m)$ given
$(S_{ijt},E_{jt},W_{ijt},t,\tau_j)$, with $\PP(\Delta=1\mid\cdot)=g_\Delta$
when censoring depends on the crossover time only through exposure time.
Hence $\EE\{(\Delta/g_\Delta-1)(Y-\QY)\mid S,E,W,t,\tau_j,\mathcal H_m\}=0$;
note that $\EE(Y-\QY\mid S,E,W,t,\mathcal H_m)$ itself is not zero, since the
history is informative about the cluster effect. Therefore
\[
  \EE(U_S+U_Y\mid\tau_j=s,\mathcal H_m)
  =\sum_{k\in\{e,0\}}\pm_k\frac{I\{E_t(s)=k\}}{g_E(k\mid t)}\,
  \EE\{Y_{ijt}(k)-\Qint(k,W_{ijt},t)\mid\mathcal H_m\},
\]
using also the independence of $\tau_j$ from the potential outcomes and from
$\mathcal H_m$ on $\{\tau_j>m\}$. Averaging the indicator coefficients over
$s=m+1$ and over $s>m+1$ (with weights $p_\tau(s)/\PP(\tau_j>m+1)$), and using
$g_E(e\mid t)=p_\tau(t-e+1)$ and $g_E(0\mid t)=\PP(\tau_j>t)$, gives the
coefficients in~\eqref{eq:alpha_star}: for $k=e$ the difference is
$I(t=m+e)/p_\tau(m+1)-I(t>m+e)/\PP(\tau_j>m+1)$, and for $k=0$ it is
$1/\PP(\tau_j>m+1)$ because $t\ge m+1$.

\emph{Part (ii).} Conditionally on the training clusters of fold $v$,
$\hat\alpha_m$ is fixed and
$\EE\{\zeta_{jm}(\hat\alpha_m-\alpha^\dagger_m)(\mathcal H_{jm})\mid
\text{training}\}=0$ exactly, so the average of these terms over the clusters
of fold $v$ is $O_P\{J^{-1/2}\|\hat\alpha_m-\alpha^\dagger_m\|\}=o_P(J^{-1/2})$.
The expansion of Theorem~\ref{M-thm:asymp} then gives the influence function.
A scalar calibration estimated by least squares changes it only at second
order. When $\Qdag=\Qint^0$, $\phij=\EICj$, and with $\alpha^\dagger_m=\alpha^*_m$
the influence function is~\eqref{M-eq:eif_history}. \qed

\emph{Part (iii).} Under~\eqref{M-eq:index_structure}, for every record of a
period $t>m$ and $k\in\{e,0\}$,
$\varrho_{k,ijt}(\mathcal H)=\lambda^\top\EE(u_j\mid\mathcal H_{jm}=\mathcal H)
=\lambda^\top B_m\bar r_{jm}$, which depends on neither $i$, $t$ nor $k$, so
the sums in~\eqref{eq:alpha_star} collapse to
$\alpha^*_m(\mathcal H_{jm})=\ell_{m,e}\lambda^\top B_m\bar r_{jm}$ with
$\ell_{m,e}$ the sum of the coefficients in~\eqref{eq:alpha_star} over the
records of periods $t>m$, which is the quantity defined below in the
cluster-effect calculation. Summing over $m$ gives~\eqref{M-eq:index}. Each
term of $Z_j$ is of the form $\zeta_{jm}\alpha(\mathcal H_{jm})$, so
$Z_j\in\Lambda$, and when $\Qdag=\Qint^0$ the canonical gradient
$\EICj(e)-\lambda^\top Z_j$ is orthogonal to $\Lambda$; hence
$\Cov\{\EICj(e),Z_j\}=\Var(Z_j)\lambda$ and the population least-squares
coefficient is $\lambda$. Part (ii) with
$\alpha^\dagger_m=\lambda^\top\ell_{m,e}B_m\bar r_{jm}$ gives the influence
function~\eqref{M-eq:eif_history}. When the structure fails, part (ii) applies
with $\alpha^\dagger_m=\kappa^{*\top}\ell_{m,e}B_m\bar r_{jm}$ and
$\kappa^*=\Var(Z_j)^{-1}\Cov\{\phij(e),Z_j\}$, and the asymptotic variance is
that of the residual of $\phij(e)$ on $Z_j$. \qed

\emph{Fixed weights.} Write $\bar r_{jm}=(\bar r_{S,jm},\bar r_{Y,jm})$ for the
means of the residuals $r_S=S-m_S(W,t)$ and $r_Y=Y-m_Y(S,W,t)$ over the
cluster's records, respectively observed records, of periods $1,\ldots,m$,
and suppose $B_m$ is fixed given the other folds. Under~\eqref{M-eq:index_structure}
with $r_S=u_{S,j}+e_S$ and $r_Y=u_{Y,j}+e_Y$,
$\bar r_{Y,jm}=u_{Y,j}+\bar e_{Y,jm}$, where $\bar e_{Y,jm}$ averages $e_Y$
over the observed records. Then
$\EE(\zeta_{jm}\bar r_{Y,jm})=\EE\{u_{Y,j}\EE(\zeta_{jm}\mid u_j)\}
+\EE\{\zeta_{jm}\EE(\bar e_{Y,jm}\mid\tau_j,u_j,\ldots)\}$. The first term is
zero by the sequential randomization of $\tau_j$ given the cluster's
pre-crossover potential data; the second is zero when the observation
indicators are independent of $e_Y$ given $(S,W,t,\tau_j)$ and
$\EE(e_Y\mid S,W,t)=0$, whatever the observation probabilities, provided
each cluster has at least one observed outcome in periods $1,\ldots,m$. The
same holds for $\bar r_{S,jm}$. With the cluster's own record counts in $B_m$
the weights depend on $\tau_j$ through the counts, and the argument needs the
observation probabilities not to depend on the wave given the period.

\paragraph{The cluster-effect calculation.} Suppose
$Y_{ijt}(k)-\Qint(k,W_{ijt},t)=\lambda\{b_j-\EE(b_j\mid\bW_j)\}+\epsilon_{ijt}$
with $\epsilon_{ijt}$ independent of the cluster's other data. Then the
cluster-effect part of $\EICj(e)$ is $\lambda\{b_j-\EE(b_j\mid\bW_j)\}L_e(\tau_j)$
with $L_e(s)=(\bar M\pi_e)^{-1}\sum_{t\in\Te}n_t\{I(s=t-e+1)/g_E(e\mid t)
-I(s>t)/g_E(0\mid t)\}$. Any mean-zero function of $\tau_j$ can be written as
$\sum_m\zeta_{jm}\ell_m$ with
$\ell_m=\EE\{L_e(\tau_j)\mid\tau_j=m+1\}-\EE\{L_e(\tau_j)\mid\tau_j>m+1\}$, and
because $\zeta_{jm}$ is independent of $b_j$ given $\mathcal H_{jm}$, the
projection replaces $b_j-\EE(b_j\mid\bW_j)$ in the $m$th term by
$b_j-\EE(b_j\mid\mathcal H_{jm})$. The variance of the cluster-effect part
therefore falls from
$\lambda^2\sum_m\EE(\zeta_m^2)\ell_m^2\,\EE\{\Var(b_j\mid\bW_j)\}$ to
$\lambda^2\sum_m\EE(\zeta_m^2)\ell_m^2\,\EE\{\Var(b_j\mid\mathcal H_{jm})\mid\tau_j>m\}$,
where the factorization uses that $\zeta_{jm}^2$ depends on $(\tau_j,\mathcal
H_{jm})$ only through $I(\tau_j>m)$ and that $\tau_j$ is independent of
$\mathcal H_{jm}$ on $\{\tau_j>m\}$. For
$\Psi(1)$ with $T=7$, uniform allocation and equal $n_t$, the products
$\EE(\zeta_m^2)\ell_m^2$ are proportional to $(7.2,7.5,8,9,12)$, so the later,
longer histories receive the larger weights, but each can at most remove
what is left of the cluster-effect variance given $\bW_j$. In the Gaussian
design of Section~\ref{M-sec:sim:design} the conditional means are available
in closed form, and Section~\ref{app:sim} uses them to evaluate the bound
exactly.

\subsection{Proof of Proposition~\ref{M-prop:design}}
\label{app:design}

Write $g_\Delta$ for the observation probability given
$(S,E,W,t,\tau_j)$; when censoring depends on the wave only through exposure
time this is the $g_\Delta$ of the main text. The cluster-level influence
function $\phij$ of SA-AIPW with a consistent bridge is then
$\bar M^{-1}\sum_{i,t}D^*_e(O_{ijt})$, with $D^*_e$ as in~\eqref{M-eq:Dstar}
but with this $g_\Delta$. Use the decomposition
$D^*_e=(\inde/\pi_e)(U_W+U_S+U_Y)$ from Section~\ref{app:burden}, and note that $\bar M\pi_e=\sum_{t\in\Te}n_t$ when
every cluster has $n_t$ records in period $t$. For a record of period $t$ in a
cluster with $\tau_j=s$ the exposure is fixed by $s$, and
$U_S+U_Y$ equals $\pm1/g_E(k\mid t)$ times
\[
  \begin{aligned}
  r_{ijt}={}&\QY(\Sijt,k,\Wijt,t)-\Qint(k,\Wijt,t)\\
  &+\frac{\Dijt}{g_\Delta(1\mid\Sijt,k,\Wijt,t)}\{\Yijt-\QY(\Sijt,k,\Wijt,t)\},
  \qquad k=E_t(s),
  \end{aligned}
\]
with the plus sign for $k=e$ and the minus sign for $k=0$; it is zero for
other exposures. Hence $\phij=A_j+R_j$ with
$A_j=(\sum_{t\in\Te}n_t)^{-1}\sum_{i}\sum_{t\in\Te}U_W(O_{ijt})$ and
$R_j=\sum_i\sum_t c_t(\tau_j)\,r_{ijt}$. For $\bar\Psi$ both pieces are the
averages over $e$ of the pieces for $\Psi(e)$, and the argument below is
unchanged.

The term $A_j$ is a function of $\bW_j$ and does not involve $\tau_j$, so its
variance $V_W$ does not depend on $p_\tau$. By~\eqref{eq:cond_tau},
$\EE(r_{ijt}\mid\tau_j,\bW_j)=0$, so $\EE(R_j\mid\tau_j,\bW_j)=0$, the
covariance of $A_j$ and $R_j$ is zero, and
$\Var(\phij)=V_W+\sum_sp_\tau(s)\EE(R_j^2\mid\tau_j=s)$.

Fix $\tau_j=s$. Under the working model, $\Yijt-\Qint=b_j+\epsilon_{ijt}$
has variance $\sigma_b^2+\sigma_\epsilon^2$ and splits into $\QY-\Qint$ and
$\Yijt-\QY$, which are uncorrelated, the second with variance $\sigma_r^2$; so
$\Var(\QY-\Qint)=\sigma_b^2+\sigma_\epsilon^2-\sigma_r^2$. For one record,
\[
  \EE(r_{ijt}^2\mid\tau_j=s)=\Var(\QY-\Qint)+\sigma_r^2\,
  \EE\Bigl(\frac{1}{g_\Delta}\,\Big|\,t,\tau_j=s\Bigr)
  =\sigma_b^2+\sigma_\epsilon^2+\sigma_r^2\,o_t(s).
\]
For two distinct records of the cluster, condition on all of the cluster's
variables other than the two observation indicators. By assumption the
indicators are then independent with probabilities $g_\Delta$, so each factor
$\Dijt/g_\Delta$ has conditional mean one and
$\EE(r_{ijt}r_{i'jt'}\mid\tau_j=s)$ equals its value with $Y$ always observed,
$\Cov(\Yijt-\Qint,Y_{i'jt'}-\Qint)=\sigma_b^2$. Adding up,
\[
  \EE(R_j^2\mid\tau_j=s)=\sum_tn_tc_t(s)^2\{\sigma_b^2+\sigma_\epsilon^2+\sigma_r^2o_t(s)\}
  +\sigma_b^2\Bigl[\Bigl\{\sum_tn_tc_t(s)\Bigr\}^2-\sum_tn_tc_t(s)^2\Bigr],
\]
which is the bracket in~\eqref{M-eq:design_var}.

For $\emax=T-2$ the only overlap period is $T-1$. Exposure $\emax$ in that
period means $s=2$, exposure $0$ means $s=T$, and
$g_E(\emax\mid T-1)=p_\tau(2)$, $g_E(0\mid T-1)=p_\tau(T)$. So
$c_{T-1}(2)=1/\{n_{T-1}p_\tau(2)\}$, $c_{T-1}(T)=-1/\{n_{T-1}p_\tau(T)\}$ and all
other coefficients vanish, and~\eqref{M-eq:design_var}
gives~\eqref{M-eq:neyman}. Writing $a$ and $b$ for the two numerators, with
$p_\tau(2)+p_\tau(T)=q$ fixed, the Cauchy--Schwarz inequality gives
$a/p_\tau(2)+b/p_\tau(T)\ge(\sqrt a+\sqrt b)^2/q$, with equality when
$p_\tau(2)/p_\tau(T)=\sqrt{a/b}$. \qed

\paragraph{Estimating the inputs.} In \texttt{design\_inputs} the inputs come
from a fitted SA-AIPW object (a pilot study, an earlier trial, or the trial at
hand when the question is how a different allocation would have done).
Take the residuals $\Yijt-\Qinthat(\Ejt,\Wijt,t)$ of the observed records.
$\sigma_\epsilon^2$ is their pooled variance within cluster-periods, and
$\sigma_b^2$ is the variance of their cluster-period means minus
$\sigma_\epsilon^2$ divided by the average number of observed records per
cell, truncated at zero. This is the usual one-way analysis-of-variance
estimator; it slightly mixes cluster and cluster-period variation, which
should matter little for planning. $\sigma_r^2$ is the average of
$(\Yijt-\hat m)^2$ over observed records, and $o_t(s)$ is the average of
$(1-\hat\pi)/\hat\pi$ over all records in period $t$ of wave-$s$ clusters. The
predictions reported in Section~\ref{M-sec:sim} and Section~\ref{app:ept} omit
$V_W$, which is zero when the conditional effect does not vary with $(W,t)$,
as in the design of Section~\ref{M-sec:sim:design}, and small otherwise. The
enumeration in \texttt{design\_search} covers all allocations with at least a
given number of clusters per wave; with $J=30$, six waves and at least two per
wave there are 33{,}649 of them, and the search takes a few seconds.

\subsection{Sensitivity to Assumption~\ref{M-ass:mar}}
\label{app:sens}

Missingness at random
places no restriction on the observed-data law \citep{gill1997coarsening}, so
Assumption~\ref{M-ass:mar} cannot be tested from the observed data. Let
$\gamma_k(S,W,t)=\EE(Y\mid S,E=k,W,t,\Delta=0)-\EE(Y\mid S,E=k,W,t,\Delta=1)$.
Then $\EE(Y\mid S,E=k,W,t)=\QY+(1-g_\Delta)\gamma_k$, and with consistent
nuisance estimators $\hat\Psi(e)$ converges to $\Psi(e)-\Lambda_e(\gamma)$, where
\begin{equation}
  \label{eq:bias_gamma}
  \Lambda_e(\gamma)
  =\EE\Bigl[\frac{\inde}{\pi_e}\sum_{k\in\{e,0\}}\pm
    \frac{I(E=k)}{g_E(k\mid t)}\,
    \bigl\{1-g_\Delta(1\mid S,k,W,t)\bigr\}\,\gamma_k(S,W,t)\Bigr],
\end{equation}
with sign $+$ for $k=e$. The bias can be estimated by plugging in the fitted
nuisances and a specified $\gamma$, and tipping points can be reported for
each exposure time \citep{scharfstein1999adjusting,vanderweele2017sensitivity}.
Because $1-g_\Delta$ is largest in the overlap periods of long exposure times,
the same departure from Assumption~\ref{M-ass:mar} moves long-exposure effects
the most, so the coupling affects sensitivity as well as precision.
Section~\ref{app:ept} applies this to the replication data set.

\section{Implementation}
\label{app:impl}

\subsection{Algorithm}
\label{app:impl:alg}

The \texttt{R} function \texttt{sa\_aipw\_dur} in the package
\texttt{swsurrdur} implements the following steps. The inputs are the data,
the design allocation probabilities $p_\tau$, the largest exposure time
$\emax$, a fold assignment for the clusters, and three model formulas: for
the outcome regression, the observation model and the bridge regression.
\begin{enumerate}[leftmargin=2em]
  \item \emph{Design quantities.} Compute $g_E(k\mid t)$ from~\eqref{M-eq:gE},
    $g_A(1\mid t)=\sum_{k\ge1}g_E(k\mid t)$, the overlap periods $\Te$ and
    $\Tov=\bigcup_e\Te$.
  \item \emph{Folds.} Assign clusters to $V$ folds once, stratified by
    crossover wave (\texttt{make\_folds\_strat}), so that every training set
    contains clusters from every wave.
  \item \emph{Stage 1.} For each fold $v$, fit the outcome regression on the
    training records with $\Delta=1$ (linear or logistic) and the observation
    model on all training records (logistic). Truncate $\hat\pi$ to
    $[c,c']$ with $c=0.01$; the default $c'=1-c$ was used in all reported
    runs, and $c'=1$ (argument \texttt{trunc\_upper}) is appropriate when
    some cells are known to be fully observed. In the EPT design the first two
    waves are fully observed, so 68\% of the fitted probabilities in
    Section~\ref{M-sec:sim:ept} were capped at 0.99 rather than 1 and the
    observation model there is, strictly, slightly misspecified in those
    cells; rerunning the growing scenario under the trial's allocation with
    $c'=1$ changed every estimate by less than 0.012 percentage points and
    left bias, standard deviation and coverage unchanged to three decimals
    (\nolinkurl{inst/sim/ept_trunc_check.R}).
  \item \emph{Stage 2.} With the fold-$v$ fits, compute $\widetilde
    Y=\hat m+\Delta(Y-\hat m)/\hat\pi$ on the training records (for Stage~3)
    and on the records of fold $v$ (for the estimator).
  \item \emph{Stage 3.} Regress the training $\widetilde Y$ on the bridge
    formula plus exposure-time indicators among training records in the
    overlap periods ($E=k$ with $t\in\mathbb{T}_k$ for $k\ge1$, and $E=0$ with
    $t\in\Tov$), and predict for the records of fold $v$ in $\Tov$ at each
    exposure level $k=0,\ldots,\emax$. For $\PsiIT$, the bridge for the exposed
    arm is $\sum_{e\ge1}w_t(e)\Qinthat(e,W,t)$. This pooled additive bridge is
    the default (\texttt{bridge = "pooled"}). With \texttt{bridge =
    "separate"}, the regression is fitted separately for each level $k$, and
    the training mean is used when fewer than ten records are available.
  \item \emph{Estimates.} Compute $\phi_k$ in~\eqref{M-eq:phi} for every
    record from the cross-fitted $\widetilde Y$ and $\Qinthat$, the
    estimates~\eqref{M-eq:psia_aipw} for $e=1,\ldots,\emax$, $\bar\Psi$, and
    $\PsiIT$, and the cluster-level estimating functions $\hat\varphi_j$.
  \item \emph{Inference.} The sandwich~\eqref{M-eq:var} is
    \texttt{sandwich\_cov}. The jackknife~\eqref{M-eq:jack_var}
    (\texttt{jackknife\_cov}) reruns Steps 3--6 for each of the $J$
    leave-one-cluster-out data sets with the fold assignment of Step~2
    unchanged. The heterogeneity test (\texttt{het\_test}) uses the contrast
    matrix of Section~\ref{M-sec:het}.
\end{enumerate}
Setting the bridge predictions to zero gives the unaugmented estimator.
Applying the same steps with $\Delta\equiv1$ gives the full-data estimator used
in Section~\ref{M-sec:sim:ept} and Section~\ref{app:ept}.

\subsection{Choices and Practical Notes}
\label{app:impl:notes}

\paragraph{Working models.} The outcome regression includes exposure-time
indicators, which allow an exposure-specific intercept in the regression of
$Y$ on $S$ (an exposure-specific slope would need interaction terms), and
period indicators for the secular trend; the bridge regression includes
period indicators for the same reason. Both are record-level regressions on
the individual covariates; the cluster's history enters through the index of
Proposition~\ref{M-prop:history}(iii) rather than through cluster-level
covariates, which with few clusters make the regressions interpolate the
training clusters (Section~\ref{app:sim}). Any learner can be used in
Stages~1 and~3. For the jackknife, condition (S) of
Proposition~\ref{M-prop:jack_valid} holds for parametric working models with
a fixed number of parameters fitted by least squares or maximum likelihood,
under bounded covariates; with penalized or data-adaptive learners we suggest
fixing tuning parameters and ensemble weights at their full-data values
across the refits.

\paragraph{Specifications in Section~\ref{M-sec:sim:design}.} SA-AIPW used a
linear outcome regression on $S$, exposure-time indicators, $W$ and period
indicators; a logistic observation model on $S$, $A$, $W_3$ and period
indicators, with $\hat\pi$ truncated at $0.01$; a linear bridge regression,
pooled over exposure levels, on exposure-time indicators, $W$ and period
indicators; and the history index with the surrogate and outcome histories,
residual regressions of $S$ on $W$ and period indicators and of $Y$ on $S$,
$W$ and period indicators, fixed weights and one component. The estimator of
the earlier versions, with the period-1 cluster means of the surrogate and
of the observed outcome as covariates in the outcome and bridge regressions
and no history term, is reported alongside in Section~\ref{app:sim}. The
unaugmented estimator and IPCW used the same Stage~1 fits and the same index,
and IPCW weighted observed outcomes by $1/\{g_E\hat\pi\}$ with the same
observation model. The mixed models had period fixed effects and $W$ in
addition to the terms listed under Comparators below, and model-based
standard errors.

\paragraph{Aliased coefficients.} The fitting functions set coefficients
that \texttt{lm.fit} or \texttt{glm.fit} report as \texttt{NA} to zero. This
happens when a training set has no records at some exposure level, as in the
leave-one-cluster-out refits with $J=12$ that remove one of the two
first-wave clusters, and when cluster-level columns outnumber the training
clusters (Section~\ref{app:sim}); the fitted values on the training set are
unaffected, but the out-of-fold predictions then use one particular
solution, and condition (S) of Proposition~\ref{M-prop:jack_valid} no longer
holds for those refits. The estimator of the earlier versions, with the
period-1 cluster mean of the outcome among the covariates, used that mean
with one slope for all periods although for the period-1 records it
contains their own outcome; the period-1 records are not in any overlap
period, so this affected only the training of the outcome model.

\paragraph{Number of folds.} In the simulations,
$V=\max\{2,\min(5,\lfloor J/(T-1)\rfloor)\}$, which keeps at least one
cluster per wave in every fold when the allocation is balanced; in the
EPT-calibrated simulation $V=5$.

\paragraph{Long exposure times.} Few records have the longest exposure
times, because few clusters reach them during the trial. Separate bridge
regressions for those levels are fitted to a handful of clusters, which is
unstable when the bridge includes cluster-level covariates; pooling over
exposure levels avoids this. By Proposition~\ref{M-prop:dml_remainder}(b) the
choice affects only efficiency.

\paragraph{Computation.} One fit for $J=30$, $T=7$ and 40 records per
cluster-period takes about 0.2 seconds on one core, and the jackknife about 6
seconds. The full simulation study took about four hours on two cores.

\paragraph{The history index.} The function \texttt{history\_index}
implements Proposition~\ref{M-prop:history}(iii). Its working structure for
the pre-crossover records is a bivariate random-intercept model for the
residuals $r_S=S-m_S(W,t)$ and $r_Y=Y-m_Y(S,W,t)$, with cluster effects
$u_j=(u_{S,j},u_{Y,j})$ of covariance $\Sigma_u$ and within-cluster variances
$\sigma_S^2$ and $\sigma_Y^2$. For fold $v$, the regressions $m_S$ and $m_Y$
and the five variance components are estimated from the unexposed records
($E=0$) of the other folds, the variance components by the method of moments
from the cluster means of the in-sample residuals. For a cluster of fold $v$
and $m=1,\ldots,T-2$, the predictor is the best linear predictor of $u_j$
from the cluster's residual means over periods $1,\ldots,m$,
$\hat u_{jm}=\Sigma_u(\Sigma_u+D_{m})^{-1}\bar r_{jm}$ with
$D_m=\mathrm{diag}(\sigma_S^2/n_S,\sigma_Y^2/n_Y)$, where by default $n_S$ and
$n_Y$ are the average record and observed-record counts of the training
clusters over those periods (fixed weights; the cluster's own counts are an
option). The index of~\eqref{M-eq:index} is formed with the design constants
of~\eqref{eq:alpha_star} (\texttt{history\_constants}) and the component of
$\hat u_{jm}$ for the outcome, which already uses the surrogate channel
through $\Sigma_u$; the two-component version is available. The loading is
the least-squares coefficient of the estimator's cluster-level estimating
function on the centred index, one per target. When only the surrogate
histories are admissible, \texttt{hist\_vars = "S"} uses $r_S$ alone. The
standard errors are the HC2 or HC3 robust standard errors of the intercept
of the regression of the estimating function on the index
(\texttt{index\_se}); the sandwich formula~\eqref{M-eq:var} applied to the
adjusted estimating functions ignores the estimation of the loading and is
anticonservative with 12 clusters (Section~\ref{app:sim}). The jackknife of
\texttt{jackknife\_cov} refits the residual regressions, the variance
components, the index and the loading for every left-out cluster
(\texttt{sa\_aipw\_index}). The function \texttt{history\_augment} keeps the
working model of the earlier versions of this paper, which regressed
residuals on cluster means separately for each $m$; it recovered little of
the bound with 30 clusters (Section~\ref{app:sim}).

\paragraph{Intervals for binary outcomes.} The function \texttt{arm\_means}
computes, for each $e$, the arm means $\hat\mu_e$ and $\hat\mu_0$ over $\Te$
(the averages of $\phi_e$ and $\phi_0$, whose difference is $\hat\Psi(e)$) and
their cluster-level estimating functions, from which a sandwich covariance
matrix $\widehat V$ of all $2\emax$ arm means is formed as
in~\eqref{M-eq:var}; \texttt{arm\_means\_index} applies the history index of
each target to its two arm-level estimating functions, so that the intervals
refer to the same estimator as the contrasts (the adjusted difference of the
arm means equals the adjusted contrast). \texttt{ci\_mover} gives each arm mean the interval
$\operatorname{expit}\{\operatorname{logit}\hat\mu_k\pm
t_{J-1,\alpha/2}\,\widehat{\mathrm{se}}_k/(\hat\mu_k(1-\hat\mu_k))\}$, or the
Wald interval if $\hat\mu_k\notin(0,1)$, and combines them for a linear
combination $\theta=\sum_ic_i\mu_i$ as follows
\citep{zou2008construction,zou2008estimation}. Let $(L_i,U_i)$ be the limits
for $c_i\mu_i$ (the arm limits multiplied by $c_i$, reordered if $c_i<0$),
$a_i=c_i\hat\mu_i-L_i$, $b_i=U_i-c_i\hat\mu_i$, and $\rho_{ik}$ the correlation
of $c_i\hat\mu_i$ and $c_k\hat\mu_k$ from $\widehat V$. Then
\[
  \Bigl(\hat\theta-\Bigl\{\sum_{i,k}\rho_{ik}a_ia_k\Bigr\}^{1/2},\;
  \hat\theta+\Bigl\{\sum_{i,k}\rho_{ik}b_ib_k\Bigr\}^{1/2}\Bigr).
\]
For $\Psi(e)$ this is~\eqref{M-eq:mover}; for $\bar\Psi$ the coefficients are
$\pm1/\emax$ on all $2\emax$ arm means. With symmetric arm limits the interval
is the sandwich interval. Replacing the logit by the arcsine square-root
scale gives an alternative, reported in Section~\ref{app:ept}.

\paragraph{Comparators.} IPCW (\texttt{ipcw\_dur}) uses the same observation
model and design weights. The complete-case mixed models (\texttt{lmm\_fits})
are fitted with \texttt{lme4} by restricted maximum likelihood: the
immediate-treatment model $Y\sim A+\text{period}+(1\mid\text{cluster})$ and the
exposure-time-indicator model $Y\sim\text{exposure time}+\text{period}+
(1\mid\text{cluster})$, with exposure time as a factor; the exposure-averaged
effect from the latter is the average of the coefficients for
$e=1,\ldots,\emax$.

\section{Additional Simulation Results}
\label{app:sim}

This section gives the results summarized in Section~\ref{M-sec:sim}.
All numbers are produced by \texttt{inst/sim/make\_tables.R} from the saved
simulation cells. Bias, standard deviation (SD), the ratio of the mean
standard error to the SD (SE/SD) and coverage of 95\% intervals with $t_{J-1}$
quantiles are computed over the replications of each cell. Cells have 500
replications, except Block~I with $J=18$ and $J=30$ (1000) and $J=48$ (250)
and Block~V (300);
the Monte Carlo standard error of a coverage near $0.95$ is about $0.007$ with
1000 replications and $0.01$ with 500.

\paragraph{Block I, all estimators.} Table~\ref{tab:supp_main_full} repeats
Table~\ref{M-tab:sim_main} with the unaugmented estimator and IPCW, both
without the history term.

\begin{table}[!htbp]
\centering
\caption{Block I, $J=30$: Table~\ref{M-tab:sim_main} of the main text with the
unaugmented estimator and IPCW (true values in parentheses). Standard errors
are jackknife for SA-AIPW, sandwich for SA-AIPW without the history term,
the unaugmented estimator and IPCW, and model-based for the mixed models
(LMM).}
\label{tab:supp_main_full}
\small
\begin{tabular}{llrrrr}
\toprule
Target & Estimator & Bias & SD & SE/SD & Coverage\\
\midrule
$\Psi(1)$ (0.046) & SA-AIPW & $-$0.001 & 0.038 & 1.14 & 0.973 \\
 & SA-AIPW, no history term & 0.000 & 0.066 & 1.04 & 0.960 \\
 & Unaugmented SA-AIPW & $-$0.001 & 0.071 & 2.55 & 1.000 \\
 & IPCW & $-$0.001 & 0.071 & 2.53 & 1.000 \\
 & Complete-case LMM, ETI & $-$0.003 & 0.036 & 1.05 & 0.970 \\
\addlinespace
$\Psi(3)$ (0.219) & SA-AIPW & $-$0.002 & 0.059 & 1.08 & 0.957 \\
 & SA-AIPW, no history term & $-$0.002 & 0.110 & 1.00 & 0.943 \\
 & Unaugmented SA-AIPW & $-$0.003 & 0.116 & 2.97 & 1.000 \\
 & IPCW & $-$0.004 & 0.118 & 2.94 & 1.000 \\
 & Complete-case LMM, ETI & $-$0.005 & 0.055 & 1.00 & 0.958 \\
\addlinespace
$\Psi(5)$ (0.309) & SA-AIPW & $-$0.001 & 0.126 & 1.16 & 0.958 \\
 & SA-AIPW, no history term & $-$0.001 & 0.183 & 1.05 & 0.950 \\
 & Unaugmented SA-AIPW & $-$0.001 & 0.197 & 2.79 & 1.000 \\
 & IPCW & 0.001 & 0.207 & 2.68 & 1.000 \\
 & Complete-case LMM, ETI & $-$0.007 & 0.083 & 1.05 & 0.971 \\
\addlinespace
$\bar\Psi$ (0.199) & SA-AIPW & $-$0.001 & 0.055 & 1.14 & 0.968 \\
 & SA-AIPW, no history term & $-$0.001 & 0.103 & 1.03 & 0.942 \\
 & Unaugmented SA-AIPW & $-$0.001 & 0.107 & 2.99 & 1.000 \\
 & IPCW & $-$0.001 & 0.109 & 2.95 & 1.000 \\
 & Complete-case LMM, ETI & $-$0.004 & 0.047 & 1.04 & 0.966 \\
\addlinespace
$\Psi_{\mathrm{IT}}$ (0.131) & SA-AIPW & $-$0.002 & 0.041 & 1.13 & 0.965 \\
 & SA-AIPW, no history term & $-$0.002 & 0.082 & 1.02 & 0.951 \\
 & Unaugmented SA-AIPW & $-$0.001 & 0.085 & 2.92 & 1.000 \\
 & IPCW & $-$0.001 & 0.086 & 2.92 & 1.000 \\
 & Complete-case LMM, IT & $-$0.081 & 0.034 & 1.02 & 0.354 \\
\bottomrule
\end{tabular}
\end{table}

\paragraph{Blocks I and II.} Tables~\ref{tab:supp_I} and~\ref{tab:supp_II} give
SA-AIPW for every target and number of clusters, with the sandwich (s) and
jackknife (j) standard errors, and the complete-case mixed model with
exposure-time indicators (ETI) for $\Psi(e)$ and $\bar\Psi$ and the
immediate-treatment model (IT) for $\PsiIT$. Table~\ref{tab:supp_estimators}
below adds the HC3 standard error and the estimators that SA-AIPW is
compared with.

\begin{table}[!htbp]
\centering
\caption{Block I (heterogeneous exposure-time effects): SA-AIPW and the
complete-case mixed model (LMM) by number of clusters.}
\label{tab:supp_I}
\footnotesize
\setlength{\tabcolsep}{4pt}
\begin{tabular}{rlrrrrrrrr}
\toprule
 & & \multicolumn{6}{c}{SA-AIPW} & \multicolumn{2}{c}{LMM}\\
\cmidrule(lr){3-8}\cmidrule(lr){9-10}
$J$ & Target & Bias & SD & SE/SD (s) & SE/SD (j) & Cov.\ (s) & Cov.\ (j) & Bias & Cov.\\
\midrule
12 & $\Psi(1)$ & $-$0.002 & 0.076 & 0.92 & 1.37 & 0.948 & 0.982 & $-$0.003 & 0.982 \\
 & $\Psi(2)$ & $-$0.006 & 0.095 & 0.91 & 1.35 & 0.930 & 0.990 & $-$0.003 & 0.982 \\
 & $\Psi(3)$ & $-$0.010 & 0.120 & 0.90 & 1.36 & 0.932 & 0.980 & $-$0.005 & 0.974 \\
 & $\Psi(4)$ & $-$0.004 & 0.150 & 0.95 & 1.50 & 0.934 & 0.986 & $-$0.004 & 0.986 \\
 & $\Psi(5)$ & 0.002 & 0.246 & 0.92 & 1.54 & 0.892 & 0.982 & 0.001 & 0.976 \\
 & $\bar\Psi$ & $-$0.004 & 0.113 & 0.90 & 1.48 & 0.916 & 0.982 & $-$0.003 & 0.970 \\
 & $\Psi_{\mathrm{IT}}$ & $-$0.005 & 0.083 & 0.90 & 1.40 & 0.938 & 0.990 & $-$0.082 & 0.738 \\
\addlinespace
18 & $\Psi(1)$ & $-$0.001 & 0.053 & 0.99 & 1.18 & 0.951 & 0.974 & $-$0.002 & 0.961 \\
 & $\Psi(2)$ & $-$0.002 & 0.065 & 0.98 & 1.18 & 0.943 & 0.967 & $-$0.003 & 0.955 \\
 & $\Psi(3)$ & 0.000 & 0.081 & 0.98 & 1.19 & 0.947 & 0.968 & $-$0.002 & 0.968 \\
 & $\Psi(4)$ & $-$0.004 & 0.107 & 0.99 & 1.24 & 0.951 & 0.980 & $-$0.005 & 0.966 \\
 & $\Psi(5)$ & $-$0.007 & 0.175 & 0.98 & 1.37 & 0.922 & 0.962 & $-$0.007 & 0.962 \\
 & $\bar\Psi$ & $-$0.003 & 0.077 & 0.97 & 1.27 & 0.938 & 0.976 & $-$0.004 & 0.957 \\
 & $\Psi_{\mathrm{IT}}$ & $-$0.003 & 0.057 & 0.97 & 1.19 & 0.955 & 0.979 & $-$0.081 & 0.621 \\
\addlinespace
30 & $\Psi(1)$ & $-$0.001 & 0.038 & 1.05 & 1.14 & 0.963 & 0.973 & $-$0.003 & 0.970 \\
 & $\Psi(2)$ & 0.001 & 0.048 & 0.98 & 1.07 & 0.949 & 0.961 & $-$0.002 & 0.962 \\
 & $\Psi(3)$ & $-$0.002 & 0.059 & 0.99 & 1.08 & 0.942 & 0.957 & $-$0.005 & 0.958 \\
 & $\Psi(4)$ & $-$0.001 & 0.077 & 1.00 & 1.12 & 0.946 & 0.964 & $-$0.004 & 0.966 \\
 & $\Psi(5)$ & $-$0.001 & 0.126 & 0.98 & 1.16 & 0.923 & 0.958 & $-$0.007 & 0.971 \\
 & $\bar\Psi$ & $-$0.001 & 0.055 & 0.99 & 1.14 & 0.946 & 0.968 & $-$0.004 & 0.966 \\
 & $\Psi_{\mathrm{IT}}$ & $-$0.002 & 0.041 & 1.02 & 1.13 & 0.956 & 0.965 & $-$0.081 & 0.354 \\
\addlinespace
48 & $\Psi(1)$ & 0.001 & 0.030 & 1.04 & 1.09 & 0.964 & 0.964 & $-$0.002 & 0.960 \\
 & $\Psi(2)$ & 0.002 & 0.039 & 0.97 & 1.01 & 0.952 & 0.952 & 0.000 & 0.972 \\
 & $\Psi(3)$ & 0.002 & 0.048 & 0.98 & 1.03 & 0.948 & 0.960 & 0.000 & 0.956 \\
 & $\Psi(4)$ & 0.001 & 0.061 & 0.98 & 1.05 & 0.928 & 0.952 & 0.001 & 0.972 \\
 & $\Psi(5)$ & 0.010 & 0.094 & 1.03 & 1.15 & 0.952 & 0.960 & 0.001 & 0.972 \\
 & $\bar\Psi$ & 0.003 & 0.044 & 0.99 & 1.07 & 0.948 & 0.952 & 0.000 & 0.964 \\
 & $\Psi_{\mathrm{IT}}$ & 0.000 & 0.033 & 1.00 & 1.06 & 0.960 & 0.972 & $-$0.080 & 0.184 \\
\bottomrule
\end{tabular}
\end{table}

\begin{table}[!htbp]
\centering
\caption{Block II (constant effect, $\Psi(e)=0.12$ for all $e$): SA-AIPW and
the complete-case mixed model (LMM) by number of clusters.}
\label{tab:supp_II}
\footnotesize
\setlength{\tabcolsep}{4pt}
\begin{tabular}{rlrrrrrrrr}
\toprule
 & & \multicolumn{6}{c}{SA-AIPW} & \multicolumn{2}{c}{LMM}\\
\cmidrule(lr){3-8}\cmidrule(lr){9-10}
$J$ & Target & Bias & SD & SE/SD (s) & SE/SD (j) & Cov.\ (s) & Cov.\ (j) & Bias & Cov.\\
\midrule
12 & $\Psi(1)$ & $-$0.001 & 0.070 & 0.99 & 1.44 & 0.952 & 0.978 & $-$0.004 & 0.972 \\
 & $\Psi(2)$ & $-$0.001 & 0.086 & 0.99 & 1.45 & 0.938 & 0.970 & $-$0.005 & 0.980 \\
 & $\Psi(3)$ & $-$0.005 & 0.108 & 0.99 & 1.47 & 0.946 & 0.988 & $-$0.007 & 0.958 \\
 & $\Psi(4)$ & $-$0.005 & 0.145 & 0.98 & 1.52 & 0.934 & 0.988 & $-$0.005 & 0.980 \\
 & $\Psi(5)$ & $-$0.005 & 0.238 & 0.98 & 1.57 & 0.906 & 0.982 & $-$0.004 & 0.980 \\
 & $\bar\Psi$ & $-$0.002 & 0.104 & 0.99 & 1.59 & 0.934 & 0.996 & $-$0.005 & 0.978 \\
 & $\Psi_{\mathrm{IT}}$ & $-$0.003 & 0.076 & 0.96 & 1.51 & 0.934 & 0.986 & $-$0.005 & 0.980 \\
\addlinespace
18 & $\Psi(1)$ & 0.001 & 0.055 & 0.94 & 1.12 & 0.924 & 0.970 & $-$0.001 & 0.952 \\
 & $\Psi(2)$ & 0.000 & 0.064 & 0.98 & 1.17 & 0.952 & 0.972 & $-$0.001 & 0.964 \\
 & $\Psi(3)$ & 0.001 & 0.081 & 0.97 & 1.17 & 0.958 & 0.984 & $-$0.001 & 0.950 \\
 & $\Psi(4)$ & $-$0.002 & 0.105 & 0.98 & 1.24 & 0.944 & 0.976 & $-$0.005 & 0.960 \\
 & $\Psi(5)$ & $-$0.004 & 0.177 & 0.96 & 1.36 & 0.916 & 0.962 & $-$0.002 & 0.958 \\
 & $\bar\Psi$ & $-$0.001 & 0.078 & 0.95 & 1.26 & 0.930 & 0.976 & $-$0.002 & 0.960 \\
 & $\Psi_{\mathrm{IT}}$ & 0.000 & 0.058 & 0.92 & 1.14 & 0.936 & 0.976 & 0.000 & 0.964 \\
\addlinespace
30 & $\Psi(1)$ & 0.001 & 0.040 & 0.99 & 1.08 & 0.960 & 0.970 & $-$0.001 & 0.960 \\
 & $\Psi(2)$ & 0.004 & 0.046 & 1.02 & 1.10 & 0.952 & 0.960 & 0.002 & 0.952 \\
 & $\Psi(3)$ & 0.000 & 0.057 & 1.04 & 1.14 & 0.958 & 0.964 & $-$0.002 & 0.958 \\
 & $\Psi(4)$ & 0.000 & 0.080 & 0.97 & 1.09 & 0.944 & 0.968 & $-$0.002 & 0.938 \\
 & $\Psi(5)$ & 0.003 & 0.119 & 1.06 & 1.27 & 0.960 & 0.982 & $-$0.001 & 0.962 \\
 & $\bar\Psi$ & 0.002 & 0.054 & 1.02 & 1.17 & 0.948 & 0.970 & $-$0.001 & 0.956 \\
 & $\Psi_{\mathrm{IT}}$ & 0.002 & 0.041 & 0.99 & 1.10 & 0.948 & 0.964 & 0.000 & 0.960 \\
\bottomrule
\end{tabular}
\end{table}

\paragraph{Block III.} Table~\ref{tab:supp_III} gives SA-AIPW and the same
estimator applied to the uncensored outcomes (full data) at $J=30$; standard
errors are sandwich.

\begin{table}[!htbp]
\centering
\caption{Block III (censoring burden, $J=30$): SA-AIPW and the full-data
estimator by outcome noise $\sigma^2_Y$ and late-period censoring $\kappa$.}
\label{tab:supp_III}
\footnotesize
\begin{tabular}{cclrrrrrr}
\toprule
 & & & \multicolumn{4}{c}{SA-AIPW} & \multicolumn{2}{c}{Full data}\\
\cmidrule(lr){4-7}\cmidrule(lr){8-9}
$\sigma^2_Y$ & $\kappa$ & Target & Bias & SD & SE/SD & Cov. & Bias & SD\\
\midrule
0.64 & 1.0 & $\Psi(1)$ & $-$0.001 & 0.038 & 1.00 & 0.956 & $-$0.001 & 0.038 \\
 &  & $\Psi(3)$ & $-$0.001 & 0.057 & 0.99 & 0.946 & $-$0.001 & 0.056 \\
 &  & $\Psi(5)$ & 0.000 & 0.112 & 1.02 & 0.952 & $-$0.001 & 0.109 \\
 &  & $\bar\Psi$ & 0.000 & 0.051 & 1.00 & 0.944 & 0.000 & 0.050 \\
 &  & $\Psi_{\mathrm{IT}}$ & $-$0.002 & 0.041 & 1.00 & 0.950 & $-$0.002 & 0.040 \\
\addlinespace
0.64 & 2.5 & $\Psi(1)$ & 0.000 & 0.039 & 1.01 & 0.956 & 0.000 & 0.037 \\
 &  & $\Psi(3)$ & $-$0.001 & 0.059 & 0.99 & 0.950 & $-$0.001 & 0.055 \\
 &  & $\Psi(5)$ & $-$0.003 & 0.126 & 1.00 & 0.954 & $-$0.004 & 0.112 \\
 &  & $\bar\Psi$ & $-$0.001 & 0.055 & 1.00 & 0.944 & $-$0.001 & 0.051 \\
 &  & $\Psi_{\mathrm{IT}}$ & 0.000 & 0.042 & 1.01 & 0.958 & 0.000 & 0.040 \\
\addlinespace
0.64 & 4.0 & $\Psi(1)$ & 0.004 & 0.047 & 0.96 & 0.938 & 0.002 & 0.039 \\
 &  & $\Psi(3)$ & 0.001 & 0.070 & 0.98 & 0.938 & $-$0.001 & 0.055 \\
 &  & $\Psi(5)$ & 0.012 & 0.174 & 0.93 & 0.918 & 0.005 & 0.106 \\
 &  & $\bar\Psi$ & 0.005 & 0.068 & 1.00 & 0.946 & 0.002 & 0.048 \\
 &  & $\Psi_{\mathrm{IT}}$ & 0.001 & 0.046 & 0.99 & 0.948 & 0.000 & 0.041 \\
\addlinespace
0.16 & 1.0 & $\Psi(1)$ & $-$0.002 & 0.023 & 0.99 & 0.952 & $-$0.002 & 0.022 \\
 &  & $\Psi(3)$ & 0.001 & 0.035 & 0.97 & 0.948 & 0.001 & 0.034 \\
 &  & $\Psi(5)$ & $-$0.001 & 0.069 & 1.10 & 0.950 & $-$0.003 & 0.067 \\
 &  & $\bar\Psi$ & 0.000 & 0.032 & 1.03 & 0.950 & 0.000 & 0.031 \\
 &  & $\Psi_{\mathrm{IT}}$ & $-$0.002 & 0.027 & 1.01 & 0.940 & $-$0.002 & 0.026 \\
\addlinespace
0.16 & 2.5 & $\Psi(1)$ & 0.001 & 0.024 & 0.96 & 0.940 & 0.001 & 0.023 \\
 &  & $\Psi(3)$ & 0.000 & 0.034 & 1.03 & 0.970 & 0.000 & 0.032 \\
 &  & $\Psi(5)$ & 0.003 & 0.072 & 1.09 & 0.952 & 0.004 & 0.065 \\
 &  & $\bar\Psi$ & 0.001 & 0.033 & 1.04 & 0.962 & 0.001 & 0.029 \\
 &  & $\Psi_{\mathrm{IT}}$ & $-$0.001 & 0.026 & 1.03 & 0.956 & $-$0.001 & 0.025 \\
\addlinespace
0.16 & 4.0 & $\Psi(1)$ & $-$0.001 & 0.027 & 0.98 & 0.964 & $-$0.001 & 0.023 \\
 &  & $\Psi(3)$ & $-$0.001 & 0.040 & 1.01 & 0.956 & $-$0.001 & 0.032 \\
 &  & $\Psi(5)$ & $-$0.003 & 0.101 & 1.01 & 0.950 & $-$0.001 & 0.066 \\
 &  & $\bar\Psi$ & $-$0.001 & 0.044 & 0.97 & 0.956 & 0.000 & 0.031 \\
 &  & $\Psi_{\mathrm{IT}}$ & $-$0.002 & 0.029 & 1.02 & 0.958 & $-$0.002 & 0.025 \\
\bottomrule
\end{tabular}
\end{table}

\paragraph{The working model of the earlier versions.} Before the index of
Proposition~\ref{M-prop:history}(iii), we modeled $\varrho$ for each $m$
separately by a cross-fitted regression of the residuals on the cluster
means of the surrogate and observed outcome over periods $1,\ldots,m$, on top
of an estimator that carried the period-1 means of both as covariates
(\texttt{history\_augment}). Table~\ref{tab:supp_hist} compares that
estimator with and without the history term in 500 replications of the
Block~I design with $J=30$, and reports two large-sample quantities from
five data sets with $J=900$: the ratio of the variances of the cluster-level
estimating functions with and without the term, and the standard deviation
that the history-augmented estimating function implies at $J=30$, an upper
bound for the bound of the proposition. With 30 clusters that term did not
reduce the standard deviation, and at $J=900$ it reduced the variance by
2\% to 6\%, against the 9\% to 23\% that the bound allows: the per-$m$
regressions are fitted to the few clusters still unexposed at the later $m$,
and the raw means mix the cluster effect with covariate and period
variation. The index of part (iii) uses the same information in a
one-dimensional form with the shrinkage weights estimated from all
pre-crossover records, which is why it attains the bound
(Table~\ref{tab:supp_estimators}). Produced by
\nolinkurl{inst/sim/run_history.R}.

\begin{table}[!htbp]
\centering
\caption{History augmentation in the Block I design. Columns 2--8: 500
replications with $J=30$ (SA-AIPW and its history-augmented version with
sandwich standard errors; SD of the complete-case exposure-time mixed model).
Last two columns: variance ratio of the estimating functions with and without
the history term, and the implied standard deviation at $J=30$, from five data
sets with $J=900$.}
\label{tab:supp_hist}
\footnotesize
\setlength{\tabcolsep}{4pt}
\begin{tabular}{lrrrrrrrrr}
\toprule
 & \multicolumn{3}{c}{SA-AIPW} & \multicolumn{3}{c}{With history} & LMM & \multicolumn{2}{c}{$J=900$}\\
\cmidrule(lr){2-4}\cmidrule(lr){5-7}\cmidrule(lr){9-10}
Target & SD & SE/SD & Cov. & SD & SE/SD & Cov. & SD & Var.\ ratio & SD at $J=30$\\
\midrule
$\Psi(1)$ & 0.051 & 0.99 & 0.940 & 0.053 & 0.92 & 0.910 & 0.041 & 0.95 & 0.046 \\
$\Psi(2)$ & 0.063 & 1.00 & 0.954 & 0.064 & 0.95 & 0.942 & 0.048 & 0.94 & 0.058 \\
$\Psi(3)$ & 0.074 & 1.04 & 0.968 & 0.076 & 0.99 & 0.946 & 0.056 & 0.95 & 0.071 \\
$\Psi(4)$ & 0.095 & 1.03 & 0.958 & 0.097 & 0.97 & 0.936 & 0.072 & 0.97 & 0.090 \\
$\Psi(5)$ & 0.140 & 1.04 & 0.932 & 0.144 & 0.96 & 0.924 & 0.092 & 0.98 & 0.135 \\
$\bar\Psi$ & 0.072 & 1.02 & 0.952 & 0.074 & 0.94 & 0.932 & 0.052 & 0.95 & 0.067 \\
$\Psi_{\mathrm{IT}}$ & 0.057 & 1.00 & 0.956 & 0.059 & 0.92 & 0.922 & 0.036 & 0.96 & 0.052 \\
\bottomrule
\end{tabular}
\end{table}

\paragraph{Block IV.} Table~\ref{tab:supp_IV} gives SA-AIPW and IPCW at $J=30$
when the outcome regression omits the exposure-time terms and uses a linear
calendar trend, when the observation model omits the surrogate, and when
both are misspecified. The first row group repeats Block~I with both models
correct; standard errors are sandwich.

\begin{table}[!htbp]
\centering
\caption{Block IV (misspecified working models, $J=30$).}
\label{tab:supp_IV}
\footnotesize
\begin{tabular}{llrrrrr}
\toprule
 & & \multicolumn{3}{c}{SA-AIPW} & \multicolumn{2}{c}{IPCW}\\
\cmidrule(lr){3-5}\cmidrule(lr){6-7}
Working models & Target & Bias & SD & Cov. & Bias & Cov.\\
\midrule
Both correct & $\Psi(1)$ & $-$0.001 & 0.038 & 0.963 & $-$0.001 & 1.000 \\
 & $\Psi(3)$ & $-$0.002 & 0.059 & 0.942 & $-$0.004 & 1.000 \\
 & $\Psi(5)$ & $-$0.001 & 0.126 & 0.923 & 0.001 & 1.000 \\
 & $\bar\Psi$ & $-$0.001 & 0.055 & 0.946 & $-$0.001 & 1.000 \\
 & $\Psi_{\mathrm{IT}}$ & $-$0.002 & 0.041 & 0.956 & $-$0.001 & 1.000 \\
\addlinespace
Outcome model wrong & $\Psi(1)$ & $-$0.001 & 0.039 & 0.944 & $-$0.002 & 1.000 \\
 & $\Psi(3)$ & $-$0.002 & 0.060 & 0.950 & $-$0.007 & 1.000 \\
 & $\Psi(5)$ & $-$0.002 & 0.124 & 0.932 & $-$0.010 & 1.000 \\
 & $\bar\Psi$ & $-$0.002 & 0.055 & 0.946 & $-$0.007 & 1.000 \\
 & $\Psi_{\mathrm{IT}}$ & $-$0.003 & 0.041 & 0.948 & $-$0.005 & 1.000 \\
\addlinespace
Observation model wrong & $\Psi(1)$ & $-$0.001 & 0.042 & 0.944 & $-$0.012 & 1.000 \\
 & $\Psi(3)$ & 0.000 & 0.061 & 0.956 & $-$0.008 & 1.000 \\
 & $\Psi(5)$ & 0.000 & 0.130 & 0.942 & $-$0.011 & 1.000 \\
 & $\bar\Psi$ & $-$0.001 & 0.059 & 0.940 & $-$0.010 & 1.000 \\
 & $\Psi_{\mathrm{IT}}$ & $-$0.003 & 0.044 & 0.940 & $-$0.009 & 1.000 \\
\addlinespace
Both wrong & $\Psi(1)$ & 0.000 & 0.039 & 0.962 & $-$0.005 & 1.000 \\
 & $\Psi(3)$ & 0.002 & 0.059 & 0.942 & 0.001 & 1.000 \\
 & $\Psi(5)$ & 0.012 & 0.128 & 0.926 & 0.001 & 1.000 \\
 & $\bar\Psi$ & 0.004 & 0.054 & 0.944 & 0.000 & 1.000 \\
 & $\Psi_{\mathrm{IT}}$ & 0.000 & 0.040 & 0.964 & 0.000 & 1.000 \\
\bottomrule
\end{tabular}
\end{table}

\paragraph{Block V.} Table~\ref{tab:alloc} compares
fixed-wave with independent crossover times at $J=30$. The standard deviations
of the augmented estimator were nearly the same under the two schemes (0.042
and 0.040 for $\Psi(1)$, 0.128 and 0.138 for $\Psi(5)$, and 0.059 and 0.056
for $\bar\Psi$), and its variance estimators were reasonably accurate under
both, with SE/SD between 0.90 and 1.20. The unaugmented estimator, without
the history term, behaved as Proposition~\ref{M-prop:alloc} predicts. With
independent crossover times its standard deviations were 2.4 to 2.9 times
larger than with fixed waves, and its sandwich was accurate. With fixed
waves its variance was smaller, but its sandwich overstated its standard
deviation by a factor of 2.5 to 3.0, because it estimates the variance under
independent allocation; the jackknife, which was not computed for this
estimator, behaves in the same way by Proposition~\ref{M-prop:alloc}.

\begin{table}[!htbp]
\centering
\caption{Block V, $J=30$: Monte Carlo standard deviation (SD) and SE/SD for the
sandwich (s) and jackknife (j) under fixed-wave and independent crossover
times; the unaugmented estimator is without the history term and its
jackknife was not computed.}
\label{tab:alloc}
\small
\setlength{\tabcolsep}{4.5pt}
\begin{tabular}{llccccccccc}
\toprule
 & & \multicolumn{3}{c}{$\Psi(1)$} & \multicolumn{3}{c}{$\Psi(5)$} & \multicolumn{3}{c}{$\bar\Psi$}\\
\cmidrule(lr){3-5}\cmidrule(lr){6-8}\cmidrule(lr){9-11}
Allocation & Estimator & SD & s & j & SD & s & j & SD & s & j\\
\midrule
Fixed waves & Augmented & 0.042 & 0.95 & 1.04 & 0.128 & 0.98 & 1.19 & 0.059 & 0.93 & 1.08 \\
 & Unaugmented & 0.073 & 2.48 & -- & 0.185 & 2.98 & -- & 0.107 & 3.03 & -- \\
Independent & Augmented & 0.040 & 1.00 & 1.11 & 0.138 & 0.90 & 1.20 & 0.056 & 0.98 & 1.19 \\
 & Unaugmented & 0.175 & 1.00 & -- & 0.527 & 1.03 & -- & 0.295 & 1.07 & -- \\
\bottomrule
\end{tabular}
\end{table}

\paragraph{Block VI.} Table~\ref{tab:blockVI} gives the bias and coverage
summarized in Section~\ref{M-sec:sim:results} for the scenarios in which the
surrogate effect varies with calendar time.

\begin{table}[!htbp]
\centering
\caption{Block VI: bias and coverage when the surrogate effect varies with
calendar time, (a) with the default censoring and (b) with strongly
surrogate-dependent censoring. SA-AIPW uses sandwich standard errors; the
mixed models (LMM) are as in Table~\ref{M-tab:sim_main}.}
\label{tab:blockVI}
\small
\setlength{\tabcolsep}{4pt}
\begin{tabular}{llrrrrrrrr}
\toprule
 & & \multicolumn{2}{c}{$\Psi(1)$} & \multicolumn{2}{c}{$\Psi(5)$} & \multicolumn{2}{c}{$\bar\Psi$} & \multicolumn{2}{c}{$\PsiIT$}\\
\cmidrule(lr){3-4}\cmidrule(lr){5-6}\cmidrule(lr){7-8}\cmidrule(lr){9-10}
Scenario & Estimator & Bias & Cov. & Bias & Cov. & Bias & Cov. & Bias & Cov.\\
\midrule
(a), $J=30$ & SA-AIPW & $-$0.003 & 0.962 & $-$0.003 & 0.954 & $-$0.002 & 0.954 & $-$0.002 & 0.964 \\
 & Complete-case LMM & $-$0.009 & 0.968 & $-$0.036 & 0.946 & $-$0.017 & 0.958 & $-$0.107 & 0.144 \\
\addlinespace
(b), $J=30$ & SA-AIPW & 0.000 & 0.956 & $-$0.001 & 0.942 & 0.000 & 0.938 & $-$0.001 & 0.956 \\
 & Complete-case LMM & $-$0.017 & 0.936 & $-$0.071 & 0.874 & $-$0.038 & 0.878 & $-$0.111 & 0.144 \\
\addlinespace
(b), $J=48$ & SA-AIPW & $-$0.001 & 0.940 & $-$0.004 & 0.942 & $-$0.001 & 0.938 & 0.000 & 0.946 \\
 & Complete-case LMM & $-$0.016 & 0.916 & $-$0.067 & 0.862 & $-$0.035 & 0.860 & $-$0.109 & 0.026 \\
\bottomrule
\end{tabular}
\end{table}

\begin{table}[!htbp]
\centering
\caption{Allocations chosen by Proposition~\ref{M-prop:design} for $J=30$, with
at least two clusters per wave, and their predicted variance relative to the
balanced allocation, at three levels of late-period censoring $\kappa$.
Allocations list the numbers of clusters crossing over at periods 2 to 7.
Inputs were estimated from one data set with $J=900$ at each $\kappa$.}
\label{tab:design}
\small
\setlength{\tabcolsep}{4.5pt}
\begin{tabular}{lllcccc}
\toprule
 & & & \multicolumn{4}{c}{Relative variance}\\
\cmidrule(lr){4-7}
$\kappa$ & Criterion & Allocation & $\Psi(1)$ & $\Psi(3)$ & $\Psi(5)$ & $\bar\Psi$\\
\midrule
1.0 & Balanced & 5, 5, 5, 5, 5, 5 & 1.00 & 1.00 & 1.00 & 1.00 \\
 & Optimal for $\Psi(5)$ & 11, 2, 2, 2, 2, 11 & 1.67 & 1.24 & 0.45 & 0.84 \\
 & Optimal for $\bar\Psi$ & 7, 4, 3, 2, 2, 12 & 1.40 & 0.94 & 0.56 & 0.75 \\
\addlinespace
2.5 & Balanced & 5, 5, 5, 5, 5, 5 & 1.00 & 1.00 & 1.00 & 1.00 \\
 & Optimal for $\Psi(5)$ & 11, 2, 2, 2, 2, 11 & 1.66 & 1.24 & 0.45 & 0.83 \\
 & Optimal for $\bar\Psi$ & 7, 4, 3, 2, 2, 12 & 1.39 & 0.94 & 0.56 & 0.74 \\
\addlinespace
4.0 & Balanced & 5, 5, 5, 5, 5, 5 & 1.00 & 1.00 & 1.00 & 1.00 \\
 & Optimal for $\Psi(5)$ & 11, 2, 2, 2, 2, 11 & 1.62 & 1.24 & 0.45 & 0.79 \\
 & Optimal for $\bar\Psi$ & 7, 4, 2, 2, 2, 13 & 1.47 & 1.06 & 0.54 & 0.71 \\
\bottomrule
\end{tabular}
\end{table}

\paragraph{Allocation.} Table~\ref{tab:design} lists the allocations chosen by
Proposition~\ref{M-prop:design} at each censoring level, which
Section~\ref{M-sec:sim:results} summarizes. Table~\ref{tab:supp_design}
checks the proposition with $\kappa=4$ and $J=30$. For the balanced allocation
and the two allocations of Table~\ref{tab:design} chosen for this $\kappa$, it gives the standard
deviation predicted by~\eqref{M-eq:design_var} (with inputs from one data set
with $J=900$), the simulated standard deviation over 500 data sets per
allocation, generated with common seeds, and the coverage of sandwich
intervals, for the estimator without the history term, which is what the
criterion models; the last column of each group gives the simulated
standard deviation of the estimator with the history term. The predictions
are within 8\% of the simulated values, a little high for every target. The
gains carry over to the estimator with the history term: the allocation
chosen for $\Psi(5)$ lowers its variance to 0.38 times the balanced value
(0.42 without the term), and the allocation chosen for $\bar\Psi$ to 0.55
(0.71). The immediate-treatment estimand is left out because its weights
change with the allocation (Proposition~\ref{M-prop:it}). Produced by
\nolinkurl{inst/sim/run_design.R}, which also produces
Table~\ref{tab:design}.

\begin{table}[!htbp]
\centering
\caption{Block~I design with $\kappa=4$, $J=30$: predicted and simulated
standard deviations of SA-AIPW without the history term, which the criterion
of Proposition~\ref{M-prop:design} models, coverage of its 95\% sandwich
intervals, and the simulated standard deviation of the estimator with the
history term (index), under three allocations (numbers of clusters crossing
over at periods 2 to 7).}
\label{tab:supp_design}
\scriptsize
\setlength{\tabcolsep}{2.2pt}
\begin{tabular}{llcccccccccccc}
\toprule
 & & \multicolumn{4}{c}{$\Psi(1)$} & \multicolumn{4}{c}{$\Psi(5)$} & \multicolumn{4}{c}{$\bar\Psi$}\\
\cmidrule(lr){3-6}\cmidrule(lr){7-10}\cmidrule(lr){11-14}
Allocation & Clusters & Pred. & SD & Cov. & Index & Pred. & SD & Cov. & Index & Pred. & SD & Cov. & Index\\
\midrule
Balanced & 5, 5, 5, 5, 5, 5 & 0.073 & 0.069 & 0.960 & 0.048 & 0.219 & 0.216 & 0.942 & 0.179 & 0.115 & 0.107 & 0.950 & 0.073 \\
Optimal for $\Psi(5)$ & 11, 2, 2, 2, 2, 11 & 0.093 & 0.092 & 0.944 & 0.058 & 0.148 & 0.140 & 0.970 & 0.111 & 0.102 & 0.096 & 0.958 & 0.055 \\
Optimal for $\bar\Psi$ & 7, 4, 2, 2, 2, 13 & 0.089 & 0.082 & 0.976 & 0.055 & 0.160 & 0.155 & 0.960 & 0.129 & 0.096 & 0.090 & 0.958 & 0.054 \\
\bottomrule
\end{tabular}
\end{table}

\paragraph{Implementation choices.} Table~\ref{tab:supp_ablation} reruns
SA-AIPW on the 1000 data sets of Block~I with $J=30$ under six
specifications: the estimator of the paper (record-level regressions, bridge
pooled over exposure levels, history index from the surrogate and outcome
histories); the same without the history term; the estimator of the earlier
versions, with the period-1 cluster means of the surrogate and observed
outcome as covariates in the outcome and bridge regressions and no history
term; that estimator with the index added; the estimator of the paper with
a separate bridge for each exposure level; and the estimator of the paper
with the index built from the surrogate histories alone. The first row
reproduces Table~\ref{M-tab:sim_main}; the estimates agree to machine
precision. All versions are close to unbiased, as
Proposition~\ref{M-prop:dml_remainder}(b) implies, and differ in precision.
The index lowers the standard deviations by 31\% to 47\% relative to the
estimator without the history term and by 9\% to 23\% relative to the
earlier estimator; added to the earlier estimator it changes almost nothing,
because the cluster-level columns in its regressions have already absorbed,
with cluster-level noise, most of what the index would recover, and its
sandwich then understates the standard deviation (coverage 0.91 to 0.95).
Separate bridges, which in the earlier versions nearly doubled the standard
deviation of $\hat\Psi(5)$ because the bridge for exposure time 5 was fitted
with cluster-level covariates to the clusters of the first wave alone, are
now harmless (0.128 against 0.126), so pooling is a matter of convenience.
The surrogate histories alone recover most of the gain (standard deviations
7\% to 14\% above those with both histories), which is the option when the
observation of the pre-crossover outcomes depends on the wave and the
weights cannot be fixed. The last row is the complete-case exposure-time
mixed model on the same data sets. Produced by
\nolinkurl{inst/sim/run_ablation.R}.

\begin{table}[!htbp]
\centering
\caption{Block~I design, $J=30$, 1000 data sets: SA-AIPW under six
specifications (see the text) and the complete-case exposure-time mixed model
(LMM). Bias, Monte Carlo standard deviation (SD) and coverage of 95\%
intervals (sandwich for SA-AIPW, model-based for the LMM).}
\label{tab:supp_ablation}
\scriptsize
\setlength{\tabcolsep}{2.6pt}
\begin{tabular}{lrrrrrrrrr}
\toprule
 & \multicolumn{3}{c}{$\Psi(1)$} & \multicolumn{3}{c}{$\Psi(5)$} & \multicolumn{3}{c}{$\bar\Psi$}\\
\cmidrule(lr){2-4}\cmidrule(lr){5-7}\cmidrule(lr){8-10}
Specification & Bias & SD & Cov. & Bias & SD & Cov. & Bias & SD & Cov.\\
\midrule
Index, pooled bridge (paper) & $-$0.001 & 0.038 & 0.963 & $-$0.001 & 0.126 & 0.923 & $-$0.001 & 0.055 & 0.946 \\
No history term, pooled bridge & 0.000 & 0.066 & 0.960 & $-$0.001 & 0.183 & 0.950 & $-$0.001 & 0.103 & 0.942 \\
Period-1 means, no index (earlier versions) & $-$0.001 & 0.047 & 0.967 & 0.000 & 0.139 & 0.945 & $-$0.001 & 0.070 & 0.947 \\
Period-1 means and the index & $-$0.001 & 0.046 & 0.953 & 0.000 & 0.141 & 0.908 & 0.000 & 0.069 & 0.923 \\
Index, separate bridges & $-$0.001 & 0.038 & 0.976 & $-$0.002 & 0.128 & 0.932 & $-$0.001 & 0.056 & 0.968 \\
Index, surrogate histories only & $-$0.002 & 0.041 & 0.959 & $-$0.003 & 0.135 & 0.932 & $-$0.002 & 0.062 & 0.942 \\
Complete-case LMM, ETI & $-$0.003 & 0.036 & 0.970 & $-$0.007 & 0.083 & 0.971 & $-$0.004 & 0.047 & 0.966 \\
\bottomrule
\end{tabular}
\end{table}

\paragraph{Cluster-level covariates with few clusters.} The estimator of the
earlier versions carried the period-1 cluster means of the surrogate and
outcome as covariates, and the paragraph on baseline information in
Section~\ref{M-sec:tmle} suggests adding those of the individual covariates
as well. On the first 500 data sets of Block~I with $J=30$, adding the three
period-1 means of $W_1$, $W_2$ and $W_3$ to that estimator's outcome and
bridge regressions lowered its standard deviations by 5\% to 9\% with bias
and coverage unchanged, and on five data sets with $J=900$ it lowered the
variance of its cluster-level estimating function by 17\% to 25\%, to the
bound of the model that treats the period-1 data as baseline
(Table~\ref{tab:oracle}). With fewer clusters the picture is different. The
three covariate means are cluster-level columns, and with them the outcome
and bridge regressions have six cluster-level parameters, including the
intercept. With $J=12$ and two folds each training set has six clusters, so
the cluster-level part of the regressions interpolates the training clusters
and its out-of-fold predictions are unstable: on the same 500 data sets the
median absolute error of $\hat\Psi(1)$ rose from 0.06 to 0.18, and 10\% to
28\% of the estimates were off by more than 1. In every leave-one-cluster-out
refit one cluster-level column is then aliased, and \texttt{lm.fit} returns
a coefficient of \texttt{NA} that the software sets to zero. With $J=18$ and
three folds (twelve training clusters) the covariate means changed the
standard deviations by less than 2\%, left the sandwich accurate, and made
the jackknife markedly conservative (SE/SD 1.45 to 1.53). This is why the
estimator of the paper keeps cluster-level information out of the
record-level regressions and lets it enter through the one-dimensional
index, whose single loading is fitted at the cluster level
(Table~\ref{tab:supp_estimators}). Produced by
\nolinkurl{inst/sim/run_wbar_check.R} and \nolinkurl{inst/sim/run_wbar_smallJ.R}.

\paragraph{Exact efficiency bounds.} In the Gaussian design of
Section~\ref{M-sec:sim:design} the nuisance functions have closed forms:
$\Qint(k,W,t)=\beta(k)+0.5\theta(k)+0.6W_1-0.35W_2+0.6W_3+1.05\lambda(t)$;
$\QY(s,k,W,t)=\beta(k)+0.5s+0.4W_1-0.2W_2+0.6W_3+0.8\lambda(t)+\kappa_b r$,
where $r=s-\theta(k)-0.4W_1+0.3W_2-0.5\lambda(t)$ is the surrogate residual
and $\kappa_b=0.6\sigma_b^2/(0.36\sigma_b^2+0.25)$ is the posterior slope of
$b_j$ on it; and $g_\Delta$ is the logistic model of the generator. The residual
$Y_{ijt}(k)-\Qint(k,W_{ijt},t)$ equals $1.3b_j+0.5\epsilon_{S,ijt}+\epsilon_{Y,ijt}$,
so for a record of a period after $m$ the conditional mean in
$\alpha^*_m$ is $1.3\,\EE(b_j\mid\mathcal H_{jm})$, and $\EE(b_j\mid\mathcal
H_{jm})$ is the Gaussian posterior mean of $b_j$ given the surrogate residuals
of all records and the outcome residuals of the observed records of periods
$1$ to $m$, whose precision is $\sigma_b^{-2}+0.36N_S/0.25+N_Y/0.64$ with $N_S$
and $N_Y$ the numbers of records and of observed outcomes. Because the
observation indicators depend on the data only through $S$, $A$ and $W_3$,
conditioning on which outcomes are observed does not bias this mean. The
efficient influence functions can therefore be evaluated exactly, and
Table~\ref{tab:oracle} reports their variances, estimated from 20{,}000
simulated clusters, as standard deviations at $J=30$, for three models: the
model of Lemma~\ref{M-lem:summation} with the individual covariates as
baseline information; the same model with the period-1 data treated as
baseline, so that $\Qint$ and $\QY$ condition on $\EE(b_j\mid\mathcal
H_{j1})$ as well; and $\mathcal M_S$, whose canonical gradient is the first
influence function minus $\sum_m\zeta_{jm}\alpha^*_m(\mathcal H_{jm})$.
Numerically, the projection was orthogonal to the result (correlations below
0.005 for every $m$) and the variance identity of
Proposition~\ref{M-prop:history}(i) held to Monte Carlo accuracy. Beside the
bounds the table gives the Monte Carlo standard deviations at $J=30$ of the
estimator of the paper, of the same estimator without the history term, and
of the estimator of the earlier versions with the period-1 means of $S$ and
$Y$ as covariates, together with the standard deviations implied at $J=30$
by the estimating functions at $J=900$ of the estimator of the paper and of
the earlier estimator with the period-1 means of $S$, $Y$ and $W$.

Three things follow. The working estimators attain the bounds of the models
they use: without cluster-level information SA-AIPW matches the first bound,
the earlier estimator with the period-1 means of $S$, $Y$ and $W$ matches
the second at $J=900$, and the estimator of the paper matches the bound of
$\mathcal M_S$ at $J=900$ and, to within 8\%, with 30 clusters, where its
standard deviations are 0.038, 0.048, 0.059, 0.077, 0.126 and 0.055 against
bounds of 0.039, 0.046, 0.057, 0.074, 0.117 and 0.053. The pre-crossover
histories lower the bound by 9\% to 23\% in variance beyond the period-1
data, and the index recovers this where the per-$m$ working model of the
earlier versions recovered 1\% to 3\%. And the complete-case mixed model with
exposure-time indicators, whose mean structure is correctly specified here,
has standard deviations 2\% to 11\% below the bound of $\mathcal M_S$ for
$\Psi(1)$ to $\Psi(4)$ and $\bar\Psi$ and 29\% below it for $\Psi(5)$ in the
1000 replications of Table~\ref{tab:supp_main_full}; the next paragraph
computes its precision exactly. The immediate-treatment estimand is left out
of this comparison because the immediate-treatment mixed model estimates a
different quantity (Table~\ref{M-tab:sim_main}). Produced by
\nolinkurl{inst/sim/oracle_bounds.R} and \nolinkurl{inst/sim/run_index_J900.R}.

\begin{table}[!htbp]
\centering
\caption{Block~I design: exact efficiency bounds, expressed as standard
deviations at $J=30$, for the cluster-level model of Lemma~\ref{M-lem:summation}
with the individual covariates as baseline information ($\mathcal M_J$), the
same model with the period-1 data as baseline ($\mathcal M_J$, P1) and the
model $\mathcal M_S$ of Proposition~\ref{M-prop:history}; Monte Carlo standard
deviations at $J=30$ (1000 replications) of SA-AIPW without cluster-level
information, of the estimator of the earlier versions with the period-1
means of $S$ and $Y$ as covariates, and of the estimator of the paper with
the history index; the standard deviation at $J=30$ implied by the
estimating function at $J=900$ of the earlier estimator with the period-1
means of $S$, $Y$ and $W$ and of the estimator of the paper; and the Monte
Carlo standard deviation of the complete-case exposure-time mixed model in
the same 1000 replications (Table~\ref{tab:supp_main_full}).}
\label{tab:oracle}
\footnotesize
\setlength{\tabcolsep}{3.5pt}
\begin{tabular}{lccccccccc}
\toprule
 & \multicolumn{3}{c}{Exact bound} & \multicolumn{3}{c}{SA-AIPW, $J=30$} & \multicolumn{2}{c}{$J=900$} & LMM\\
\cmidrule(lr){2-4}\cmidrule(lr){5-7}\cmidrule(lr){8-9}\cmidrule(lr){10-10}
Target & $\mathcal M_J$ & $\mathcal M_J$, P1 & $\mathcal M_S$ & no cluster inf. & $S$, $Y$ means & index & $S$, $Y$, $W$ means & index & ETI\\
\midrule
$\Psi(1)$ & 0.068 & 0.043 & 0.039 & 0.066 & 0.047 & 0.038 & 0.043 & 0.038 & 0.036 \\
$\Psi(2)$ & 0.087 & 0.053 & 0.046 & 0.086 & 0.060 & 0.048 & 0.053 & 0.047 & 0.045 \\
$\Psi(3)$ & 0.108 & 0.065 & 0.057 & 0.110 & 0.077 & 0.059 & 0.064 & 0.057 & 0.055 \\
$\Psi(4)$ & 0.130 & 0.082 & 0.074 & 0.130 & 0.092 & 0.077 & 0.082 & 0.075 & 0.067 \\
$\Psi(5)$ & 0.182 & 0.123 & 0.117 & 0.183 & 0.139 & 0.126 & 0.123 & 0.119 & 0.083 \\
$\bar\Psi$ & 0.102 & 0.060 & 0.053 & 0.103 & 0.070 & 0.055 & 0.059 & 0.053 & 0.047 \\
\bottomrule
\end{tabular}
\end{table}

\paragraph{The estimator of the paper and its predecessors.}
Table~\ref{tab:supp_estimators} gives, for every number of clusters in
Block~I, the bound of $\mathcal M_S$, the standard deviation, bias and
SE/SD ratios (sandwich, HC3 and jackknife) of the estimator of the paper,
and the standard deviations of the same estimator without the history term,
of the estimator of the earlier versions with the period-1 means of $S$ and
$Y$ as covariates, and of the complete-case exposure-time mixed model fitted
to all periods and to periods 1 to 6. With 30 and 48 clusters the estimator
of the paper is within 8\% of the bound for every target; with 18 it is 6\%
to 16\% above it and with 12 it is 25\% to 36\% above it, the loss coming
from the estimated loading and variance components. Its bias never exceeds
0.010. Against the earlier estimator it lowers the standard deviations by
14\% to 23\% for $\Psi(1)$ to $\Psi(4)$ and $\bar\Psi$ at every $J$, and by
3\% to 10\% for $\Psi(5)$. The sandwich understates the standard deviation
by 5\% to 10\% with 12 clusters and is accurate from 18 clusters on; HC3 is
1\% to 20\% above it and the jackknife 1\% to 54\% above it, both most
conservative with 12 clusters. The mixed model fitted to all periods
has standard deviations 5\% to 7\% below those of the estimator of the paper
for $\Psi(1)$ to $\Psi(3)$ with 30 clusters, 13\% and 15\% below for $\Psi(4)$
and $\bar\Psi$, and 34\% below for $\Psi(5)$; its advantage grows as the
number of clusters falls, to 22\% to 43\% with 12 clusters.

\begin{table}[!htbp]
\centering
\caption{Block~I: the estimator of the paper (record-level regressions and
the history index) against the bound of $\mathcal M_S$, the same estimator
without the history term (no index), the estimator of the earlier versions
with the period-1 means of $S$ and $Y$ as covariates (P1 means), and the
complete-case exposure-time mixed model fitted to all periods (LMM) and to
periods 1 to 6 (LMM, no final). SE/SD for the sandwich (s), HC3 (h) and
jackknife (j). Replications: 500 ($J=12$), 1000 ($J=18$, 30), 250 ($J=48$).}
\label{tab:supp_estimators}
\footnotesize
\setlength{\tabcolsep}{3.3pt}
\begin{tabular}{rlcccccccccc}
\toprule
 & & & \multicolumn{5}{c}{Estimator of the paper} & \multicolumn{4}{c}{Standard deviation of}\\
\cmidrule(lr){4-8}\cmidrule(lr){9-12}
$J$ & Target & Bound & SD & Bias & s & h & j & no index & P1 means & LMM & LMM, no final\\
\midrule
12 & $\Psi(1)$ & 0.061 & 0.076 & $-$0.002 & 0.92 & 1.13 & 1.37 & 0.112 & 0.092 & 0.059 & 0.060 \\
 & $\Psi(2)$ & 0.073 & 0.095 & $-$0.006 & 0.91 & 1.11 & 1.35 & 0.145 & 0.119 & 0.072 & 0.074 \\
 & $\Psi(3)$ & 0.091 & 0.120 & $-$0.010 & 0.90 & 1.10 & 1.36 & 0.178 & 0.149 & 0.089 & 0.092 \\
 & $\Psi(4)$ & 0.117 & 0.150 & $-$0.004 & 0.95 & 1.20 & 1.50 & 0.207 & 0.178 & 0.107 & 0.114 \\
 & $\Psi(5)$ & 0.185 & 0.246 & 0.002 & 0.92 & 1.16 & 1.54 & 0.286 & 0.253 & 0.141 & 0.164 \\
 & $\bar\Psi$ & 0.083 & 0.113 & $-$0.004 & 0.90 & 1.14 & 1.48 & 0.167 & 0.139 & 0.078 & 0.080 \\
\addlinespace
18 & $\Psi(1)$ & 0.050 & 0.053 & $-$0.001 & 0.99 & 1.14 & 1.18 & 0.089 & 0.068 & 0.048 & 0.049 \\
 & $\Psi(2)$ & 0.060 & 0.065 & $-$0.002 & 0.98 & 1.13 & 1.18 & 0.113 & 0.084 & 0.061 & 0.062 \\
 & $\Psi(3)$ & 0.074 & 0.081 & 0.000 & 0.98 & 1.13 & 1.19 & 0.140 & 0.103 & 0.071 & 0.075 \\
 & $\Psi(4)$ & 0.096 & 0.107 & $-$0.004 & 0.99 & 1.15 & 1.24 & 0.171 & 0.129 & 0.090 & 0.096 \\
 & $\Psi(5)$ & 0.151 & 0.175 & $-$0.007 & 0.98 & 1.16 & 1.37 & 0.243 & 0.187 & 0.114 & 0.132 \\
 & $\bar\Psi$ & 0.068 & 0.077 & $-$0.003 & 0.97 & 1.16 & 1.27 & 0.135 & 0.098 & 0.064 & 0.065 \\
\addlinespace
30 & $\Psi(1)$ & 0.039 & 0.038 & $-$0.001 & 1.05 & 1.14 & 1.14 & 0.066 & 0.047 & 0.036 & 0.037 \\
 & $\Psi(2)$ & 0.046 & 0.048 & 0.001 & 0.98 & 1.07 & 1.07 & 0.086 & 0.060 & 0.045 & 0.047 \\
 & $\Psi(3)$ & 0.057 & 0.059 & $-$0.002 & 0.99 & 1.08 & 1.08 & 0.110 & 0.077 & 0.055 & 0.058 \\
 & $\Psi(4)$ & 0.074 & 0.077 & $-$0.001 & 1.00 & 1.10 & 1.12 & 0.130 & 0.092 & 0.067 & 0.071 \\
 & $\Psi(5)$ & 0.117 & 0.126 & $-$0.001 & 0.98 & 1.10 & 1.16 & 0.183 & 0.139 & 0.083 & 0.101 \\
 & $\bar\Psi$ & 0.052 & 0.055 & $-$0.001 & 0.99 & 1.11 & 1.14 & 0.103 & 0.070 & 0.047 & 0.049 \\
\addlinespace
48 & $\Psi(1)$ & 0.031 & 0.030 & 0.001 & 1.04 & 1.09 & 1.09 & 0.054 & 0.035 & 0.028 & 0.029 \\
 & $\Psi(2)$ & 0.037 & 0.039 & 0.002 & 0.97 & 1.01 & 1.01 & 0.074 & 0.047 & 0.036 & 0.037 \\
 & $\Psi(3)$ & 0.045 & 0.048 & 0.002 & 0.98 & 1.03 & 1.03 & 0.088 & 0.059 & 0.045 & 0.047 \\
 & $\Psi(4)$ & 0.059 & 0.061 & 0.001 & 0.98 & 1.04 & 1.05 & 0.109 & 0.072 & 0.053 & 0.057 \\
 & $\Psi(5)$ & 0.092 & 0.094 & 0.010 & 1.03 & 1.11 & 1.15 & 0.148 & 0.104 & 0.067 & 0.080 \\
 & $\bar\Psi$ & 0.042 & 0.044 & 0.003 & 0.99 & 1.05 & 1.07 & 0.085 & 0.053 & 0.037 & 0.039 \\
\bottomrule
\end{tabular}
\end{table}

\paragraph{Where the mixed model's precision comes from.} Under the generator
the marginal model for $Y$ given $(E,W,t)$ is a Gaussian random-intercept
model (the cluster effect enters directly and through the surrogate, so its
variance is $1.69\sigma_b^2$, and the residual variance is $0.0625+0.64$), so
the complete-case mixed model with exposure-time indicators is correctly
specified and its variance is $(X^\top V^{-1}X)^{-1}$, which depends on the
data only through the realized design. Table~\ref{tab:lmm_exact} gives the
resulting standard errors, averaged over twenty generated designs (their
spread is below $10^{-3}$), for the model fitted to all periods and to
periods $1,\ldots,T-1$, next to the bound of $\mathcal M_S$; the ratios are the
same for every $J$. Without the final period, in which every cluster is
exposed, the mixed model's standard errors equal the bound for $\Psi(1)$ to
$\Psi(3)$ to within 1\%, and lie 2\% below it for $\Psi(4)$, 12\% for
$\Psi(5)$ and 5\% for $\bar\Psi$. The final period lowers them by a further
1.5\%, 2.3\%, 3.5\%, 6\%, 15\% and 3\%. For the longest exposure time the
mixed model's advantage of 25\% in standard deviation therefore splits into
about 12\% from its assumption that effects do not change with calendar
time, which lets the clusters at other exposure levels in period 6 inform
the period effect against which the single exposure-5 wave is compared, and
about 15\% from the extrapolation to the final period. Two forces offset
each other for the shorter exposure times: the complete-case mixed model
discards the surrogates of censored records, which the estimators of
$\mathcal M_S$ use, while its additive structure pools information across
exposure levels; that the two balance to within 1\% is a feature of this
design, not a general fact. The mixed model's model-based standard errors in
Table~\ref{tab:supp_main_full} average these exact values, and its Monte
Carlo standard deviations there lie 0\% to 5\% below them, within about two
simulation standard errors.

\begin{table}[!htbp]
\centering
\caption{Block~I design: exact standard errors of the complete-case
exposure-time mixed model fitted to all periods and to periods $1$ to $T-1$,
and the bound of $\mathcal M_S$, all for $J=30$ (they scale with $J^{-1/2}$).}
\label{tab:lmm_exact}
\small
\begin{tabular}{lcccccc}
\toprule
 & $\Psi(1)$ & $\Psi(2)$ & $\Psi(3)$ & $\Psi(4)$ & $\Psi(5)$ & $\bar\Psi$\\
\midrule
Mixed model, all periods & 0.038 & 0.046 & 0.056 & 0.068 & 0.087 & 0.049\\
Mixed model, without the final period & 0.039 & 0.047 & 0.058 & 0.073 & 0.102 & 0.050\\
Bound of $\mathcal M_S$ & 0.039 & 0.046 & 0.057 & 0.074 & 0.117 & 0.053\\
\addlinespace
Ratio, all periods to without & 0.985 & 0.977 & 0.965 & 0.938 & 0.852 & 0.970\\
Ratio, without the final period to bound & 1.000 & 1.011 & 1.005 & 0.981 & 0.875 & 0.954\\
\bottomrule
\end{tabular}
\end{table}

\section{The EPT-Calibrated Simulation and the Replication Data}
\label{app:ept}

\paragraph{The replication data set.} The file
\nolinkurl{inst/extdata/data_ept.csv} is a synthetic data set in the design of
the EPT trial \citep{golden2015uptake}: 23 LHJs in four waves of 6, 6, 6 and 5
crossing over at periods 2 to 5, with 162 records per LHJ-period, the sample
size assumed in the trial's planning, and a positivity near the 5\% assumed
there. Each record carries a 30-day uptake indicator, a 12-month positivity
indicator and an observation indicator; none of these comes from the trial,
whose primary outcomes were community-level chlamydia positivity among women
tested at sentinel clinics and reported gonorrhea incidence.

\paragraph{Design and censoring.} Table~\ref{tab:ept_design} gives the design
law of exposure time for the allocation $(6,6,6,5)$ and the proportion of
censored records by wave and period in the replication data. The overlap
periods are $\mathbb{T}_1=\{2,3,4\}$, $\mathbb{T}_2=\{3,4\}$ and
$\mathbb{T}_3=\{4\}$; period~1 has no exposed LHJs and period~5 no unexposed
LHJs. Censoring is a property of the LHJ's wave rather than of its exposure:
the last-wave LHJs are 84--87\% censored in every period, including the
periods before they cross over, and the third-wave LHJs are 36--39\% censored
in periods 4 and 5. Under the structural model this does not affect
Assumption~\ref{M-ass:mar}, because the crossover time is independent of the
potential outcomes given exposure time and period, but it is why the
observation model includes period-by-wave indicators. The only unexposed LHJs
in period 4, and hence the only controls for $\Psi(3)$, are the five of the
last wave.

\begin{table}[ht]
\centering
\caption{EPT design: exposure law $g_E(k\mid t)$ under the allocation
$(6,6,6,5)$ (top) and proportion of records with the delayed outcome
censored in the replication data, by crossover wave and period (bottom).}
\label{tab:ept_design}
\small
\begin{tabular}{lccccc}
\toprule
 & $t=1$ & $t=2$ & $t=3$ & $t=4$ & $t=5$\\
\midrule
$g_E(0\mid t)$ & 1.000 & 0.739 & 0.478 & 0.217 & 0.000 \\
$g_E(1\mid t)$ & 0.000 & 0.261 & 0.261 & 0.261 & 0.217 \\
$g_E(2\mid t)$ & 0.000 & 0.000 & 0.261 & 0.261 & 0.261 \\
$g_E(3\mid t)$ & 0.000 & 0.000 & 0.000 & 0.261 & 0.261 \\
$g_E(4\mid t)$ & 0.000 & 0.000 & 0.000 & 0.000 & 0.261 \\
\midrule
Wave 1 (crossover $t=2$) & 0.00 & 0.00 & 0.00 & 0.00 & 0.00 \\
Wave 2 (crossover $t=3$) & 0.00 & 0.00 & 0.00 & 0.00 & 0.00 \\
Wave 3 (crossover $t=4$) & 0.00 & 0.00 & 0.00 & 0.36 & 0.39 \\
Wave 4 (crossover $t=5$) & 0.85 & 0.87 & 0.84 & 0.84 & 0.85 \\
\bottomrule
\end{tabular}
\end{table}

\paragraph{Generator.} Each simulated trial assigns 23 LHJs to crossover
periods by a random permutation of the allocation and generates 162 records
per LHJ-period. With $u_j\sim\mathcal N(0,0.25^2)$ and
$v_j\sim\mathcal N(0,0.20^2)$ independent,
\begin{align*}
  \logit\PP(S=1)&=-3.08+\theta_S(E)+\gamma_S(t)+u_j,\\
  \logit\PP(Y=1)&=-2.88-0.41\,S+\beta(E)+\gamma_Y(t)+v_j,
\end{align*}
with $\theta_S=(0.76,0.71,0.64,0.66)$ for $E=1,\ldots,4$,
$\gamma_S=(0.27,0.34,0.34,0.35)$ and $\gamma_Y=(-0.11,-0.16,-0.16,-0.26)$ for
$t=2,\ldots,5$, and $\theta_S(0)=\gamma_S(1)=\gamma_Y(1)=0$. These are
population-averaged logistic fits to the replication data; the LHJ standard
deviations are assumed. The effect of exposure beyond uptake is $\beta(E)=0$
in scenario \emph{none} and $\beta=(-0.10,-0.25,-0.40,-0.40)$ in scenario
\emph{growing}; fitted to the replication data it would be
$(-0.14,-0.03,-0.27,0.13)$, which is used only in the check below. In the
replication data, uptake within 30 days was 5.3\% among unexposed and 11.6\%
among exposed records, and 12-month positivity among unexposed records was
4.9\%, with an odds ratio of 0.67 for uptake. Under scenario \emph{none}
every $\Psi(e)$ is about $-0.08$ percentage points; under \emph{growing},
$\Psi(e)=(-0.51,-1.04,-1.50)$ percentage points and $\bar\Psi=-1.02$. A record of
wave $s$ in period $t$ is observed with probability one minus the censored
proportion in Table~\ref{tab:ept_design}, shifted by $-0.05$ on the logit
scale per unit of centred uptake. The true values of $\Psi(e)$ average
$\EE\{Y(e)-Y(0)\}$ over the overlap periods with equal weights, integrating
over uptake exactly and over the LHJ effects by Monte Carlo. Table~\ref{tab:supp_eptcal} compares 200
generated trials, with $\beta$ set to its fitted value, with the replication
data. The generator reproduces uptake, positivity and censoring; the spread
of LHJ-level positivity is larger by design, because the replication data
have almost none beyond binomial noise.

\begin{table}[ht]
\centering
\caption{Calibration check: percentages in the replication data and averaged
over 200 trials from the generator, with the effect of exposure beyond uptake
set to its value fitted to the replication data.}
\label{tab:supp_eptcal}
\small
\begin{tabular}{lcc}
\toprule
 & Replication data & Generator\\
\midrule
Uptake, unexposed records & 5.3 & 5.5 \\
Uptake, exposed records & 11.6 & 11.9 \\
Positivity, unexposed records & 4.9 & 5.0 \\
Positivity, exposed records & 3.9 & 4.0 \\
Positivity, with uptake & 2.9 & 2.9 \\
Positivity, without uptake & 4.5 & 4.6 \\
Censored, wave 3, periods 4--5 & 37.8 & 37.7 \\
Censored, wave 4, all periods & 85.2 & 85.2 \\
SD of LHJ positivity & 0.8 & 1.1 \\
\bottomrule
\end{tabular}
\end{table}

\paragraph{Analysis.} SA-AIPW used a logistic outcome regression on $S$,
exposure-time indicators and period indicators; a logistic observation model
on $S$ and period-by-wave indicators; a linear bridge regression, pooled over
exposure levels, on exposure-time indicators and period indicators; and the
history index of Proposition~\ref{M-prop:history}(iii) from the uptake and
positivity histories, with residual regressions of $S$ on period indicators
and of $Y$ on $S$ and period indicators and with fixed weights, since the
observation probabilities depend on the wave. The estimator of the earlier
versions, with the LHJ's period-1 uptake proportion as a covariate and no
history term, is reported alongside. The LHJ heterogeneity calibrated from
the replication data is small (Table~\ref{tab:supp_eptcal}), so the index has
little to recover here; the point of including it is that it is the same
estimator as in Section~\ref{M-sec:sim:design} and that its validity does
not depend on the LHJ structure. The $V=5$ folds follow the original analysis rather than
the default rule of Section~\ref{app:impl:notes}. The jackknife takes about ten
times as long as a single fit, which is why it was computed for 500 trials
only.

\paragraph{Results.} Table~\ref{tab:supp_eptsim} gives the full results.
Biases were at most 0.08 percentage points. The ratio of the mean standard
error to the Monte Carlo SD was between 0.87 and 0.99 for the sandwich,
between 0.97 and 1.10 for HC3 and between 1.11 and 1.17 for the jackknife.
The sandwich of the index-adjusted estimating functions ignores the
estimation of the loading, and with 23 clusters and an index that carries
almost no information here (the estimated between-LHJ variance components
are of order $10^{-4}$ on the probability scale) that loading is poorly
determined, so the sandwich is anticonservative and HC3 and the jackknife,
which account for it, are not; the estimator of the earlier versions, with
no loading to fit, had sandwich ratios of 0.95 to 1.06 in the previous
version of this study, at the same standard deviations. Under the trial's
allocation the coverage for $\Psi(3)$ and $\bar\Psi$ was 0.88 to 0.90 with
the sandwich and 0.92 to 0.94 with HC3. As explained in
Section~\ref{M-sec:sim:ept}, the estimates and their standard errors are
negatively correlated ($-0.46$ for $\hat\Psi(3)$ in the growing scenario),
so that in 8.0\% of trials the sandwich interval for $\Psi(3)$ lay entirely
above the truth and in 3.0\% entirely below it (6.1\% and 2.4\% for HC3).
On the 500 trials with the jackknife the coverage for $\Psi(3)$ was 0.878
with the sandwich and 0.938 with the jackknife. Under the allocation
$(7,2,2,12)$ the coverage was 0.93 to 0.95 with the sandwich and 0.95 to
0.97 with HC3, and the full-data standard deviation of $\hat\Psi(3)$ was also
smaller (0.77 against 1.06 percentage points), so part of the gain from that
allocation would remain without censoring; the censoring cost, the ratio in
the table, remained above 3 for $\Psi(3)$.

\paragraph{Design pilot.} The inputs of Proposition~\ref{M-prop:design} were
estimated from one trial of 230 LHJs drawn from the generator in the growing
scenario, with the trial's allocation proportions: $\hat\sigma_b^2=0.0004$,
$\hat\sigma_\epsilon^2=0.039$ and $\hat\sigma_r^2=0.039$, with censoring odds
of 5.4 to 6.4 for the last wave in every period and about 0.6 for the third
wave in periods 4 and 5. Keeping at least two LHJs per wave, the allocation
minimizing the predicted variance of $\hat\Psi(3)$ is $(7,2,2,12)$ and the one
minimizing that of $\hat{\bar\Psi}$ is $(6,3,2,12)$.

\begin{table}[ht]
\centering
\caption{Simulation calibrated to the EPT trial, 1000 trials per setting:
truth, bias and Monte Carlo SD in percentage points; ratio of mean standard
error to SD and coverage for the sandwich, HC3 and the jackknife (first 500
trials of the growing scenario under the trial allocation); SD of the same
estimator applied to the uncensored outcomes and the ratio of the two
variances; and the SD of the estimator of the earlier versions (period-1
uptake proportion as a covariate, no history term).}
\label{tab:supp_eptsim}
\scriptsize
\setlength{\tabcolsep}{2.6pt}
\begin{tabular}{llrrrccccccccc}
\toprule
 & & & & & \multicolumn{2}{c}{Sandwich} & \multicolumn{2}{c}{HC3} & \multicolumn{2}{c}{Jackknife} & \multicolumn{2}{c}{Full data} & Earlier\\
\cmidrule(lr){6-7}\cmidrule(lr){8-9}\cmidrule(lr){10-11}\cmidrule(lr){12-13}\cmidrule(lr){14-14}
Setting & Target & Truth & Bias & SD & SE/SD & Cov. & SE/SD & Cov. & SE/SD & Cov. & SD & Ratio & SD\\
\midrule
None, (6, 6, 6, 5) & $\Psi(1)$ & $-$0.09 & $-$0.01 & 0.88 & 0.95 & 0.927 & 1.05 & 0.949 & -- & -- & 0.57 & 2.40 & 0.89 \\
 & $\Psi(2)$ & $-$0.08 & 0.01 & 1.27 & 0.89 & 0.914 & 0.99 & 0.937 & -- & -- & 0.76 & 2.78 & 1.26 \\
 & $\Psi(3)$ & $-$0.07 & 0.04 & 2.11 & 0.89 & 0.903 & 0.99 & 0.935 & -- & -- & 1.16 & 3.32 & 2.11 \\
 & $\bar\Psi$ & $-$0.08 & 0.01 & 1.32 & 0.87 & 0.899 & 0.97 & 0.918 & -- & -- & 0.69 & 3.66 & 1.31 \\
\addlinespace
Growing, (6, 6, 6, 5) & $\Psi(1)$ & $-$0.51 & 0.01 & 0.89 & 0.94 & 0.932 & 1.04 & 0.949 & 1.15 & 0.964 & 0.55 & 2.59 & 0.90 \\
 & $\Psi(2)$ & $-$1.04 & 0.02 & 1.19 & 0.94 & 0.931 & 1.05 & 0.950 & 1.17 & 0.970 & 0.70 & 2.86 & 1.18 \\
 & $\Psi(3)$ & $-$1.50 & $-$0.02 & 2.09 & 0.89 & 0.890 & 0.99 & 0.915 & 1.11 & 0.938 & 1.06 & 3.85 & 2.06 \\
 & $\bar\Psi$ & $-$1.02 & 0.00 & 1.29 & 0.89 & 0.884 & 1.00 & 0.916 & 1.13 & 0.944 & 0.64 & 4.03 & 1.28 \\
\addlinespace
Growing, (7, 2, 2, 12) & $\Psi(1)$ & $-$0.51 & 0.01 & 0.99 & 0.99 & 0.951 & 1.10 & 0.968 & -- & -- & 0.69 & 2.02 & 0.99 \\
 & $\Psi(2)$ & $-$1.04 & 0.04 & 1.11 & 0.97 & 0.932 & 1.08 & 0.955 & -- & -- & 0.76 & 2.13 & 1.14 \\
 & $\Psi(3)$ & $-$1.50 & 0.08 & 1.40 & 0.95 & 0.930 & 1.04 & 0.947 & -- & -- & 0.77 & 3.32 & 1.38 \\
 & $\bar\Psi$ & $-$1.02 & 0.04 & 0.97 & 0.94 & 0.932 & 1.03 & 0.947 & -- & -- & 0.54 & 3.20 & 0.98 \\
\bottomrule
\end{tabular}
\end{table}

\paragraph{Intervals for binary outcomes.} Table~\ref{tab:supp_eptmover}
compares the sandwich interval~\eqref{M-eq:ci_t} with the logit-scale
interval~\eqref{M-eq:mover} and its arcsine square-root analogue, computed on
the same 1000 trials per setting from the index-adjusted arm means
(\texttt{arm\_means\_index}, Section~\ref{app:impl:notes}) with their
sandwich covariance, and with their HC3 covariance for the logit version.
The logit version was chosen in advance as the primary one. With the
sandwich covariance of the arm means it covered 0.92 to 0.95, inheriting the
anticonservativeness of the sandwich for the index-adjusted estimator with
23 clusters; with the HC3 covariance it covered 0.946 to 0.969 across targets
and settings. For $\Psi(3)$ and $\bar\Psi$ under the trial's allocation the
logit transformation moved the misses to the other side: with the sandwich
covariance, 1\% to 3\% of its intervals lay entirely above the truth and
4\% to 6\% entirely below it, where the symmetric interval for $\Psi(3)$ in
the growing scenario missed above in 8.0\% of trials and below in 3.0\%. The
price is intervals 11\% to 22\% wider than the sandwich interval with the
sandwich covariance and 37\% to 55\% wider with the HC3 covariance; in the
other cases the widening was 2\% to 6\% and 18\% to 23\%. The arcsine
version widened the intervals less but covered $\Psi(3)$ and $\bar\Psi$ only
0.90 to 0.94 of the time. The estimates in these runs agree exactly with
those of Table~\ref{tab:supp_eptsim}.

\begin{table}[!htbp]
\centering
\caption{Coverage of 95\% intervals in the EPT-calibrated simulation: sandwich
(s); logit-scale MOVER interval (logit) from the sandwich covariance of the
arm means, with the proportions of intervals lying entirely above and
entirely below the truth and its mean width relative to the sandwich
interval; the same interval from the HC3 covariance of the arm means (logit,
HC3); and the arcsine square-root version (asin) with its relative width.
1000 trials per setting.}
\label{tab:supp_eptmover}
\footnotesize
\setlength{\tabcolsep}{3pt}
\begin{tabular}{llccccccccc}
\toprule
 & & & \multicolumn{4}{c}{Logit} & \multicolumn{2}{c}{Logit, HC3} & \multicolumn{2}{c}{Asin}\\
\cmidrule(lr){4-7}\cmidrule(lr){8-9}\cmidrule(lr){10-11}
Setting & Target & Cov.\ (s) & Cov. & Above & Below & Width & Cov. & Width & Cov. & Width\\
\midrule
None, (6, 6, 6, 5) & $\Psi(1)$ & 0.927 & 0.928 & 0.044 & 0.028 & 1.02 & 0.952 & 1.20 & 0.930 & 1.00 \\
 & $\Psi(2)$ & 0.914 & 0.922 & 0.043 & 0.035 & 1.04 & 0.946 & 1.22 & 0.917 & 1.00 \\
 & $\Psi(3)$ & 0.903 & 0.936 & 0.022 & 0.042 & 1.22 & 0.958 & 1.55 & 0.925 & 1.02 \\
 & $\bar\Psi$ & 0.899 & 0.927 & 0.031 & 0.042 & 1.13 & 0.957 & 1.40 & 0.914 & 1.01 \\
\addlinespace
Growing, (6, 6, 6, 5) & $\Psi(1)$ & 0.932 & 0.931 & 0.042 & 0.027 & 1.02 & 0.961 & 1.20 & 0.931 & 1.00 \\
 & $\Psi(2)$ & 0.931 & 0.949 & 0.029 & 0.022 & 1.04 & 0.965 & 1.23 & 0.940 & 1.00 \\
 & $\Psi(3)$ & 0.890 & 0.928 & 0.014 & 0.058 & 1.19 & 0.953 & 1.50 & 0.907 & 1.00 \\
 & $\bar\Psi$ & 0.884 & 0.929 & 0.020 & 0.051 & 1.11 & 0.962 & 1.37 & 0.902 & 1.00 \\
\addlinespace
Growing, (7, 2, 2, 12) & $\Psi(1)$ & 0.951 & 0.950 & 0.031 & 0.019 & 1.03 & 0.969 & 1.18 & 0.950 & 1.00 \\
 & $\Psi(2)$ & 0.932 & 0.940 & 0.037 & 0.023 & 1.05 & 0.962 & 1.23 & 0.933 & 1.00 \\
 & $\Psi(3)$ & 0.930 & 0.947 & 0.027 & 0.026 & 1.06 & 0.964 & 1.19 & 0.938 & 1.00 \\
 & $\bar\Psi$ & 0.932 & 0.938 & 0.033 & 0.029 & 1.04 & 0.962 & 1.18 & 0.928 & 1.00 \\
\bottomrule
\end{tabular}
\end{table}

\paragraph{Analysis of the replication data set.} Table~\ref{tab:ept_example} analyses
the replication data set itself, which can be read as one draw from a design
like the one above; because the outcome is recorded for every record, the
same estimator can also be applied to the complete data. The full-data
benchmark shows no exposure-time effect (all estimates within 0.7 percentage
points of zero; heterogeneity test $p=0.71$). SA-AIPW on the censored data
agrees, with every interval containing both zero and the benchmark and a
heterogeneity test of $F=0.24$ on $(2,21)$ degrees of freedom ($p=0.79$); its
estimates lie within 0.5 percentage points of the benchmark. Without the
history term the estimates lie 0.7 to 0.8 points above it, a shift common to
all targets because all of them rely on the period-4 controls of the last
wave: of their 810 records, 126 have the outcome observed, and positivity is
2.4\% among these and 4.1\% among the censored ones. The history term moves
the estimates toward the benchmark here, but the three standard errors of
SA-AIPW disagree about its cost: the sandwich (0.79 to 1.44 points) is
smaller than that of the estimator without the term, HC3 is 4\% to 15\%
larger than the sandwich and the jackknife 27\% to 59\% larger. The reason is
the one seen in the simulation above: the estimated between-LHJ variance
components are essentially zero, so the index carries almost no information
in these data and its loading, fitted on 23 LHJs, is poorly determined; the
in-sample sandwich takes the fitted loading at face value, while the
jackknife refits it. With real cluster heterogeneity, as in Block~I, the
three agree. We read the jackknife as the honest standard error here, and
would report it, together with the estimator without the term, in a trial
of this size whose estimated variance components are this small. Since the
replication data contain the outcome for every record, the missing-at-random
assumption can be checked, which is not possible in practice; adding the true
outcome to a logistic regression of the observation indicator on the
surrogate and period-by-wave indicators gives a likelihood-ratio $p$-value of
0.27. The SA-AIPW jackknife standard errors are 1.29, 1.75 and 1.58 times
those of the benchmark for $e=1,2,3$. The logit-scale
intervals~\eqref{M-eq:mover} with the HC3 covariance of the arm means for
$\Psi(1)$, $\Psi(2)$, $\Psi(3)$ and $\bar\Psi$ are $(-2.02,1.78)$,
$(-2.36,2.78)$, $(-5.72,3.96)$ and $(-3.11,2.19)$ percentage points, against
HC3 intervals $(-1.71,1.72)$, $(-1.54,3.02)$, $(-2.89,4.01)$ and
$(-1.66,2.58)$. They reach further below the estimate, most of all for
$\Psi(3)$, whose control mean of 2.9\% rests on the three observed positive
outcomes. The unaugmented estimator and IPCW have standard errors 23\% to 66\%
larger than those of SA-AIPW without the history term, partly because
i.i.d.-based variance estimates for them are conservative under fixed waves
(Proposition~\ref{M-prop:alloc}). The complete-case mixed models give
estimates 0.6 to 1.8 percentage points below the benchmark, with model-based
standard errors smaller than those of SA-AIPW.

\begin{table}[ht]
\centering
\caption{Analysis of the replication data set: estimates (standard errors) on
the risk-difference scale, in percentage points. Standard errors are
jackknife unless stated, sandwich for SA-AIPW without the history term, the
unaugmented estimator and IPCW, and model-based for the mixed models (LMM).
ETI: exposure-time indicators; IT: immediate treatment. The full-data
benchmark applies SA-AIPW to the uncensored outcomes.}
\label{tab:ept_example}
\footnotesize
\setlength{\tabcolsep}{3pt}
\begin{tabular}{lccccc}
\toprule
 & $\Psi(1)$ & $\Psi(2)$ & $\Psi(3)$ & $\bar\Psi$ & $\PsiIT$\\
\midrule
Full-data benchmark & $-0.51$ (0.78) & $0.26$ (0.90) & $0.70$ (1.31) & $0.10$ (0.77) & $-0.53$ (0.65) \\
SA-AIPW, jackknife SE & $0.00$ (1.00) & $0.74$ (1.57) & $0.56$ (2.08) & $0.46$ (1.40) & $0.10$ (0.88) \\
SA-AIPW, HC3 SE & $0.00$ (0.82) & $0.74$ (1.10) & $0.56$ (1.66) & $0.46$ (1.02) & $0.10$ (0.75) \\
SA-AIPW, sandwich SE & $0.00$ (0.79) & $0.74$ (0.99) & $0.56$ (1.44) & $0.46$ (0.89) & $0.10$ (0.71) \\
SA-AIPW, no history term & $0.21$ (0.88) & $1.04$ (1.16) & $1.41$ (1.95) & $0.89$ (1.16) & $0.32$ (0.80) \\
Unaugmented SA-AIPW & $0.22$ (1.08) & $1.04$ (1.91) & $1.24$ (2.40) & $0.83$ (1.65) & $0.33$ (1.38) \\
IPCW & $0.20$ (1.09) & $1.04$ (1.92) & $1.20$ (2.51) & $0.81$ (1.69) & $0.32$ (1.39) \\
Complete-case LMM, ETI & $-1.14$ (0.57) & $-0.32$ (0.71) & $-1.12$ (0.89) & $-0.86$ (0.62) & -- \\
Complete-case LMM, IT & -- & -- & -- & -- & $-0.97$ (0.55) \\
\bottomrule
\end{tabular}
\end{table}

\paragraph{Sensitivity to Assumption~\ref{M-ass:mar}.} For constant departures
$\gamma_k(S,W,t)\equiv\gamma_k$, the bias~\eqref{eq:bias_gamma} is
$\Lambda_e(\gamma)=s_e(e)\gamma_e-s_e(0)\gamma_0$ with
$s_e(k)=\EE\{\inde I(E=k)(1-\Delta)/g_E(k\mid t)\}/\pi_e$, which can be
estimated by the corresponding empirical mean because
$\EE(1-\Delta\mid S,E,W,t)=1-g_\Delta$. In the replication data
(Table~\ref{tab:ept_sens}), if censored control records had positivity
$\gamma_0$ higher than observed control records with the same surrogate value,
the estimate of $\Psi(e)$ would have to be lowered by $s_e(0)\gamma_0$: by
about half of $\gamma_0$ for $e=1$ and by $0.84\gamma_0$ for $e=3$. Departures
among exposed records matter only for $e=1$, through the third-wave LHJs in
period 4. The sensitivity grows with exposure time, as the coupling of
Section~\ref{M-sec:burden} suggests.

\begin{table}[ht]
\centering
\caption{Sensitivity coefficients in the replication data: the change in the
limit of $\hat\Psi(e)$ per unit departure from Assumption~\ref{M-ass:mar}
among exposed ($s_e(e)$) and unexposed ($s_e(0)$) records.}
\label{tab:ept_sens}
\small
\begin{tabular}{lccc}
\toprule
 & $e=1$ & $e=2$ & $e=3$\\
\midrule
$s_e(e)$ & 0.12 & 0.00 & 0.00 \\
$s_e(0)$ & 0.49 & 0.61 & 0.84 \\
\bottomrule
\end{tabular}
\end{table}

\paragraph{Reproduction.} The simulation is produced by
\nolinkurl{inst/sim/run_ept_sim.R}, its tables by
\nolinkurl{inst/sim/make_eptsim_tables.R} and the calibration check by
\nolinkurl{inst/sim/ept_calibration_check.R}; the analysis of the replication
data set is \nolinkurl{inst/ept/run_ept_example.R}. All read
\nolinkurl{inst/extdata/data_ept.csv}.

\bibliographystyle{apalike}
\bibliography{references_2}